\newif\ifFull \Fulltrue

\documentclass[runningheads,orivec,envcountsect,envcountsame]{llncs}
\usepackage[T1]{fontenc}
\usepackage{graphicx}
\usepackage[square,numbers,sectionbib]{natbib}

\usepackage{amsmath,amssymb}
\usepackage[cmtip,all]{xy}
\usepackage{stmaryrd}
\usepackage{mathtools}

\usepackage{stackengine}
\usepackage{mathrsfs}
\usepackage{braket}
\usepackage{annotate-equations}
\usepackage{scalerel}
\usepackage{todonotes}

\usepackage{adjustbox}

\usepackage{xspace}

\usepackage{tikz}
\usetikzlibrary{automata, positioning, arrows.meta, fit, shadows, decorations.pathmorphing}
\tikzset{every node/.style={node font=\tiny, node distance=5mm and 1cm}}
\definecolor{darkgreen}{rgb}{0.0, 0.5, 0.0}
\tikzstyle{smallstate}=[inner sep=1.2pt,draw,circle,fill]
\tikzstyle{textstate}=[node font=\small]
\tikzstyle{flowchartnode}=[node distance=1.75cm, node font=\footnotesize, inner sep=5pt, draw=gray, dashed, minimum height=0.75cm,minimum width=1.25cm]
\tikzstyle{output-edge}=[-Implies, double distance=1pt]
\tikzstyle{transition-edge}=[-{Stealth[length=2mm,width=1.5mm]},shorten >=0.8mm, shorten <=1.2mm]
\tikzstyle{output}=[draw=none, node distance=.6cm, inner sep=0pt,yshift=1pt]

\usepackage{tikz-cd}

\usepackage{hyperref}
\usepackage{cleveref}
\crefname{apx}{appendix}{appendices}
\Crefname{apx}{Appendix}{Appendices}

\makeatletter
\def\@acmplainindent{0pt}
\def\@acmdefinitionindent{0pt}
\def\@proofindent{\noindent}
\makeatother

\newcommand\myparagraph[1]{%
  \par\addvspace{6pt plus 2pt minus 1pt}%
  \noindent\textbf{\textit{#1}.}\hspace{.5em plus 0.1em minus 0.1em}%
}

\def\orcidID#1{\href{http://orcid.org/#1}{\protect\raisebox{-1.25pt}{\protect\includegraphics{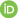}}}}
 
\usepackage{mathpartir}
\mprset{fractionaboveskip=0, fractionbelowskip=0}
\usepackage{enumitem}

\usepackage{listings}
\crefname{lstlisting}{listing}{listings}
\Crefname{lstlisting}{Listing}{Listings}

\usepackage{algorithm}
\usepackage[noend]{algpseudocode}
\algnewcommand\algorithmicmatch{\textbf{match}}
\algnewcommand\algorithmicwith{\textbf{with}}
\algnewcommand\algorithmiccase{\textbf{case}}
\algnewcommand\algorithmicdefault{\textbf{default}}
\algnewcommand\Continue{\textbf{continue}}
\algdef{SE}[MATCH]{Match}{EndMatch}[1]{\algorithmicmatch\ #1\ \algorithmicwith}{\algorithmicend\ \algorithmicmatch}%
\algdef{SE}[CASE]{Case}{EndCase}[1]{\algorithmiccase\ #1 \algorithmicthen}{\algorithmicend\ \algorithmiccase}%
\algdef{SE}[DEFAULT]{Default}{EndDefault}{\hskip\algorithmicindent\algorithmicdefault}{\algorithmicend\algorithmicdefault}%
\algtext*{EndMatch}%
\algtext*{EndCase}%
\algtext*{EndDefault}%
\newcommand{\IndentState}{\State\hspace{\algorithmicindent}}

\usepackage{booktabs}

\usepackage{subcaption} %% For complex figures with subfigures/subcaptions
\usepackage{wrapfig}

\usepackage[labelfont={bf},font=footnotesize,skip=2pt]{caption}           % reduces space after caption
\tikzset{2tail/.code={\pgfsetarrowsstart{Implies[reversed]}}}
\newcommand{\Rightarrowtail}{
  \mathrel{
    \begin{tikzpicture}
      \draw[double equal sign distance, -implies, 2tail] (0,0) -- (0.4,0);
    \end{tikzpicture}
  }
}

\newcommand{\true}{\mathrm{true}}
\newcommand{\false}{\mathrm{false}}
\newcommand{\At}{\mathbf{At}}
\newcommand{\JmpRes}[1]{\!\downarrow _{#1}}
\newcommand{\sem}[1]{\llbracket #1 \rrbracket}
\newcommand{\transCorr}[1]{\rightarrowtail^{#1}} %"transition correctly"
\newcommand{\transCorrSymb}[1]{\Rightarrowtail^{#1}} %"transition correctly" for symbolic
\newcommand{\powSet}[1]{\mathcal{P}({#1})}
\newcommand{\refClo}{\mathsf{ref}}
\newcommand{\symClo}{\mathsf{sym}}
\newcommand{\transClo}{\mathsf{trans}}

\newcommand{\accu}{\mathrm{accu}}
\newcommand{\res}{\mathrm{res}}
\newcommand{\con}{\mathrm{con}}

\newcommand{\command}[1]{\texttt{#1}}
\newcommand{\comAssert}[1]{\command{assert}~#1}
\newcommand{\comSkip}{\command{skip}}

\newcommand{\comIfElse}[3]{\command{if}~#1~\command{then}~#2~\command{else}~#3}
\newcommand{\comIfThen}[2]{\command{if}~#1~\command{then}~#2}

\newcommand{\comWhile}[2]{\command{while}~#1~\command{do}~#2}
\newcommand{\comDoWhil}[2]{\command{do}~#1~\command{while}~#2}
\newcommand{\comLabel}[1]{\command{label}\!~#1}
\newcommand{\comBrk}{\command{break}}
\newcommand{\comCont}{\command{continue}}
\newcommand{\comRet}{\command{return}}
\newcommand{\comGoto}[1]{\command{goto}\!~#1}

\definecolor{forKeyword}{RGB}{13,71,161}
\definecolor{forCommand}{RGB}{130,119,23}
\newcommand{\colorCommand}[1]{{\color{forCommand}{\texttt{\small{#1}}}}}
\newcommand{\colorKeyword}[1]{{\color{forKeyword}{\texttt{\small{#1}}}}}
\newcommand{\colorAssert}[1]{\colorCommand{assert}~#1}

\newcommand{\colorDiverge}{\colorCommand{diverge}}

\newcommand{\colorIfThen}[2]{\colorKeyword{if}~#1~\colorKeyword{then}~#2}
\newcommand{\colorElseIf}[2]{\colorKeyword{else}~\colorKeyword{if}~#1~\colorKeyword{then}~#2}
\newcommand{\colorElse}[1]{\colorKeyword{else}~#1}
\newcommand{\colorWhile}[2]{\colorKeyword{while}~#1~\colorKeyword{do}~#2}
\newcommand{\colorLabel}[1]{\colorCommand{label}\!~#1}

\newcommand{\colorRet}{\colorCommand{return}}
\newcommand{\colorGoto}[1]{\colorCommand{goto}\!~#1}

\newcommand{\algCall}[1]{{\textsc{#1}}}

\newcommand{\algIsDead}{\algCall{isDead}}
\newcommand{\algKnownDead}{\algCall{knownDead}}
\newcommand{\algEquiv}{\algCall{equiv}}

\newcommand{\algUnion}{\algCall{union}}
\newcommand{\algRep}{\algCall{rep}}

\newcommand{\unfold}{\vartriangleleft}
\newcommand{\seq}{\fatsemi}

\newcommand{\longrightsquigarrow}{\xymatrix{{}\ar@{~>}[r]&{}}}
\newcommand{\longleftrightsquigarrow}{\xymatrix{{}\ar@{<~>}[r]&{}}}

\newcommand{\conti}[1]{\mathbf{#1}}
\newcommand{\acc}[1]{{\conti{acc} ~ #1}}  % accept with an output indicator value
\newcommand{\contc}[1]{{\conti{cont} ~ #1}}  % set variables in NetKAT
\newcommand{\retc}{{\conti{ret}}}  % return, terminate the program
\newcommand{\brkc}[1]{\conti{brk} ~ #1}  % break with an output indicator value
\newcommand{\jmpc}[1]{{\conti{jmp} ~ #1}} % goto a label

\newcommand{\theoryOf}[1]{\ensuremath{\mathsf{#1}}}
\newcommand\gkat{\textsf{GKAT}\xspace}

\newcommand\cfgkat{\textsf{CF-GKAT}\xspace}

\newcommand\CFOrGKAT{\textsf{(CF-)GKAT}\xspace}
\newcommand{\kat}{\textsf{KAT}\xspace}

\newcommand{\BExp}{\theoryOf{BExp}}
\newcommand{\BExpI}{\theoryOf{BExp}^I}
\newcommand{\transvia}[1]{
  \mathrel{\raisebox{-2pt}{\(\xrightarrow{#1}\)}}
}
\newcommand{\transOut}[1]{
  \mathrel{\adjustbox{trim=0 5pt 0 0}{\(\xRightarrow{#1}\)}}
}
\newcommand{\transRes}[1]{
  \mathrel{\adjustbox{trim=0 2pt 0 0}{\(\xhookrightarrow{#1}\)}}
}
\newcommand{\transRej}[1]{\lightning^{#1}}

\begin{document}

%%
%% The "title" command has an optional parameter,
%% allowing the author to define a "short title" to be used in page headers.
\title{Outrunning Big KATs: Efficient Decision Procedures for Variants of GKAT}

\author{
  Cheng Zhang
  \inst{1}
  \orcidID{0000-0002-8197-6181} 
  \and
  Qiancheng Fu
  \inst{2}
  \orcidID{0000-0002-5234-8565} 
  \and
  Hang Ji
  \inst{2}
  \and  
  Ines Santacruz Del Valle
  \inst{2}
  \and  
  Alexandra Silva 
  \inst{4}
  \orcidID{0000-0001-5014-9784}
  \and  
  Marco Gaboardi 
  \inst{2}
  \orcidID{0000-0002-5235-7066}
}
% \author{Cheng Zhang}
% \authornote{Work largely performed at University College London and Boston University}
% \email{czhang13@wpi.edu}
% \orcid{0000-0002-8197-6181}
% \affiliation{%
%   \institution{Worcester Polytechnic Institute}
%   \city{Worcester}
%   \state{M.A.}
%   \country{USA}
% }

% \author{Qiancheng Fu}
% \email{qcfu@bu.edu}
% \orcid{0000-0002-5234-8565}
% \affiliation{%
%   \institution{Boston University}
%   \city{Boston}
%   \state{M.A.}
%   \country{USA}
% }

% \author{Hang Ji}
% \affiliation{%
%   \institution{Boston University}
%   \city{Boston}
%   \state{M.A.}
%   \country{USA}
% }

% \author{Ines Santacruz Del Valle}
% \affiliation{%
%   \institution{Boston University}
%   \city{Boston}
%   \state{M.A.}
%   \country{USA}
% }

% \author{Alexandra Silva}
% \email{alexandra.silva@cornell.edu}
% \orcid{0000-0001-5014-9784}
% \affiliation{%
%   \institution{Cornell University}
%   \city{Ithaca}
%   \state{N.Y.}
%   \country{USA}
% }

% \author{Marco Gaboardi}
% \email{gaboardi@bu.edu}
% \orcid{0000-0002-5235-7066}
% \affiliation{%
%   \institution{Boston University}
%   \city{Boston}
%   \state{M.A.}
%   \country{USA}
% }
%
\authorrunning{C. Zhang et al.}
% First names are abbreviated in the running head.
% If there are more than two authors, 'et al.' is used.
%
\institute{
  Worcester Polytechnic Instituite, Worcester, MA, U.S.A.
  \email{czhang13@wpi.edu}
  \and
  Boston University, Boston, MA, U.S.A.
  \email{\{qcfu, gaboardi\}@bu.edu}
  \and
  Cornell University, Ithaca, NY, U.S.A. 
  \email{alexandra.silva@cornell.edu}}
\maketitle              % typeset the header of the contribution
\begin{abstract}
  This paper presents several efficient decision procedures for trace equivalence of \gkat automata, which make use of on-the-fly symbolic techniques via SAT solvers.
%  First, we overcome the incompatibility of normalization with on-the-fly equivalence checking by only executing the dead-state detection on discrepancies of bisimulation. 
  %This development allows our algorithms to terminate immediately upon counter-example, preventing the iteration of the entire automaton, which is required by the original normalization procedure. 
 % Second, we overcome the exponential blowup of transitions by grouping them into symbolic representations based on Boolean formulae. 
  %Compared to the previous BDD-based approaches to \kat equivalence, our algorithms allow for the use of arbitrary (UN)SAT solvers for equivalence checking, including BDDs. 
  To demonstrate applicability of our algorithms, we designed symbolic derivatives for \cfgkat, a practical system based on \gkat designed to validate control-flow transformations. 
  We implemented the algorithms in Rust and evaluated them on both randomly generated benchmarks and real-world control-flow transformations. 
  Indeed, we observed order-of-magnitude performance improvements against existing implementations for both \kat and \cfgkat. Notably, our experiments also revealed a bug in Ghidra, an industry-standard decompiler, highlighting the practical viability of these systems.
%\keywords{
%  Guarded Kleene Algebra with Tests (GKAT)
%  \and Control-flow GKAT 
%  \and Trace Equivalence
%  \and Control-flow Transformation}
\end{abstract}

\section{Introduction}

% \begin{itemize}
%   \item Front matter 
%   \begin{itemize}
%     \item automata approach of program verification -> KAT out-perform state of the art -> fast fragment -> GKAT
%     \item \gkat is not just a fragment of KAT, infinite trace and probabilistic models 
%     \item \gkat is powerful, \cfgkat in particular
%     \item finite trace v.s. infinite trace.
%   \end{itemize}
%   \item Contribution 
%   \begin{itemize}
%     \item Inefficiencies when reason about finite traces: normalization.
%     \item symbolic techniques: formula v.s. BDD; mention a bit that the differences between previous formula based systems.
%   \end{itemize}
%   \item Contribution Bullet Point
% \end{itemize}

%Algebraic approaches for program equivalence and verification -- KAT to GKAT
%Automata-theoretical approaches to program verification have seen
%increased development in recent years.  One such approach is
%based on 
Kleene Algebra with Tests
(KAT)~\cite{kozen_KleeneAlgebraTests_1997c} is an algebraic framework which has 
been widely applied for modeling and verification in different
domains, including but not limited to software-defined
networks~\cite{anderson_NetKATSemanticFoundations_2014,smolka_CantorMeetsScott_2017},
concurrency~\cite{wagemaker_PartiallyObservableConcurrent_2020,hoare_ConcurrentKleeneAlgebra_2009a},
program
logics~\cite{zhang_IncorrectnessLogicKleene_2022c,zhang_DomainReasoningTopKAT_2024,desharnais_ModalKleeneAlgebra_2004},
and program
transformations~\cite{angus_KleeneAlgebraTests_2001,kozen_CertificationCompilerOptimizations_2000}.
Besides its broad applicability, Kleene Algebra with Tests has also
attracted interest for the efficiency of some of its decision procedures.
Despite the fact that deciding equivalence for general KAT expressions
is PSPACE-complete~\cite{cohen_ComplexityKleeneAlgebra_1996a}, in some
applications, decision procedures based on subsets of KAT are able to
outperform state-of-the-art tools with domain-specific
tasks~\cite{smolka_ScalableVerificationProbabilistic_2019}.
This has led to the systematic study of an efficient
fragment of KAT, named \emph{Guarded Kleene Algebra with Tests}
(GKAT)~\cite{smolka_GuardedKleeneAlgebra_2020}, where the
non-deterministic addition \(e + f\) and iteration \(e^*\) operations
are restricted to if-statement \(e +_b f\) and while-loop \(e^{(b)}\),
respectively.
% This work also establishes a nearly-linear time decision procedure for (finite-)trace equivalence via \gkat automata~\cite{kozen_NonlocalFlowControl_2008,kozen_BohmJacopiniTheorem_2008}, when the number of primitive tests is fixed.

%\gkat is powerful despite restricted

Despite being a more restrictive language, \gkat has proven to be a powerful framework, and also found applications in different domains, from control-flow equivalences~\cite{zhang_CFGKATEfficientValidation_2025}, to network verification~\cite{smolka_ScalableVerificationProbabilistic_2019}, to weighted and probabilistic programming~\cite{ro.zowski_ProbabilisticGuardedKAT_2023,gomes_KleeneAlgebraTests_2024,koevering_WeightedGKATCompleteness_2025}. 
In particular, the extension presented in~\cite{zhang_CFGKATEfficientValidation_2025}---Control-flow \gkat (\cfgkat)---augments \gkat with non-local control-flow structures, like \command{goto} and \command{break}, and indicator variables; providing a comprehensive framework to validate control-flow transformations.
Take the programs in~\Cref{fig:equiv-cfgkat-prog} as an example: the program on the left is a high-level program, and the program on the right can be an optimized pseudo-assembly that was obtained from the high-level program.
Despite their different appearances, \cfgkat is able to validate their equivalence with automata-based techniques.

\begin{figure}[!t]
  \begin{subfigure}{0.5\textwidth}
    \[
      \begin{aligned}
        & \colorIfThen{t_{1}}{p}; \\[-.5ex]  
        & \colorWhile{t_{1} \lor t_{2}}{\\[-.5ex]
        & \quad \colorIfThen{t_{1} \land \overline{t_{2}}}{(q; \colorAssert{t_{1} \land \overline{t_{2}}});}} \\[-.5ex]   
        & \quad \colorElse{p;}\\[-.5ex] 
        & \colorRet;\\[-3ex]
      \end{aligned}
    \]
    \caption{A \texttt{while}-loop program.}\label{fig:high-level-equiv-cfgkat-prog}
  \end{subfigure}
  \begin{subfigure}{0.5\textwidth}
    \[
      \begin{aligned}
        & \colorIfThen{t_{1} \lor t_{2}}{p;} \\[-.5ex] 
        & \colorElse{\colorLabel{l}; \\[-.5ex] 
        & \quad \colorIfThen{t_{2}}{p; \colorGoto{l};} \\[-.5ex]   
        & \quad \colorElseIf{t_{1}}{\colorDiverge;}} \\[-.5ex]   
        & \colorRet;\\[-2ex]
      \end{aligned}
    \]
    \caption{A \texttt{goto} program.}\label{fig:low-level-equiv-cfgkat-prog}
  \end{subfigure}
  \caption{Two finite-trace equivalent \cfgkat programs, where \((q; \comAssert{t_{1} \land \overline{t_{2}}})\) encodes the knowledge that the test \(t_{1} \land \overline{t_{2}}\) always holds after action \(q\).}\label{fig:equiv-cfgkat-prog} \vspace*{-.4cm}
\end{figure}

 Indeed, at the core of the success of languages like \kat, \gkat, and extensions thereof, is a strong automata-theoretic foundation. The original decision procedures of both \gkat and \cfgkat rely on a special type of automaton introduced as the operational semantics of the languages---\emph{\gkat automata}~\cite{zhang_CFGKATEfficientValidation_2025,smolka_GuardedKleeneAlgebra_2020}.
Concretely, both decision procedures first convert expressions into \gkat automata, then check the trace equivalence of the generated automata.
Notably, the trace-equivalence checking for \gkat automata is \emph{sound and complete}, that is, the algorithm outputs true \emph{if and only if} the input automata are trace equivalent.

Thus, improving the existing equivalence procedure for \emph{\gkat automata} will not only speed up the decision procedure for \gkat, but also some of its variants.
In this paper, we made two important improvements on the \textcolor{red}{original} algorithm for equivalence check of \gkat-automata, as outlined in~\Cref{fig:equiv-algs-summary}. 
Firstly, we designed an {\color{blue}on-the-fly algorithm} for equivalence checking: the algorithm can terminate as soon as a counter example is encountered without iterating through the entire generated automata.
Secondly, we lift the {\color{blue}on-the-fly algorithm} to a {\color{teal}symbolic algorithm}, which not only enjoys short-circuiting upon counter-examples, but also prevents the exponential blow-up with the addition of new primitive tests.
In the following paragraphs, we will delve deeper into the challenges in designing these algorithms and our contributions.

% %original decision procedures
% Indeed, there exist decision procedures for both (finite-)trace equivalence and bisimulation equivalence~\cite{smolka_CantorMeetsScott_2017,schmid_GuardedKleeneAlgebra_2021}\footnote{Bisimulation equivalence can be checked by performing bisimulation on \gkat automata}.
% Impressively, both algorithm only takes nearly-linear time for a fixed number of primitive tests; and both algorithms are \emph{sound and complete}, that is, they do not produce any false positive and false negative with respect to their respective notion of equivalences.

% %problems and solutions
% However, translating the complexity result of \gkat to an algorithm
% that scales to thousands of primitives is not so straightforward.  In
% the following paragraphs, we will outline the challenges in designing
% efficient decision procedures for variants of \gkat, and present our
% contributions.
%
% \marco{Here something I find confusing: so far we didn't mention at
%   all that the existing approaches are based on automata, and now we
%   dig deep in how these approaches work. I think we should first give
%   a couple of sentence high-level description of existing algorithms,
%   then we can discuss the details. Also, we never mentioned automata
%   before, except for the first sentence and the footnote. I think it
%   would be clearer to say that KAT is an equational theory but when we
%   are interested in deciding equivalence, automata methods are more
%   effective.}
\noindent\textbf{Normalization v.s. On-the-fly:} 
originally, deciding trace equivalence requires a \emph{normalization} procedure, which reroutes transitions that do not produce finite traces (e.g. infinite loops) into immediate rejection.
This simple step introduces significant performance penalties: although \gkat automata can be generated on-the-fly~\cite{schmid_GuardedKleeneAlgebra_2021}, normalization  requires iterating through the entire automaton once, which nullifies the performance improvements of on-the-fly verification~\cite{fernandez_OntheflyVerificationFinite_1992}.
In~\Cref{sec:on-the-fly-norm-overview}, we present our {\color{blue}on-the-fly algorithm}, which invokes normalization {\em lazily}: we only perform it when the bisimulation failed.
% %
% \marco{I'm confused, what is ``bisimulation'' here referring to? we
%   said above that checking trace equivalence is different from
%   bisimulation, so why now we have that bisimulation is a component of
%   checking trace equivalence?}
% %
This change allows the algorithm to halt immediately when a
counter-example is encountered, greatly improving the efficiency in the
negative case.  Moreover, our normalization algorithm only iterates
through the reachable states of the failed checks, instead of the
entire automata, improving the performance in the positive cases as
well.

\begin{figure}[!t]
  \begin{tikzpicture}
    \node[flowchartnode] (prog) {\CFOrGKAT Program};
    \node[flowchartnode] (aut) [right=of prog] {\gkat Automata};
    \node[flowchartnode] (normAut) [above right=0.5cm and -.5cm of aut, align=center] {Normalized\\\gkat Automata};
    \node[flowchartnode] (equivAut) [right=2cm of aut, align=center] {Trace\\Equivalence};
    \node[flowchartnode] (symbAut) [below=1cm of aut, align=center] {Symbolic\\\gkat Automata};
    % \node[flowchartnode] (equivSymb) at ($(equivAut |- symbAut)$) {\
    % \begin{tabular}{c}
    %   Trace\\Equivalence
    % \end{tabular}};
    
    \draw[->, thick, red] (prog) edge[bend left=7] node[auto, text=black] {
      \begin{tabular}{c}
        Thompson's\\Construction\\
        \cite{smolka_GuardedKleeneAlgebra_2020,zhang_CFGKATEfficientValidation_2025}
      \end{tabular}
    } (aut);
    \draw[->, thick, blue] (prog) edge[bend right=7] node[auto, swap, text=black] {
      \begin{tabular}{c}
        Derivative\\
        \cite{schmid_GuardedKleeneAlgebra_2021}
      \end{tabular}
    } (aut);
    \draw[->, thick, red] (aut) edge node[auto, text=black] {Normalization} (normAut);
    \draw[->, thick, red] (normAut) edge node[auto, text=black, pos=0.6] {Bisimulation} (equivAut);
    \draw[->, thick] (symbAut) edge node[auto, swap, align=left] {
      Concretization\\\Cref{def:concretization}} (aut);
    \draw[->, thick, teal] (symbAut) edge[bend right=25] node[below right, text=black, align=left, pos=0.7] {
      On-the-fly\\Symbolic Algorithm\\
      \Cref{alg:symb-bisim}} (equivAut);
    % \draw[double distance=3pt,thick] (equivSymb) edge (equivAut);
    \draw[->, thick, teal] (prog) edge[bend right=25] node[auto, swap, text=black, align=right, pos=0.3] {
      Symbolic Derivative
      \\\Cref{sec:aut-construction-cfgkat}
      \ifFull\\\Cref{ap:aut-construction-gkat}\fi} (symbAut);
    \draw[->, thick, blue] (aut) edge node[below, text=black, align=center] {On-the-fly\\\Cref{alg:nonsymb-bisim}} (equivAut);
  \end{tikzpicture}
  \caption{A summary of algorithms mentioned in this paper: the {\color{red}original} one is marked in {\color{red}red}; the {\color{blue}on-the-fly} is marked in {\color{blue}blue}; and the {\color{teal}symbolic on-the-fly} is marked in {\color{teal}green}.}\label{fig:equiv-algs-summary}\vspace{-.4cm}
\end{figure}
% \alex{both in (1) and (2) we need to add forward references to the algorithms, either sections or figure numbers? }
\noindent\textbf{Symbolic Equivalence Checking:}
\gkat and \kat automata both transition via \emph{atoms}, which are truth assignments for all the primitive tests.
This means the sizes of automata are exponential to the number of primitive tests. 
In~\Cref{sec:symb-aut-overview}, we present the notion of {\color{teal}symbolic \gkat automata}, which groups collections of atoms using Boolean formulae, circumventing the exponential blow-up.
Additionally, we also present a {\color{teal}symbolic on-the-fly algorithm} to determine the trace equivalence between two symbolic \gkat automata.
We use a Boolean-formulae-based approach which contrasts with a more standard symbolic approach to \kat pioneered in SymKAT~\cite{pous_SymbolicAlgorithmsLanguage_2015} using  binary decision diagrams (BDDs). 
We will show that our approach opens the door to the use of arbitrary SAT solvers to store and compare Boolean expressions, which in turn leads to better performance than SymKAT (see \Cref{sec:implementation}) in checking the trace equivalence of \gkat expressions.  
Of course, as a trade-off, our symbolic equivalence checking algorithm currently only works on \gkat automata, instead of the more general \kat automata.

% \alex{ you had here mentioned Boolean formulae~\cite{veanes_ApplicationsSymbolicFinite_2013,stanford_SymbolicBooleanDerivatives_2021} from Veanes -- do we say something about comparison to then?}.

\noindent{\textbf{Correctness:}}
In~\Cref{sec:correctness-and-complexity}, we show that both our {\color{blue}on-the-fly algorithm} (\Cref{alg:nonsymb-bisim}) and {\color{teal}symbolic on-the-fly algorithm} (\Cref{alg:symb-bisim}) for trace equivalence are \emph{sound and complete}, i.e. they will produce neither false-positive nor false-negative with respect to trace equivalence.
Additionally, we show that both of our algorithms can be tweaked to correctly check infinite-trace equivalence~\cite{schmid_CoalgebraicCompletenessTheorems_2024}, a useful notion of equivalence which SymKAT does not support.

\noindent\textbf{Expressions to Symbolic Automata:} 
while the equivalence algorithm is designed for automata, one ultimately wants to reason about program equivalence, therefore a bridge between expressions and automata is needed. 
% In~\Cref{sec:aut-construction-gkat}, We design a {\color{teal}symbolic automata generation algorithm for \gkat expressions}; when combined with the symbolic equivalence checking for automata, yields an efficient equivalence checking algorithm for \gkat expressions. 
% Moreover, our foundational work surrounding \gkat can be extended to more practical systems derived from \gkat. 
In~\Cref{sec:aut-construction-cfgkat}, we present an on-the-fly algorithm to generate {\color{teal} symbolic \gkat automata} from \cfgkat program~\cite{zhang_CFGKATEfficientValidation_2025}, an extension of \gkat with non-local control-flow structures and indicator variables. 
The complexity of these extra constructs make the derivative construction for \cfgkat quite intricate and not a simple instance of \gkat.
Additionally, since \gkat is a subsystem of \cfgkat, this algorithm also induces an on-the-fly generation algorithm for \gkat expressions.
%Thus, we believe the symbolic derivative for \cfgkat constitutes a contribution as well.

In a nutshell, in this paper we designed symbolic algorithms to overcome several crucial performance issues in the original decision procedure, while preserving the soundness and completeness guarantee. 
To demonstrate the efficiency of our approach, we provide a Rust implementation and compared it with state-of-the-art tools for \CFOrGKAT equivalence checking, as detailed in~\Cref{sec:implementation}. We will show that our implementation exhibits orders-of-magnitude improvements in both runtime and memory usage, and the flexibility to switch between different (UN)SAT solvers allows us to utilize the most efficient solvers for each specific scenario. 
Moreover, during our evaluation, we also identified a bug~\cite{nationalsecurityagency_possbile_2025} in the industry standard decompiler Ghidra~\cite{NationalSecurityAgency_Ghidra_2025,nationalsecurityagency_possbile_2025} demonstrating the real-world applicability of our tools.
\ifFull\else
The complete proof and the detailed symbolic derivative for \gkat expressions is presented in the full version~\cite{zhangOutrunningBigKATs2026}.
\fi

In the next section, we explain the motivations and intuitions behind our approaches with examples, and introduce primaries used in our development.

% \alex{in the contributions you don't mention the soundness section?}

% \alex{there's the issue of original decision procedure being for trace and now you have both -- this is not highlighted}

\section{Overview}\label{sec:overview}

In this section, we provide an overview of the paper's main contributions through examples. In \Cref{sec:original-alg-overview}, we introduce basic notations and explain the core motivation of our developments: to overcome two challenging inefficiencies in the original trace-equivalence algorithms for \gkat automata~\cite{smolka_GuardedKleeneAlgebra_2020}: the lack of support for short-circuiting on counter-examples due to the ``normalization'' process and the exponential increase of transitions with the addition of primitive tests.

In \Cref{sec:on-the-fly-norm-overview,sec:symb-aut-overview}, we explain how we address the aforementioned problems in two stages. 
First, we propose a new trace-equivalence decision procedure, that only invokes the normalization when necessary.
This crucial change not only enables short-circuiting on counter-examples, but also reduces the number of states needed to be checked. 
% For bisimilar automata with no dead states, our proposed algorithm can even half the time complexity of the original algorithm: where the original algorithm requires two iterations through both automata, ours only requires one.
Second, we extend our algorithm to a novel symbolic representation of \gkat automata, where transitions with the same source and target are often combined into one, effectively mitigating the exponential growth of transitions, and enabling scaling to large programs.

\subsection{Original \cfgkat Decision Procedures}\label{sec:original-alg-overview}

%% Sketch of the subsection
% \begin{itemize}
%   \item Sketch the intuition why these two are equivalent, i.e. the \(c\) branch will always lead to infinite loops, which is equivalent to rejection in trace semantics.
%   \item Sketch the algorithm of equivalence checking
%   \begin{itemize}
%     \item \gkat Automata construction, derivative;
%     \item Identifying dead states and normalization;  
%     \item Exhibit bisimulation by coloring state.
%   \end{itemize}
%   \item Problem with the algorithm 
%   \begin{itemize}
%     \item automata have too many transitions.
%     \item non-locality: even though some transition (\(p; e_{loop} \transvia{\true \mid p} e_{loop}\)) will be performed regardless of the truth value of primitive tests \(b\) and \(c\), the transition still needs to carry the truth value of \(b\) and \(c\), simply because \(b\) and \(c\) appears elsewhere in the program. This prevents scaling the program to larger programs.
%     \item dead state detection goes through all the states, whereas two dead states only have two reachable state. 
%   \end{itemize}
% \end{itemize}

We begin by briefly walking through original decision procedures for Guarded Kleene Algebra with Tests (GKAT)~\cite{smolka_GuardedKleeneAlgebra_2020} and Control-flow \gkat (\cfgkat)~\cite{zhang_CFGKATEfficientValidation_2025}, using the example \cfgkat programs presented in~\Cref{fig:equiv-cfgkat-prog},
% \alex{the reference is scrambled, says 2a?}

% We need to introduce some notations to explain how the equivalence algorithm of~\citet{zhang_CFGKATEfficientValidation_2025} works on these programs. 
The original algorithm~\cite{zhang_CFGKATEfficientValidation_2025} first converts the \cfgkat programs in~\Cref{fig:equiv-cfgkat-prog} into \gkat automata, where each state corresponds to a location in the program, and transitions are based on the current truth assignment of all the primitive tests. 
These truth assignments are represented as \emph{atoms}~\cite{kozen_KleeneAlgebraTests_1997c}: for a finite set of primitive tests \(T \triangleq \{t_{1}, t_{2}, t_{3}, \dots, t_{n}\}\), atoms \(\At\) over \(T\) are sequences of all the tests (in a fixed order) in its positive or negative form:
\[ \alpha,  \beta, \gamma \in \At = \{\tau_{1}~\tau_{2} ~\tau_{3} ~ \dots ~\tau_{n} \mid  \tau_{i} \in \{t_{i}, \overline{t_{i}}\}\}.\]
Based on the given atom, a state in a \gkat automaton will either accept, denoted as \(\retc\); reject, denoted as \(\bot\); or transition to the next state while executing an action in \(\Sigma\).
Formally, \gkat automata consist of a finite set of states \(S\), a start state \(s_{0} \in S\), and a transition function:
\[ \zeta : S \to \At \to \{ \bot , \retc\} + S \times \Sigma.\]
We often use \(s \transOut{\alpha} \retc\) to denote \( \zeta (s)_ \alpha  = \retc\); use \(s \transRej{\alpha}\) to denote \( \zeta (s)_ \alpha  =  \bot \); use \(s \transvia{\alpha \mid p} s'\) to denote \( \zeta (s)_ \alpha  = (s', p)\); and use \(s \transRes{\alpha} r\) to denote \(\zeta(s)_\alpha = r\).

{It is important to observe that the number of atoms is exponential with respect to the number of primitive tests, and so is the number of transitions in a \gkat automaton.}
For example, the programs in~\Cref{fig:equiv-cfgkat-prog} contains two primitive tests: \(t_{1}\) and \(t_{2}\); which gives \(2^2 = 4\) distinct atoms labelling the outgoing transitions of each state in the generated automata, as shown in~\Cref{fig:aut-equiv-cfgkat-prog}.

\begin{figure}[t]
  \begin{subfigure}{0.45\textwidth}
    \centering
    \begin{tikzpicture}
      \node (init) {};
      \node[textstate] (s0) [right=5mm of init] {\({\color[RGB]{27,94,32}s_0}\)};
      \node[textstate] (s1) [right=of s0] {\({\color{blue}s_1}\)};
      \node[textstate] (s2) [below right=of s1] {\({\color{blue}s_2}\)};
      \node[textstate] (s3) [above right=of s2] {\({\color{red}s_3}\)};  
      \node (s0ret) [above=3mm of s0] {\(\retc\)};
      \node (s1ret) [above=3mm of s1] {\(\retc\)};
      \node (s2ret) [below=3mm of s2] {\(\retc\)};
      \draw[->] (init) edge (s0);
      \draw[->] (s0) edge[output-edge] node[left] {\(\overline{t_1}\overline{t_2}\)} (s0ret);
      \draw[->] (s1) edge[output-edge] node[right] {\(\overline{t_1}\overline{t_2}\)} (s1ret);
      \draw[->] (s2) edge[output-edge] node[left] {\(\overline{t_1}\overline{t_2}\)} (s2ret);
      \draw[->] (s0) edge node[above, align=center] {\(t_1 t_2 \mid p\) \\ \(t_1 \overline{t_2} \mid p\)} (s1);
      \draw[->] (s0) edge[out=-90, in=190] node[below] {\(\overline{t_1} t_2 \mid p\)} (s2);
      \draw[->] (s1) edge node[below left=-1mm and -1mm, align=center] {\(t_1 t_2 \mid p\) \\ \(\overline{t_1} t_2 \mid p\)} (s2);
      \draw[->] (s2) edge node[below right=-1mm and -1mm] {\(t_1 \overline{t_2} \mid p\)} (s3);
      \draw[->] (s1) edge node[above] {\(t_1 \overline{t_2} \mid p\)} (s3);
      \draw[->] (s3) edge[out=80,in=125,loop,looseness=5] node[above] {\(t_1 \overline{t_2} \mid q\)} (s3);
      \draw[->] (s2) edge[out=-30,in=0,loop,looseness=8] node[below right=-2mm and 0mm, align=center] {\(t_1 t_2 \mid p\) \\ \(\overline{t_1} t_2 \mid p\)} (s2);
    \end{tikzpicture}
    \caption{The automaton for~\Cref{fig:high-level-equiv-cfgkat-prog}.}\label{fig:aut-high-level-equiv-cfgkat-prog}
  \end{subfigure}
  % \begin{subfigure}{0.45\textwidth}
  %   \[\begin{tikzcd}[row sep=1cm,column sep = 1cm]
  %    	{} & {\color[RGB]{27,94,32}s_0} & \retc \\
  %     & {\color{blue}s_1} & \retc \\
  %     \retc & {\color{blue}s_2} & {\color{red}s_3}
  %     \arrow[from=1-1, to=1-2]
  %     \arrow["{\overline{t_1}\overline{t_2}}", Rightarrow, from=1-2, to=1-3]
  %     \arrow["{t_1t_2 \mid p}"', shift left, from=1-2, to=2-2]
  %     \arrow["t_1\overline{t_2} \mid p"{pos=0.4}, shift left, curve={height=-9pt}, from=1-2, to=2-2]
  %     \arrow["{\overline{t_1}t_2 \mid p}"'{pos=0.4}, shift right, curve={height=24pt}, from=1-2, to=3-2]
  %     \arrow["{\overline{t_1}\overline{t_2}}", Rightarrow, from=2-2, to=2-3]
  %     \arrow["{t_1 t_2 \mid p}"', shift left, from=2-2, to=3-2]
  %     \arrow["{\overline{t_1}t_2 \mid p}", shift left=2, curve={height=-6pt}, from=2-2, to=3-2]
  %     \arrow["{t_1\overline{t_2} \mid q}"{pos=0.7}, shift right, curve={height=-24pt}, from=2-2, to=3-3]
  %     \arrow["{\overline{t_1}\overline{t_2}}"', Rightarrow, from=3-2, to=3-1]
  %     \arrow["{t_1 t_2 \mid p}"', from=3-2, to=3-2, loop, in=300, out=240, distance=5mm]
  %     \arrow["{\overline{t_1}t_2 \mid p}"', from=3-2, to=3-2, loop, in=310, out=230, distance=15mm]
  %     \arrow["{t_1\overline{t_2} \mid q}"', from=3-2, to=3-3]
  %     \arrow["{t_1\overline{t_2} \mid p}"', from=3-3, to=3-3, loop, in=300, out=240, distance=5mm]
  %   \end{tikzcd}\]
  %   \caption{The automaton for program in~\Cref{fig:high-level-equiv-cfgkat-prog}.}\label{fig:aut-high-level-equiv-cfgkat-prog}
  % \end{subfigure}
  \hfil
  \begin{subfigure}{0.45\textwidth}    
    \centering
    \begin{tikzpicture}
      \node (init) {};
      \node[textstate] (u0) [right=5mm of init] {\({\color[RGB]{27,94,32}u_0}\)};
      \node[textstate] (u1) [right=of u0] {\({\color{blue}u_1}\)};
      \node[textstate] (u2) [right=of u1] {\({\color{blue}u_2}\)};
      \node (u0ret) [above=3mm of u0] {\(\retc\)};
      \node (u1ret) [above=3mm of u1] {\(\retc\)};
      \node (u2ret) [below=3mm of u2] {\(\retc\)};
      \draw[->] (init) edge (u0);
      \draw[->] (u0) edge[output-edge] node[left] {\(\overline{t_1}\overline{t_2}\)} (u0ret);
      \draw[->] (u1) edge[output-edge] node[left] {\(\overline{t_1}\overline{t_2}\)} (u1ret);
      \draw[->] (u2) edge[output-edge] node[right] {\(\overline{t_1}\overline{t_2}\)} (u2ret);
      \draw[->] (u0) edge node[below, align=center] {\(t_1 t_2 \mid p\) \\ \(\overline{t_2} t_2 \mid p\) \\ \(t_1 \overline{t_2} \mid p\)} (u1);
      \draw[->] (u1) edge node[above, align=center] {\(t_1 t_2 \mid p\) \\ \(\overline{t_1} t_2 \mid p\)} (u2);
      \draw[->] (u2) edge[out=45,in=100,loop,looseness=4] node[above, align=center] {\(t_1 t_2 \mid p\) \\ \(\overline{t_1} t_2 \mid p\)} (u2);
    \end{tikzpicture}
    \caption{The automaton for~\Cref{fig:low-level-equiv-cfgkat-prog}.}\label{fig:aut-low-level-equiv-cfgkat-prog}
  \end{subfigure}
  %   \begin{subfigure}{0.45\textwidth}
  %     \[\begin{tikzcd}[row sep=1cm, column sep=1cm]
  %       {} & {\color[RGB]{27,94,32}u_0} & \retc \\
  %       & {\color{blue}u_1} & \retc \\
  %       & {\color{blue}u_2} & \retc
  %       \arrow[from=1-1, to=1-2]
  %       \arrow["{\overline{t_1}\overline{t_2}}", Rightarrow, from=1-2, to=1-3]
  %       \arrow["{t_1t_2 \mid p}", shift right, from=1-2, to=2-2]
  %       \arrow["{{\overline{t_1}t_2 \mid p}}"{pos=0.4}, shift left, curve={height=-18pt}, from=1-2, to=2-2]
  %       \arrow["{{t_1\overline{t_2} \mid p}}"'{pos=0.4}, shift right, curve={height=18pt}, from=1-2, to=2-2]
  %       \arrow["{\overline{t_1}\overline{t_2}}", Rightarrow, from=2-2, to=2-3]
  %       \arrow["{t_1 t_2 \mid p}"', shift right, from=2-2, to=3-2]
  %       \arrow["{\overline{t_1}t_2 \mid p}", shift left, from=2-2, to=3-2]
  %       \arrow["{\overline{t_1}t_2 \mid p}"', from=3-2, to=3-2, loop, in=310, out=230, distance=15mm]
  %       \arrow["{t_1t_2 \mid p}"', from=3-2, to=3-2, loop, in=300, out=240, distance=5mm]
  %       \arrow["{\overline{t_1}\overline{t_2}}", Rightarrow, from=3-2, to=3-3]
  %     \end{tikzcd}\]
  %     \caption{The automaton for program in~\Cref{fig:low-level-equiv-cfgkat-prog}.}\label{fig:aut-low-level-equiv-cfgkat-prog}
  %   \end{subfigure}
  \caption{The automata generated for programs in~\Cref{fig:equiv-cfgkat-prog} using the algorithm in~\cite[Section 3]{zhang_CFGKATEfficientValidation_2025}. We omit transitions that output rejection. All the dead states of the automata are highlighted in {\color{red}red} and equivalent states are marked using the same colors.}\label{fig:aut-equiv-cfgkat-prog}\vspace{-.4cm}
\end{figure}

After generating the automata, the algorithm performs a process called \emph{normalization}, which reroutes all transitions that cannot lead to acceptance into immediate rejection~\cite{smolka_GuardedKleeneAlgebra_2020}.
To identify these transitions, the algorithm iterates through the entire automaton and identifies all the \emph{dead states}: states that never reaches acceptance; and then modifies all transitions to dead states into rejection.
For instance, \({\color{red}s_{3}}\) is a dead state in~\Cref{fig:aut-high-level-equiv-cfgkat-prog}, therefore the transition \({\color{blue}s_{1}} \transvia{t_{1}\overline{t_{2}} \mid p} {\color{red}s_{3}}\) in~\Cref{fig:aut-high-level-equiv-cfgkat-prog} will be modified to rejection \({\color{blue}s_{1}} \transRej{t_{1}\overline{t_{2}}}\) during normalization.

Finally, to decide the trace equivalence of these programs (\Cref{def:trace-sem}), the algorithm performs bisimulation~\cite[Algorithm 1]{smolka_GuardedKleeneAlgebra_2020} on the normalized \gkat automata. 
In~\Cref{fig:aut-equiv-cfgkat-prog}, we mark the trace equivalent states with the same color.

Despite its simplicity, the original algorithm is sufficient to decide trace equivalence between two \gkat automata without producing false negatives or false positives; that is, it is \emph{correct} or \emph{sound and complete}~\cite{smolka_GuardedKleeneAlgebra_2020,zhang_CFGKATEfficientValidation_2025}.
Notably, \textbf{normalization is necessary for this correctness result}: if the transition into \({\color{red}s_3}\) is not removed, the pair of automata in~\Cref{fig:aut-equiv-cfgkat-prog} will fail the bisimulation checking, even though they are indeed trace equivalent.

\medbreak
In this work, we propose novel algorithms for (\textsf{CF)}-\gkat equivalence that we will show scales to expressions with thousands of tests and actions, while preserving the simplicity, soundness, and completeness of the original one. Our algorithms are designed to solve the two aforementioned challenging performance problems that remained open since the work of~\citet{smolka_GuardedKleeneAlgebra_2020}.
\begin{enumerate}[leftmargin=*, nosep]
 \item 
    The original algorithm requires one iteration of the entire automata to perform normalization. 
    This necessity forces the generation of the entire automata, preventing the use of techniques like derivatives, which enable short-circuiting on counter-examples and avoid generating the remaining automata.
 \item 
    The number of transitions is exponential to the number of primitive tests.  
    In program with just tens of primitive tests, the generated automata can already contain more than thousands of outgoing transitions per state, resulting in inefficient storage and iteration.
\end{enumerate}

In the following subsections, we will delve deeper into these two problems through illustrative examples, and outline our solutions. 
We then demonstrate how our algorithm outperforms the original approaches, including the symbolic algorithm provided by~\citet{pous_SymbolicAlgorithmsLanguage_2015}.
Besides the performance improvement, our algorithm also preserves the soundness and completeness of the original algorithms, providing a strong theoretical guarantee for the correctness of our results.

\subsection{When Normalization Flies}\label{sec:on-the-fly-norm-overview}

On-the-fly automata generation techniques, like derivatives~\cite{brzozowski_DerivativesRegularExpressions_1964,antimirov_PartialDerivativesRegular_1996,kozen_CoalgebraicTheoryKleene_2017}, have long been an essential for efficient bisimulation construction.
Since the automata are generated on-demand while checking for bisimulation, the generation process can stop when the bisimulation fails, avoiding unnecessary state generation~\cite{fernandez_OntheflyVerificationFinite_1992,almeida_DecidingKATHoare_2012}.

\begin{figure}[t]
  \begin{subfigure}{0.45\textwidth}
     \[\begin{aligned}
      & \colorIfThen{t_{1}}{p;} \\[-0.5ex]
      & \colorElse{\colorLabel{l}; \\[-0.5ex]
      & \quad\colorIfThen{t_{2}}{p; \colorGoto{l};} \\[-0.5ex]
      & \quad\colorElseIf{t_{1}}{\colorDiverge}} \\[-0.5ex] 
      & \colorRet;\\[-3ex]
    \end{aligned}\]
  \caption{}
  \end{subfigure} 
  \hfil
  \begin{subfigure}{0.45\textwidth}
    \centering
    \begin{tikzpicture}
      \node (init) {};
      \node[textstate] (u0) [right=5mm of init] {\(u_0'\)};
      \node[textstate] (u1) [right=of u0] {\(u_1'\)};
      \node[textstate] (u2) [right=of u1] {\(u_2'\)};
      \node (u0ret) [above=6mm of u0] {\(\retc\)};
      \node (u1ret) [above=3mm of u1] {\(\retc\)};
      \node (u2ret) [below=3mm of u2] {\(\retc\)};
      \draw[->] (init) edge (u0);
      \draw[->] (u0) edge[output-edge] node[left, align=center] {\(\overline{t_1} t_2\)\\\(\overline{t_1}\overline{t_2}\)} (u0ret);
      \draw[->] (u1) edge[output-edge] node[left] {\(\overline{t_1}\overline{t_2}\)} (u1ret);
      \draw[->] (u2) edge[output-edge] node[right] {\(\overline{t_1}\overline{t_2}\)} (u2ret);
      \draw[->] (u0) edge node[below, align=center] {\(t_1 t_2 \mid p\) \\ \(t_1 \overline{t_2} \mid p\)} (u1);
      \draw[->] (u1) edge node[above, align=center] {\(t_1 t_2 \mid p\) \\ \(\overline{t_1} t_2 \mid q\)} (u2);
      \draw[->] (u2) edge[out=45,in=100,loop,looseness=4] node[above, align=center] {\(t_1 t_2 \mid p\) \\ \(\overline{t_1} t_2 \mid p\)} (u2);
    \end{tikzpicture}
    \caption{}\label{fig:aut-err-low-level-equiv-cfgkat-prog}
  \end{subfigure}
  % \begin{subfigure}{0.45\textwidth}
  %   \[\begin{tikzcd}[row sep=1cm, column sep=1cm]
  %     & {} \\
  %     \retc & {u_0'} & \retc \\
  %     & {u_1'} & \retc \\
  %     & {u_2'} & \retc
  %     \arrow[from=1-2, to=2-2]
  %     \arrow["\overline{t_1}t_2"', Rightarrow, from=2-2, to=2-1]
  %     \arrow["\overline{t_1}\overline{t_2}", Rightarrow, from=2-2, to=2-3]
  %     \arrow["t_1t_2 \mid p"', shift right, from=2-2, to=3-2]
  %     \arrow["t_1 \overline{t_2} \mid p", shift left, from=2-2, to=3-2]
  %     \arrow["\overline{t_1}\overline{t_2}"', Rightarrow, from=3-2, to=3-3]
  %     \arrow["\overline{t_1}t_2 \mid p", shift left, from=3-2, to=4-2]
  %     \arrow["t_1t_2 \mid p"', shift right, from=3-2, to=4-2]
  %     \arrow["\overline{t_1}t_2 \mid p"', from=4-2, to=4-2, loop, in=310, out=230, distance=15mm]
  %     \arrow["t_1t_2 \mid p"', from=4-2, to=4-2, loop, in=300, out=240, distance=5mm]
  %     \arrow["\overline{t_1}\overline{t_2}"', Rightarrow, from=4-2, to=4-3]
  %   \end{tikzcd}\]
  %   \caption{The automata for the erroneous program in~\Cref{fig:err-low-level-equiv-cfgkat-prog}}\label{fig:aut-err-low-level-equiv-cfgkat-prog}
  % \end{subfigure}
  \caption{On the left, an erroneous version of the low-level program that is not trace equivalent to~\Cref{fig:high-level-equiv-cfgkat-prog}, where the \(t_{1} \lor t_{2}\) in the first if-statement is mistaken as \(t_{1}\). On the right, the corresponding automaton (with omitted reject transitions).}\label{fig:err-low-level-equiv-cfgkat-prog}\vspace{-.4cm}
\end{figure}

The normalization algorithm needed in \gkat equivalence checking requires iterating through the entire automaton \emph{prior to bisimulation checking}, and thus forces the generation of the entire automaton, even when a counter-example might have been encountered very early in the bisimulation.
For instance, consider the program in~\Cref{fig:err-low-level-equiv-cfgkat-prog}, which erroneously uses \(t_{1}\) as the condition of the first if-statement, instead of \(t_{1} \lor t_{2}\) in~\Cref{fig:low-level-equiv-cfgkat-prog}.
This mistake makes the program 
in~\Cref{fig:err-low-level-equiv-cfgkat-prog} no longer equivalent to the one in~\Cref{fig:high-level-equiv-cfgkat-prog}.
Indeed, we can conclude the inequivalence of their respecitve automata, presented in~\Cref{fig:aut-err-low-level-equiv-cfgkat-prog,fig:aut-high-level-equiv-cfgkat-prog}; by examining the \(t_{1}\overline{t_{2}}\) transition from the start states \(s_0\) and \(u_0'\).
To show that these two automata are not trace equivalent, we notice that the start state \(u_{0}'\) accepts the atom \(t_{1}\overline{t_{2}}\), which implies that \(t_{1}\overline{t_{2}}\) is a trace of \(u_{0}'\). 
In contrast, the start state \({\color[RGB]{27,94,32}s_{0}}\) executes the action \(p\) on atom \(t_{1} \overline{t_{2}}\), excluding \(t_{1}\overline{t_{2}}\) to be a trace of \({\color[RGB]{27,94,32}s_{0}}\).
However, \textbf{despite this inequality being witnessed by two transitions from the start state, the original algorithm still requires generating the entire automata} in order to perform normalization.

In this paper, we introduce both a naive and a symbolic algorithm that avoid this issue. The non-symbolic version of our trace-equivalence algorithm, presented in~\Cref{alg:nonsymb-bisim}, prevents generating the remaining automaton upon encountering a counter-example.
\textbf{This efficiency improvement is achieved by only invoking normalization, or dead-state detection, when necessary.}
In~\Cref{alg:nonsymb-bisim}, the primary checks are listed on lines 8-13, where we attempt to verify the bisimulation conditions for each atom, and cache all the checked states using a union-find data structure, inspired by previous bisimulation works~\cite{smolka_GuardedKleeneAlgebra_2020,hopcroft_LinearAlgorithmTesting_1971}. 
When the bisimulation checks fail, the algorithm invokes dead state detection on the mismatched states, via the \(\algIsDead\) function.
For example, the condition on line 10 in~\Cref{alg:nonsymb-bisim} states when \(s \transvia{\alpha \mid p} s'\) yet \(u\transRej{\alpha}\), then the \(s\) and \(u\) are trace-equivalent only if \(s'\) is dead.
In addition to the bisimulation check, this algorithm can sometimes take a shortcut, as detailed on lines 4-7: if one of the states \(s, u\) is known to be dead, then \(s\) and \(u\) are trace equivalent if and only if the other state is also dead; saving us from further bisimulation checks.

\begin{figure}[t]
    \centering
    \begin{algorithmic}[1]
        \Function{equiv}{$s, u$}
        \If {\Call{rep}{$s$} = \Call{rep}{$u$}} {\Return true} \EndIf
        \State \Call{union}{s, u};
        \If {\Call{knownDead}{$s$}} 
           \Return \Call{isDead}{$u$}
        \EndIf
        \If {\Call{knownDead}{$u$}} 
           \Return \Call{isDead}{$s$}
        \EndIf
        \State\Return \( \forall  \alpha \in \At\)
        \IndentState
          \(s \transOut{\alpha} \retc \text{ iff } u \transOut{\alpha} \retc \mathrel{\&\!\&}\)
        \IndentState
          \(s \transvia{\alpha \mid p} s' \text{ and } u \transRej{\alpha} \text{ implies } \Call{isDead}{s'} \mathrel{\&\!\&}\)
        \IndentState
          \(u \transvia{\alpha \mid p} u' \text{ and } s \transRej{\alpha} \text{ implies } \Call{isDead}{u'} \mathrel{\&\!\&}\)
        \IndentState
          \(s \transvia{\alpha \mid p} s' \text{ and } u \transvia{\alpha \mid q} u' \text{ and } p \neq q \text{ implies } \Call{isDead}{s'} \text{ and } \Call{isDead}{u'} \mathrel{\&\!\&}\)
        \IndentState
          \(s \transvia{\alpha \mid p} s' \text{ and } u \transvia{\alpha \mid p} u' \text{ implies } \Call{equiv}{s', u'}\)
        \EndFunction
    \end{algorithmic}
    \caption{Non-symbolic on-the-fly bisimulation algorithm. 
    The call \algIsDead(\(s\)) first checks whether \(s\) is in the set of cached ``known-dead states'' (\algKnownDead(\(s\))): if so, returns true, else it runs a DFS to see if the state \(s\) can reach any accepting state; if \(s\) cannot reach an accepting state, all the reachable states of \(s\) are cached as known-dead states.
    The call \algUnion{}(\(s, u\)) links \(s, u\) to the same representative in the union-find object, and \algRep(\(s\)) returns the representative of \(s\) from said union-find object.
    ``iff'', ``and'', ``implies'' are the respective logical operators on Booleans, and \(\&\!\&\) is the logical conjunction.}\label{alg:nonsymb-bisim}\vspace{-.4cm}
\end{figure}

The simple change underlying the new algorithm has big gains as we can now benefit from on-the-fly generation of automata, and terminate immediately after a counter-example.
We revisit the automata in~\Cref{fig:aut-equiv-cfgkat-prog,fig:aut-err-low-level-equiv-cfgkat-prog} using this new algorithm, and assume these automata are generated on-the-fly. 
If the first explored atom is \(t_{1} \overline{t_{2}}\), then the transitions from the start states (${\color[RGB]{27,94,32}s_{0}} \transvia{t_{1}\overline{t_{2}} \mid p} {\color{blue}s_{1}}$ and $u_{0}'  \transOut{t_{1}\overline{t_{2}}} \retc$) will immediately violate the first condition on line 9 of~\Cref{alg:nonsymb-bisim}, allowing us to short-circuit the algorithm and conclude that the input programs are not trace equivalent, without generating any other states and transitions.

\begin{remark}\label[remark]{rem:symb-check-eps-first}
  In the previous example, we assumed the algorithm checks atom \(t_{1}\overline{t_{2}}\) first. 
  However, this assumption is not necessary for the symbolic case, since all the atoms leading to immediate acceptance are grouped into Boolean expressions; allowing the algorithm to always check the terminating atoms first. 
\end{remark}

Besides early termination on negative cases, our algorithm can offer benefits in the positive cases as well, since \(\algIsDead\) only checks the reachable states of the discrepancy. 
For example, consider the automata in~\Cref{fig:aut-high-level-equiv-cfgkat-prog,fig:aut-low-level-equiv-cfgkat-prog}:
the original algorithm iterates through all $7$ states in both automata for normalization, then checks bisimulation on the $6$ live states to determine their equivalence.
Our algorithm only iterates through all the $7$ states once, including one invocation of \algIsDead{} on \({\color{red}s_{3}}\), nearly halving the complexity to confirm equivalence.
% Specifically, the trace-equivalence algorithm never invokes \(\algIsDead(s)\) on more than one live state; whereas the original algorithm requires iterating through all the live states twice, first for normalization, then for bisimulation. 
% Indeed, \(\algIsDead(s)\) will always halt the equivalence-checking algorithm when it returns false. 
% Similar to the original algorithm, the new dead state detection algorithm never checks the liveness of any state twice, thanks to the caching of dead state.

% To demonstrate the efficiency of our new algorithm in positive cases, 

In the extreme case, when the two input programs generate bisimilar \gkat automata~\cite{smolka_GuardedKleeneAlgebra_2020}, the bisimulation condition check will always succeed, preventing the invocation of dead-state checking in our algorithm.
Thus, our algorithm only needs one iteration of the automata to confirm the equivalence of bisimilar automata, whereas the original algorithm requires two passes: one for normalization, and the next for bisimulation.
This improvement is crucial for several applications, as many program-transformation algorithms often produce output program that is bisimilar to the input~\cite{erosa_TamingControlFlow_1994,kozen_BohmJacopiniTheorem_2008} i.e. they do not discard actions that will eventually diverge.
Hence, our algorithm will bring a significant performance improvement when validating the correctness of these transformations, which involves validating the equivalence of the input and output program.

\begin{remark}\label{rem:check-inf-trace-equiv}
  If the functions \(\algIsDead\) and \(\algKnownDead\) are substituted with the constant false function, the algorithm presented in~\Cref{alg:nonsymb-bisim} will detect a finer notion of equivalence, namely infinite-trace equivalence (coinciding with bisimularity)~\cite{schmid_GuardedKleeneAlgebra_2021}, which accounts for actions in non-terminating executions. 
  For instance, this semantics can distinguish between expressions such as {\colorWhile{\texttt{true}}{p}} and {\colorWhile{\texttt{true}}{q}}. 
  This modification is also applicable to the symbolic algorithm in~\Cref{alg:symb-bisim}, resulting in a efficient algorithm for infinite-trace equivalence. 
  However, we do not investigate these generalizations in our evaluation, as the existing tools we compared against do not support infinite-trace equivalence.
\end{remark}

\subsection{Symbolic Approach For \gkat Automata}\label{sec:symb-aut-overview}

The transitions of both \gkat and KAT automata are indexed by atoms; and since atoms are truth assignments to \emph{all the primitive tests}, (G)KAT automata obtained from the original constructions can often blow-up exponentially with the increase of primitive tests.
In fact, this problem is noted in several previous works~\cite{smolka_GuardedKleeneAlgebra_2020,zhang_CFGKATEfficientValidation_2025}, and both works hypothesized that symbolic techniques can aid in this regard. Consider the transitions from \({\color[RGB]{27,94,32}s_{0}}\) to \({\color{blue}s_{1}}\) in~\Cref{fig:aut-high-level-equiv-cfgkat-prog}, as shown on 

\begin{wrapfigure}{r}{2cm}
  \vspace{-8mm}
  \begin{tikzpicture}
    \node[textstate,align=left,
    draw=gray, dashed, rounded corners=5pt, thick, draw opacity=0.5,
    label={[anchor=south west]north west:\(T = \{t_1, t_2\}\)}] (origTrans)
    {\({\color[RGB]{27,94,32}s_{0}} \xrightarrow{t_{1} t_{2} \mid p} {\color{blue}s_{1}}\)\\
    \({\color[RGB]{27,94,32}s_{0}} \xrightarrow{t_{1} \overline{t_{2}} \mid p} {\color{blue}s_{1}}\)
    };
    \node[textstate,align=left,below=1.3cm of origTrans,
    draw=gray, dashed, rounded corners=5pt, thick, draw opacity=0.5,
    label={[anchor=north west]south west:\(T = \{t_1, t_2, t_3\}\)}] (t3Trans)
    {
  \({\color[RGB]{27,94,32}s_{0}} \xrightarrow{t_{1} t_{2} t_{3} \mid p} {\color{blue}s_{1}} \)\\
  \({\color[RGB]{27,94,32}s_{0}} \xrightarrow{t_{1} t_{2} \overline{t_{3}} \mid p} {\color{blue}s_{1}}\)\\
  \({\color[RGB]{27,94,32}s_{0}} \xrightarrow{t_{1} \overline{t_{2}} t_{3} \mid  p} {\color{blue}s_{1}}\)\\
  \({\color[RGB]{27,94,32}s_{0}} \xrightarrow{t_{1} \overline{t_{2}} \overline{t_{3}} \mid  p}{\color{blue}s_{1}}\)
    };

  \draw[->, decorate, decoration={snake, amplitude=2px, segment length=3mm}]
      (origTrans.230) -- node[right=1mm, align=left,node font=\scriptsize]{adding\\unused \(t_3\)} (t3Trans.115);
  \end{tikzpicture}
  \vspace{-12mm}
\end{wrapfigure} 
\noindent  the right. These transitions correspond to the first if statement in~\Cref{fig:high-level-equiv-cfgkat-prog}, namely \(\comIfThen{t_{1}}{p}\): when \(t_{1}\) is true, the state \({\color[RGB]{27,94,32}s_{0}}\) will transition to   \({\color{blue}s_{1}}\) and execute \(p\), regardless of the truth value of \(t_{2}\). However, both transitions have to include \(t_{2}\) in the atom that labels them. 
Even worse, the number of transitions from \({\color[RGB]{27,94,32}s_{0}}\) to \({\color{blue}s_{1}}\) will double when adding an unused primitive test \(t_{3}\), because atoms are defined to carry the truth assignment of \emph{every} primitive test.
This problem is not exclusive to \gkat and some solutions have been explored. 
Indeed, the seminal work of~\citet{pous_SymbolicAlgorithmsLanguage_2015} introduced an alternative presentation of \kat automata using Binary Decision Diagrams (BDDs)~\cite{bryant_BDD_1986}, which was also generalized to NetKAT~\cite{moeller_KATchFastSymbolic_2024}. 
However, the construction of Pous requires the use of BDDs when comparing Boolean expressions, and we aim to provide a more modular approach by preserving the Boolean expressions.
Our symbolic automata are designed to carry Boolean formulae on the transitions, allowing the algorithm to leverage arbitrary (UN)SAT solvers, including BDD-based solvers and other more efficient ones.

% Notably, we do \emph{not} combine transitions to the same state with same action, like the following:
% \begin{mathpar}
%   {\color[RGB]{27,94,32}p +_{t_{1}} 1 \seq \mathrm{loop}} \transvia{ \neg  t_{1} \land t_{2} \mid p} {\color{blue}\mathrm{loop}} \and 
%   {\color[RGB]{27,94,32}p +_{t_{1}} 1 \seq \mathrm{loop}} \transvia{t_{1} \mid p} {\color{blue}\mathrm{loop}} \and
%   \not \longrightsquigarrow  \and
%   {\color[RGB]{27,94,32}p +_{t_{1}} 1 \seq \mathrm{loop}} \transvia{( \neg  t_{1} \land t_{2}) \lor t_{1} \mid p} {\color{blue}\mathrm{loop}}
% \end{mathpar}
% We sometimes say that a transition is \emph{blocked} when the Boolean expression enabling such transition is equivalent to falsehood, like
% \[{\color{red}t_{1} \seq \mathrm{loop}} \transvia{0 \mid p} {\color{blue}\mathrm{loop}},\]
% in~\Cref{fig:symb-gkat-aut-equiv-cfgkat-prog}.
% We usually keep these transitions in our presentations for clarity, but we remove these blocked transitions upon generation in our implementation for efficiency, as detailed in ~\Cref{sec:implementation}.
% Indeed, the \(\color{blue}\mathrm{loop}\) have the same transition \(s_{2}\) in~\Cref{fig:aut-high-level-equiv-cfgkat-prog}: 

\begin{figure}[t]
  \begin{subfigure}{0.45\textwidth}
    \centering
    \begin{tikzpicture}
      \node (init) {};
      \node[textstate] (s0) [right=3mm of init] {\({\color[RGB]{27,94,32}(p +_{t_1} 1) \seq \mathrm{loop}}\)};
      \node[textstate] (s2) [right=of s0] {\({\color{blue}\mathrm{loop}}\)};
      \node[textstate] (s3text) [below=7mm of s2] {\({\color{red}(t_1 \land \overline{t_2}) \seq \mathrm{loop}}\)};  
      \node (s3) [fit=(s3text)] {};
      \node (s0ret) [below=6mm of s0] {\(\retc\)};
      \node (s2ret) [right=8mm of s2] {\(\retc\)};
      \draw[->] (init) edge (s0);
      \draw[->] (s0) edge[output-edge] node[auto] {\(\overline{t_1} \land\overline{t_2}\)} (s0ret);
      \draw[->] (s2) edge[output-edge] node[auto] {\(\overline{t_1} \land \overline{t_2}\)} (s2ret);
      \draw[->] (s0) edge node[above, align=left] {\(t_1 \mid p\) \\ \(\overline{t_1} \land t_2 \mid p\)} (s2);
      \draw[->] (s2) edge[out=125,in=55,loop,looseness=3] node[above] {\(t_2 \mid p\)} (s2);
      \draw[->] (s2) edge[bend left] node[right] {\(t_1 \land \overline{t_2} \mid q\)} (s3);
      \draw[->] (s3) edge[bend left] node[left] {\(0 \mid p\)} (s2);
      \draw[->] (s3) edge[out=150,in=-150,loop,looseness=3] node[below left] {\(t_1 \land \overline{t_2} \mid q\)} (s3);
    \end{tikzpicture}
    \caption{Symbolic \gkat automaton.}\label{fig:symb-gkat-aut-equiv-cfgkat-prog}
  \end{subfigure}
  % \begin{subfigure}{0.45\textwidth}
  %   \[\begin{tikzcd}[row sep=1.3cm,column sep = 0.75cm]
  %   {} & {\color[RGB]{27,94,32}(p +_{t_1} 1) \seq \mathrm{loop}} & \retc \\
  %     & {\color{blue}\mathrm{loop}} & \retc \\
  %     & {\color{red}(t_1 \land \overline{t_2}) \seq \mathrm{loop}}
  %     \arrow[from=1-1, to=1-2]
  %     \arrow["\overline{t_1} \land \overline{t_2}", Rightarrow, from=1-2, to=1-3]
  %     \arrow["t_1 \mid p", shift left, curve={height=-6pt}, from=1-2, to=2-2]
  %     \arrow["\overline{t_1} \land t_2 \mid p"', shift right, curve={height=6pt}, from=1-2, to=2-2]
  %     \arrow["t_2 \mid p", from=2-2, to=2-2, loop, in=150, out=210, distance=5mm]
  %     \arrow["\overline{t_1} \land \overline{t_2}", Rightarrow, from=2-2, to=2-3]
  %     \arrow["t_1 \land \overline{t_{2}} \mid q", shift right, curve={height=-6pt}, from=2-2, to=3-2]
  %     \arrow["0 \mid p", shift left, curve={height=-6pt}, from=3-2, to=2-2]
  %     \arrow["t_1 \land \overline{t_2} \mid q", shift left=4, from=3-2, to=3-2, loop, in=145, out=215, distance=5mm]
  %   \end{tikzcd}\]
  %   \caption{The symbolic \gkat automaton for \gkat program in~\Cref{fig:high-level-equiv-cfgkat-prog} (ignoring the final \(\comRet\) statement). 
  %   Both the program and Boolean expressions are subjected to standard optimizations.}\label{fig:symb-gkat-aut-equiv-cfgkat-prog}
  % \end{subfigure}
  \hfil
  \begin{subfigure}{0.45\textwidth}
    \centering
    \begin{tikzpicture}[node distance=10mm]
      \node (init) {};
      \node[textstate] (start) [right=3mm of init] {\({\color[RGB]{27,94,32}(p +_{t_1} 1) \seq \mathrm{loop}}\)};
      \node[textstate] (t1start) [right=4mm of start] {\({\color{teal} t_1}\)};
      \node[textstate] (t2start) [below=of t1start] {\({\color{teal} t_2}\)};
      \node[textstate] (loop) [right=5mm of t1start] {\(\color{blue} \mathrm{loop}\)};
      \node[textstate] (t1loop) [below=of loop] {\(\color{teal} t_1\)};
      \node[textstate] (dead) [below left=3mm and 5mm of t1loop] {\(\color{red} (t_1 \land \overline{t_2}) \seq \mathrm{loop}\)};
      \node[textstate] (t2loop) [right=5mm of dead] {\(\color{teal} t_2\)};
      \node (startRet) [below=4mm of start] {\(\retc\)};
      \node (loopRet) [right=5mm of loop] {\(\retc\)};
      \draw[->] (init) edge (start);
      \draw[->] (start) edge[output-edge] node[left] {\(\overline{t_1} \land\overline{t_2}\)} (startRet);
      \draw[->] (loop) edge[output-edge] node[above] {\(\overline{t_1} \land \overline{t_2}\)} (loopRet);
      \draw[->] (start) edge node[auto] {\(p\)} (t1start);
      \draw[->] (t1start) edge[dashed] (t2start);
      \draw[->] (t1start) edge (loop);
      \draw[->] (t2start) edge[bend left] (loop);  
      \draw[->] (loop) edge[bend left] node[above] {\(p\)} (t2start);
      \draw[->] (loop) edge node[auto] {\(q\)} (t1loop); 
      \draw[->] (t1loop) edge (t2loop); 
      \draw[->] (t2loop) edge[dashed] (dead);  
      \draw[->] (dead) edge node[auto,swap] {\(q\)} (t1loop);
    \end{tikzpicture}
    \caption{SymKAT~\cite{pous_SymbolicAlgorithmsLanguage_2015} automaton.}\label{fig:symkat-aut-equiv-cfgkat-prog}
  \end{subfigure}
  % \begin{subfigure}{0.45\textwidth}
  %   \[\begin{tikzcd}[row sep=0.75cm,column sep = 0.75cm]
  %     \retc & {\color[RGB]{27,94,32}p+_{t_1} 1 \seq \mathrm{loop}} \\
  %     {\color{teal}t_2} & {\color{teal}t_1} \\
  %     \retc & {\color{blue}\mathrm{loop}} \\
  %     {\color{teal}t_2} & {\color{teal}t_1} \\
  %     {\color{red}(t_1 \land \overline{t_2}) \seq \mathrm{loop}}
  %     \arrow["{{\overline{t_1} \land \overline{t_2}}}"', Rightarrow, from=1-2, to=1-1]
  %     \arrow["p", from=1-2, to=2-2]
  %     \arrow[shift right, from=2-1, to=3-2]
  %     \arrow[shift right, curve={height=6pt}, dotted, from=2-2, to=2-1]
  %     \arrow[from=2-2, to=3-2]
  %     \arrow["p"', shift right, from=3-2, to=2-1]
  %     \arrow["{{\overline{t_1}\land\overline{t_2}}}", Rightarrow, from=3-2, to=3-1]
  %     \arrow["q", from=3-2, to=4-2]
  %     \arrow[dotted, from=4-1, to=5-1]
  %     \arrow[from=4-2, to=4-1]
  %     \arrow["q"', from=5-1, to=4-2]
  %   \end{tikzcd}\]
  %   \caption{The SymKAT automata~\cite{pous_SymbolicAlgorithmsLanguage_2015} for program in~\Cref{fig:high-level-equiv-cfgkat-prog} (ignoring the final \(\comRet\) statement), through the standard embedding of \gkat into KAT~\Cite{smolka_GuardedKleeneAlgebra_2020}.}\label{fig:symkat-aut-equiv-cfgkat-prog}
  % \end{subfigure}
  \caption{A comparison between symbolic \gkat and SymKAT automata using the example in~\Cref{fig:high-level-equiv-cfgkat-prog}. We use ``\(\mathrm{loop}\)'' to refer to the while loop in~\Cref{fig:high-level-equiv-cfgkat-prog}.}\vspace{-.4cm}
\end{figure}

In~\Cref{fig:symb-gkat-aut-equiv-cfgkat-prog}, we present the symbolic GKAT automata generated from the program in~\Cref{fig:high-level-equiv-cfgkat-prog}, using symbolic derivative described in~\Cref{sec:aut-construction-cfgkat}.
Compare to the \gkat automaton generated by the original algorithm in~\Cref{fig:aut-high-level-equiv-cfgkat-prog}, our new automaton is more compact in several ways. Firstly, two trace-equivalent states \({\color{blue}s_{1}}\) and \({\color{blue}s_{2}}\) are merged into one state \({\color{blue}\mathrm{loop}}\).
This efficiency is due to the derivative construction generally producing more compact automata than other original constructions, e.g. Thompson's construction~\cite{owens_RegularexpressionDerivativesReexamined_2009}.
Second, \textbf{multiple transitions between the same states are consolidated into Boolean}
\begin{wrapfigure}{l}{4cm}
  \vspace{-9.5mm}
  \begin{tikzpicture}
    \node[textstate,align=left,
    draw=gray, dashed, rounded corners=5pt, thick, draw opacity=0.5,
    label={[anchor=south west]north west:concrete transitions}] (conc)
    {
    \({\color[RGB]{27,94,32}s_{0}} \transvia{t_{1} t_{2} \mid p} {\color{blue}s_{1}}\) 
    \\
    \({\color[RGB]{27,94,32}s_{0}} \transvia{t_{1} \overline{t_{2}} \mid p} {\color{blue}s_{1}}\)
    };
    \node[textstate,align=left,below=1cm of conc,
    draw=gray, dashed, rounded corners=5pt, thick, draw opacity=0.5,
    label={[anchor=north west]south west:symbolic transitions}] (symb)
    {
    \({\color[RGB]{27,94,32}(p +_{t_{1}} 1) \seq \mathrm{loop}} \transvia{t_{1} \mid p} {\color{blue}\mathrm{loop}}\)
    };

    \draw[->, thick]
      (symb) -- node[auto, align=left,node font=\scriptsize]{concretize} (conc);
  \end{tikzpicture}
  \vspace{-10mm}
\end{wrapfigure}
\textbf{expressions}; intuitively, a Boolean expression represents the collection of atoms that satisfies them.
For example, the two transitions from \(\color[RGB]{27,94,32} s_{0}\) to \(\color{blue}s_{1}\) in~\Cref{fig:aut-high-level-equiv-cfgkat-prog} is merged into one in~\Cref{fig:symb-gkat-aut-equiv-cfgkat-prog}, as shown on the left, where \(t_1 t_2\) and \(t_1 \overline{t_2}\) are the two atoms that satisfy the Boolean expression \(t_1\).
Unlike the non-symbolic case, the addition of an unused primitive test \(t_{3}\) will not enlarge the symbolic automata; thereby avoiding the exponential size blowup of the \gkat automata.

To decide the trace equivalence of symbolic \gkat automata, we lift the non-symbolic algorithm in~\Cref{alg:nonsymb-bisim} to the symbolic setting, which is presented in~\Cref{alg:symb-bisim}, where \(\rho(s)\) contains all the atoms that are rejected by \(s\), and \(\epsilon(s)\) contains all the Boolean expression that \(s\) accepts (see~\Cref{sec:symb-aut-detail} for details).
As mentioned in~\Cref{rem:symb-check-eps-first}, the algorithm will immediately conclude that \(s\) is not trace-equivalent to \(u\) when there exists an atom \( \alpha \), s.t. \(s \transOut{\alpha} \retc\) and \(u \transvia{\alpha \mid p} u'\); since that would imply \(\alpha \in \bigvee  \varepsilon(s)\) but \( \alpha \notin \bigvee \varepsilon(u)\), thus \(\bigvee \varepsilon(s) \not\equiv \bigvee \varepsilon(u)\), contradicting the first requirement on line 9 of the symbolic algorithm.

\begin{figure}[t]
    \begin{algorithmic}[1]
        \Function{equiv}{$s, u$}
        \If {\Call{rep}{$s$} = \Call{rep}{$u$}} {\Return true;} \EndIf
        \State\Call{union}{$s$, $u$};
        \If {\Call{knownDead}{$s$}} 
            \Return \Call{isDead}{$u$}
        \EndIf
        \If {\Call{knownDead}{$u$}} 
            \Return \Call{isDead}{$s$}
        \EndIf
        \State\Return {
        \IndentState
          \( \bigvee \varepsilon(s) \equiv \bigvee \varepsilon(u) \mathrel{\&\!\&}\)
        \IndentState 
          \(s \transvia{b \mid p} s' \text{ and } (b \land  \rho (u)) \not\equiv 0 \text{ implies \Call{isDead}{$s'$}} \mathrel{\&\!\&}\)
        \IndentState
          \(u \transvia{b \mid p} u' \text{ and } (b \land  \rho (s)) \not\equiv 0 \text{ implies \Call{isDead}{$u'$}} \mathrel{\&\!\&}\)
        \IndentState 
          \(s \transvia{b \mid p} s' \text{, } u \transvia{a \mid q} u' \text{, } b \land a \not\equiv 0 \text{, and } p \neq q \text{ implies \Call{isDead}{$s'$}, \Call{isDead}{$u'$}} \mathrel{\&\!\&}\)
        \IndentState 
          \(s \transvia{b \mid p} s' \text{ and } u \transvia{a \mid p} u' \text{ and } b \land a \not\equiv 0 \text{ implies } \text{\Call{equiv}{$s', u'$}}\)
        }
        \EndFunction
    \end{algorithmic}
    \caption{Symbolic on-the-fly bisimulation algorithm, lifting the algorithm in~\Cref{alg:nonsymb-bisim} to symbolic \gkat automata.
    We use \algKnownDead{} and \algIsDead{} to denote dead-state detection: \algIsDead($s$) runs a DFS on all the \emph{non-blocking transitions} (\Cref{rem:blocked-transition}) from \(s\), and return if \(s\) can reach an accepting state \(u\), which is checked by \( \bigvee  \varepsilon(u) \not\equiv 0\).
    }\label{alg:symb-bisim}\vspace{-.4cm}
\end{figure}

We also present the corresponding SymKAT automaton~\cite{pous_SymbolicAlgorithmsLanguage_2015} in~\Cref{fig:symkat-aut-equiv-cfgkat-prog}, in which all Boolean expressions are unfolded into BDDs during the generation process, creating several intermediate nodes marked by \({\color{teal} t_{1}}\) and \({\color{teal} t_{2}}\).
Our symbolic \gkat automaton preserves the Boolean expressions on transitions,  providing the flexibility to leverage external tools for their efficient representation and satisfiability solving.
In~\Cref{sec:implementation}, we detailed our implementation which offers orders-of-magnitude performance improvement over SymKAT, in both runtime and memory consumption using  miniSAT~\cite{_BooleworksLogicngrs_2025} and a state-of-the-art BDD solver~\cite{lewis_PclewisCuddsys_2025,somenzi_CUDDCUDecision_2018}.
% Additionally, on synthetic benchmarks with larger Boolean expressions as well as more if‑statements and while‑loops, MiniSAT~\cite{_BooleworksLogicngrs_2025} can outperform state‑of‑the‑art BDD implementations~\cite{lewis_PclewisCuddsys_2025,somenzi_CUDDCUDecision_2018}. 
% Conversely, control‑flow equivalence checking for CoreUtils~\cite{gnuprojectCoreutilsGNUCore2024} performs best when using BDDs. 
Notably, our approach enables the user to select the most efficient solver for specific tasks, and to employ high‑performance BDD implementations as black‑boxes---capabilities lacking in prior BDD‑based methods. 
% To ensure that this performance improvement is not solely due to implementation details, we also compared the performance of miniSAT~\cite{_BooleworksLogicngrs_2025} versus state-of-the-art BDD implementations~\cite{lewis_PclewisCuddsys_2025,somenzi_CUDDCUDecision_2018} as (UN)SAT solvers in our symbolic algorithm.
% In~\Cref{sec:imple-evaluation}, we find miniSAT is able to significantly outperform BDD in some cases; but for applications where BDD outperforms miniSAT, our tool can optionally use BDD as the solver, thanks to the flexibility of our symbolic representation.

% Beyond efficiency, \gkat also provides the option to reason about bisimulation equivalence~\cite{schmid_GuardedKleeneAlgebra_2021} by turning off dead state detection, i.e. replacing both \algIsDead{} and \algKnownDead{} with constant false.
% Indeed, bisimulation equivalence can be the desirable notion of equivalence in some applications.
% For example, consider the pair of programs in~\labelcref{exp:bisim-equiv-desirable}; SymKAT will not be able to distinguish these two programs, since KAT do not support infinite traces yet.

\myparagraph{Roadmap for the rest of the paper} 
In the current section, we illustrated our approaches through concrete examples. 
In the rest of the paper, we will delve deeper into the formalism, including the correctness guarantees of our symbolic and non-symbolic algorithm, their complexity, and the on-the-fly symbolic automata constructions for \cfgkat.
\textbf{\Cref{sec:correctness-and-complexity}} proves the termination, soundness, completeness, and complexity of our decision procedures, and compares it with prior approaches. 
% \textbf{\Cref{sec:aut-construction-gkat}} provides an on-the-fly algorithm to generate symbolic \gkat automata from \gkat expressions using derivatives. 
% This development yields both an efficient decision procedure for equivalence of \gkat expressions, and also a blueprint for other variants. 
\textbf{\Cref{sec:aut-construction-cfgkat}} provides an on-the-fly algorithm to generate symbolic \gkat automata from \cfgkat programs. 
Despite its complexity, we are able to give an efficient decision procedure for \cfgkat based on our framework.
\textbf{\Cref{sec:implementation}} provides more details on our implementations, and our evaluation results against the current tools to check \CFOrGKAT equivalences. 

\section{Correctness and Complexity}\label{sec:correctness-and-complexity}

% \begin{itemize}
%   \item General strategy: reducing symbolic to non-symbolic 
%   \item Correctness of non-symbolic: 
%   \begin{itemize}
%     \item define bisimulation equivalence
%     \item notice the algorithm exactly aligns with bisimulation equivalence
%     \item define trace semantics 
%     \item prove the correctness of trace equivalence
%   \end{itemize}
%   \item Correctness of symbolic case: 
%   \begin{itemize}
%     \item lowering: converting a symbolic automata back to non-symbolic case
%     \item correspondence with the correctness of the lowered automata: implies the correctness of both trace equivalence and bisimulation equivalence.
%   \end{itemize}
%   \item Complexity Non-symbolic case 
%   \begin{itemize}
%     \item in general it is EXPTIME to the size of primitive tests.
%     \item the worst case two pass
%     \item in special case can even be slower than original algorithm. 
%     \item but the special case is unrealistic.
%   \end{itemize}
%   \item Complexity Symbolic case 
%   \begin{itemize}
%     \item the algorithm is still unfortunately in PSPACE
%     \item compare to EXPSPACE of SymKAT, because we don't unfold Boolean expressions
%   \end{itemize}
% \end{itemize}

In this section, we will first provide some comparison with the previous decision procedure by~\citet{smolka_GuardedKleeneAlgebra_2020} and prove that both of our algorithms in~\Cref{alg:nonsymb-bisim,alg:symb-bisim} are sound and complete.
Additionally, as mentioned in~\Cref{rem:check-inf-trace-equiv}, these algorithms can also be adapted to check for infinite-trace equality~\cite{schmid_GuardedKleeneAlgebra_2021}, thus we prove the correctness of these variants as well.
Details of these correctness proofs can be found 
\ifFull
in~\Cref{sec:non-symb-corr-proof,sec:symb-corr-proof}.
\else   
in the full version~\cite{zhangOutrunningBigKATs2026}.
\fi
Finally, we also discuss the theoretical complexity of both the symbolic and non-symbolic algorithms; whereas the real-world performances of our implementations are presented in~\Cref{sec:imple-evaluation}.

\subsection{Comparison With Previous Decision Procedures}

In this section, we outline the tradeoffs and advantages of our \emph{non-symbolic} decision procedure v.s. the original procedures proposed by~\citet{smolka_GuardedKleeneAlgebra_2020}.

\myparagraph{Differences in complexity}
In the worst case, all calls to \(\algIsDead\) in our algorithm will iterate through the entire automata at most once, thanks to the caching of explored states; and the \(\algEquiv\) function will again iterate through all the states at most once, since explored states are cached via the union-find object~\cite{hopcroft_LinearAlgorithmTesting_1971}.
In the same vein, the original algorithm will also make two passes in the worst case: one for normalization and one for bisimulation. Unfortunately, 
\begin{wrapfigure}{l}{3cm}
  \vspace{-9mm}
  \centering
  \begin{tikzpicture}
    \node (initS) {};
    \node[textstate] (s0) [right=2mm of initS] {\({s_0}\)};
    \node[textstate] (s1) [right=6mm of s0] {\({s_1}\)};
    \node[textstate] (s2) [right=6mm of s1] {\({s_2}\)};
    \draw[->] (initS) edge (s0);
    \draw[->] (s0) edge node[auto,swap] {\(\alpha \mid p\)} (s1);
    \draw[->] (s1) edge node[auto,swap] {\(\alpha \mid p\)} (s2);
    \draw[->] (s2) edge[draw=red, bend right] node[auto,swap] {\(\color{red}\alpha \mid p\)} (s0);

    \node (initU) [below=8mm of initS] {};
    \node[textstate] (u0) [right=2mm of initU] {\({u_0}\)};
    \node[textstate] (u1) [right=6mm of u0] {\({u_1}\)};
    \node[textstate] (u2) [right=6mm of u1] {\({u_2}\)};
    \draw[->] (initU) edge (u0);
    \draw[->] (u0) edge node[auto] {\(\alpha \mid p\)} (u1);
    \draw[->] (u1) edge node[auto] {\(\alpha \mid p\)} (u2);
    \draw[->] (u2) edge[draw=red, bend left] node[auto] {\(\color{red}\alpha \mid q\)} (u0);
  \end{tikzpicture}
  \vspace{-12mm}
\end{wrapfigure}
when the automata have many dead states, the original algorithm can outperform our new algorithm.
Take the pair of \gkat automata on the left, these automata are indeed trace-equivalent, despite \(s_2\) and \(u_2\) executes different actions (transition marked in {\color{red}red}), since \emph{all their states are dead}.
The original algorithm~\cite{smolka_GuardedKleeneAlgebra_2020} will iterate through both automata (6 states) for the normalization, then conclude that all the states are dead; then in the bisimulation phase, only the start states \(s_0\) and \(u_0\) is checked. 
In total, the original algorithm performed 7 checks, including 6 states for normalization and 1 pair of states for bisimulation.
Yet, the new algorithm will first check the transition of the pairs \((s_0, u_0)\) and \((s_1, u_1)\), then because \(s_2\) and \(u_2\) executes different actions, failing the premise on line 12 of~\Cref{alg:nonsymb-bisim} and invoking \(\algIsDead(s_2)\) and \(\algIsDead(u_2)\).
Our new algorithm will proceed to iterate through both automata again, to confirm that \(s_2\) and \(u_2\) are dead.
Thus, the new algorithm performs \emph{9 checks}: 3 pairs of state during bisimulation, and 6 states during liveness checking; exceeding the \emph{7 checks} performed by the original algorithm.

When taking the number of primitive tests into account, our algorithm achieves PSPACE complexity, by storing one atom at a time, instead of storing the entire automaton for normalization.
This complexity represents a theoretical speedup over the original EXPSPACE algorithm given by~\citet[Section 5.3]{smolka_GuardedKleeneAlgebra_2020}.
Although it is known that the complexity of \gkat is within PSPACE using the decision procedure for \kat~\cite{cohen_ComplexityKleeneAlgebra_1996a}, leveraging Savitch's theorem; our algorithm is the \textbf{first concrete PSPACE algorithm} for \gkat equivalence.

However, the worst-case complexity does not necessarily reflect real-world performance. 
For example, our algorithm only explores the reachable states of the mismatched in bisimulation (see~\Cref{sec:on-the-fly-norm-overview}) instead of the entire automata, this often leads to fewer explored states in examples like~\Cref{fig:aut-equiv-cfgkat-prog}.
Additionally, many real-world transformations produce bisimilar programs, since it is a stronger equivalence than trace equivalence; in these cases, our algorithm skips dead state detection entirely, potentially halving the complexity of the original algorithm.

\myparagraph{Differences in implementation}
Despite its simple description, the normalization procedure is comparatively complex to implement.
Consider implementation by \citet{zhang_2024_13938565}: first, the algorithm makes one pass to identify all the accepting states, and produces a backward map, which maps every state \(s\) to all states that transitions to \(s\); then with the backward map, the algorithm performs a reverse search from accepting states, marking the reachable states as live states; finally, the algorithm executes another forward pass to reroute all the transition to dead state, i.e. those are not marked as live, into rejection.

The backward map and additional passes are necessary because the normalization procedure requires the identification of \emph{all the live states}.
And since liveness is a global property (a property that cannot be determined by only inspecting a individual state and its transitions), it is unclear how to identify all the live/dead states without an additional pass of the automata.

Our algorithm, on the other hand, forgoes identifying the liveness of all the states, removing the need for a backward map and additional passes.
Specifically, \(\algIsDead\) only checks the liveness of one state at a time: it continues when detecting a dead state; halts \emph{the entire equivalence check} otherwise. 
Additionally, since \(\algIsDead(s)\) explore and stores all the reachable states from \(s\), if \(s\) is dead, we can mark all its reachable state as dead, allowing \(\algIsDead\) to iterate through the automata at most once, no matter how many times it is called.
The same property doesn't hold for live states: dead states can be reachable from live ones.
These intricate designs allow our simple algorithm to not only outperform the original algorithm in many real-world tasks, but also provide simpler implementations.

\myparagraph{Challenges in proving correctness}
However, this efficiency also introduces significant challenge in the correctness proof: although the intuition behind our algorithm is simple, our algorithm no longer \emph{explicitly} identifies a bisimulation on normalized \gkat automata, which underpins the correctness result of~\citet{smolka_GuardedKleeneAlgebra_2020}.
Indeed, proving our algorithm yields such bisimulation is challenging; instead, we show conditions in our algorithm defines a progression~\cite{pous_CompleteLatticesUpTo_2007}, and that trace equivalence is the maximal invariant of this progression. 
This approach not only provides a sufficiently strong induction hypothesis to establish the desired results, but also enables the use of up-to technique~\cite{pous_CompleteLatticesUpTo_2007}, an important framework to justify the use of union-find data-structure.
Additionally, we also resolved a significant challenge in applying the up-to technique: because its subtlety, the progression defined by our algorithm is no longer compatible with the transitive closure%
\ifFull%
, which we demonstrate in~\Cref{exp:trans-clos-incompatible}.
\else  
.
\fi
To solve this problem, we show that transitive closure satisfy a weaker condition than compatiblity, but nevertheless allows us to utilize up-to technique in our correctness proof.

\subsection{Correctness of the Non-symbolic Algorithm}

In this section, we offer a sketch of the correctness of the algorithm in~\Cref{alg:nonsymb-bisim}, to provide intuition. The complete proof can be found 
\ifFull
in~\Cref{sec:non-symb-corr-proof}. 
\else   
in the full version~\cite{zhangOutrunningBigKATs2026}.
\fi
\begin{definition}[Trace Semantics~\cite{smolka_GuardedKleeneAlgebra_2020}]\label{def:trace-sem}
  Fix a \gkat automaton \(A \triangleq (S, s_0, \zeta)\), the finite trace semantics \(\sem{s}\) for a state \(s \in S\) is defined as follows: 
  if \(s \transOut{\alpha} \retc\), then \(\alpha \in \sem{s}\); if \(s \transvia{\alpha \mid p} s'\) and \(w \in \sem{s'}\), then \(\alpha p \cdot w \in \sem{s}\). 
  Two states \(s\) and \(u\) are \emph{trace equivalent} when \(\sem{s} = \sem{u}\).
\end{definition}

Our proof uses coinduction up-to \cite{pous_CompleteLatticesUpTo_2007, bonchi_GeneralAccountCoinduction_2017a}.
First, we define an instance of progression~\cite{pous_CompleteLatticesUpTo_2007}: for a pair of relations over state sets \(R, R' \subseteq S \times U\), we say that \(R\) \emph{progresses} into \(R'\) when all the conditions from lines 9-13 in~\Cref{alg:nonsymb-bisim} holds for all atoms, except the last condition is replaced by \(s \transvia{\alpha \mid p} s'\) and \(u \transvia{\alpha \mid p} u'\) implies \((s', u') \in R'\).
We write \(R \transCorr{} R'\) when \(R\) progresses into \(R'\); and we say \(R\) is an \emph{invariant} when \(R \transCorr{} R\). 
Then we prove two core theorems (1) trace equivalence is the maximal invariant
\ifFull
in~\Cref{thm:trace-equiv-is-max-invari}
\fi
; (2) the equivalence closure \(e\) is ``sound'' with respect to \(\transCorr{}\)
\ifFull
in~\Cref{thm:equiv-closure-sound-trace-prog}
\fi, i.e. \(R \transCorr{} e(R)\) implies \(R\) is contained in an invariant, and thereby contained in the trace equivalence relation.

Then we show that the algorithm in~\Cref{alg:nonsymb-bisim} indeed identifies an invariant. 
We let \(R\) denote all the pairs of states that have been sent into the \(\algUnion\) function and let \(R_0\) denote the \(R\) at the start of the function \(\algEquiv\).
If \(\algEquiv(s, u)\) returns true, then we can derive \((s, u) \in e(R)\) and \((R \setminus R_0) \transCorr{} e(R)\) both hold after the program terminates.
Finally, we assume that \(\algEquiv\) starts with a fresh union-find object, \(R_0\) will be empty and \(R \transCorr{} e(R)\), implying \(R\) is contained in an invariant.
Then because \((s, u) \in R\), the states \(s\) and \(u\) are trace-equivalent.

\begin{theorem}[Correctness]\label{thm:non-symb-corr}
   \(\algEquiv\) in~\Cref{alg:nonsymb-bisim} always terminates when the input states are from finite \gkat automata; and the algorithm is \emph{sound and complete}, in the sense that
  (1)\ \(\algEquiv(s, u)\) returns true iff \(s\) and \(u\) are trace-equivalent;
  (2) when we replace \(\algIsDead\) and \(\algKnownDead\) with false, \(\algEquiv(s, u)\) returns true if and only if \(s\) and \(u\) are infinite-trace equivalent i.e. bisimilar.
\end{theorem}

Our non-symbolic algorithm improves upon the original \gkat decision procedure, while preserving correctness; but like the original algorithm it does not scale well with a large number of primitive tests.
To address this issue, we introduce a symbolic decision procedure, which is outlined in~\cref{sec:symb-aut-overview}.

% However, the speedup of the original algorithm occurs only when there are a \emph{significant number of dead states}. 
% Since the actions performed by these states can never lead to termination, they are essentially redundant in the context of finite trace equivalence. 
% Since a large number of unnecessary actions is rare in real-world codebases, we believe that the performance improvements offered by our algorithm, when there are fewer dead states, represent a worthwhile trade-off.

% \cheng{TODO: work PSPACE complexity into this section.}

% Despite its linear complexity with respect to the number of states, our non-symbolic algorithm remains EXPSPACE in relation to the number of primitive tests, as it must iterate through and store all the atoms associated with the transitions. 
% Consequently, the non-symbolic algorithm cannot scale to applications that require more than a limited number of primitive tests.

\subsection{Symbolic Automata and Its Equivalence Algorithm}\label{sec:symb-aut-detail}

Before we introduce the correctness proof, we first formally define \emph{symbolic \gkat automata}, where instead of giving a transition result for every atom, \textbf{we compact atoms that lead to the same results into Boolean expressions}.

\begin{definition}[Symbolic \gkat Automata]
  A \emph{symbolic \gkat automaton} over primitive actions \(\Sigma\) and primitive tests \(T\) consists of four components: a finite state set \(S\), a start state \(s_0 \in S\), an accepting function \( \varepsilon \), and a transition function \(\delta\) with the following types:
  \begin{align*}
     \varepsilon &: S \to \powSet{\BExp_T}; &
     \delta &: S \to \powSet{\BExp_T \times S \times \Sigma},
  \end{align*}
  where \(\powSet{X}\) represents the power set of \(X\).
  In addition, we require a \emph{disjointedness condition}: if we take any two distinct elements \(r_{1}, r_{2} \in \varepsilon(s) \cup \delta(s),\) then Boolean expressions within \(r_{1}, r_{2}\) need to be disjoint i.e. their conjunction is equivalent to falsehood under Boolean Algebra.
\end{definition}
We inherit the notation of non-symbolic \gkat automata: we write \(s \transOut{b} \retc\) when \(b \in \varepsilon(s)\) and \(s \transvia{b \mid p} s'\) when \((b, s', p) \in \delta(s)\).
In~\Cref{fig:symb-gkat-aut-equiv-cfgkat-prog}, we provided an example of symbolic \gkat automata, and the accepting function \(\varepsilon\) and transition function \(\delta\) for the state \({\color{blue}\mathrm{loop}}\) here as an example:
\begin{align*}
  % \varepsilon(p +_{t_{1}} 1 \seq \mathrm{loop}) & = \{ \neg  t_{1} \land \neg t_{2}\} &
  % \delta(p +_{t_{1}} 1 \seq \mathrm{loop}) & 
  %   = \{
  %     ( \neg  t_{1} \land t_{2}, \mathrm{loop}, p),
  %     (t_{1}, \mathrm{loop}, p),
  %   \}\\
  \varepsilon({\color{blue}\mathrm{loop}}) & = \{\overline{t_{1}} \land \overline{t_{2}}\} &
  \delta({\color{blue}\mathrm{loop}}) & = \{
      (t_{2}, {\color{blue}\mathrm{loop}}, p),
      (t_{1} \land \overline{t_{2}}, {\color{red}(t_{1} \land \overline{t_{2}}) \seq \mathrm{loop}}, q).
  \}
  % \\
  % \varepsilon(t_{1} \seq \mathrm{loop}) & = \{ \neg  t_{1} \land \neg t_{2}\} &
  % \delta(t_{1} \seq \mathrm{loop}) & = \{
  %   ( \neg  t_{1} \land t_{2}, \mathrm{loop}, p),
  %   (0, \mathrm{loop}, p)
  % \}
\end{align*}
As discussed in~\Cref{sec:symb-aut-overview}, \(s \transvia{b \mid p} s'\) represents that \(s \transvia{\alpha \mid p} s'\) for all \(\alpha\) satisfying \(b\); and \(s \transOut{b} \retc\) represents that \(s \transOut{\alpha} \retc\) for all \(\alpha\) satisfying \(b\), and the rest of the atoms lead to rejection.
We do not include the rejection function \(\rho\) in the signature of the symbolic automata, since it can be computed from \(\varepsilon\) and \(\delta\):
\[
  \rho(s) \triangleq \bigwedge \{\overline{b} \mid b \in \varepsilon(s)\} \cup \{\overline{b} \mid (b, s', p) \in \delta(s)\}.
\]

To give semantics to a symbolic \gkat automaton, we will first unfold it into a \emph{concrete \gkat automaton}, and then compute the semantics of said concrete automaton.
Indeed, the concretization process exactly defines the operation to unfold all Boolean expressions into atoms.

\begin{definition}[Concretization]\label{def:concretization}
  Given a symbolic \gkat automaton, with start state \(s_0 \in S\), accepting function \(\varepsilon\), and a transition function \(\delta\); then we can define a \gkat automaton which is the \emph{concrete} automaton, with the same state set \(S\) and start state \(s_0\) and the following transition function:
  \[\zeta_{\varepsilon, \delta}(s, \alpha) \triangleq \begin{cases}
    \retc & \text{if there exists } b \in \varepsilon(s), \alpha \text{ satisfies } b \\  
    (s', p) & \text{if there exists } (b, s', p) \in \delta(s), \alpha \text{ satisfies } b \\  
    \bot & \text{otherwise}
  \end{cases}\]
\end{definition}
Because of the disjointedness condition, the concretization of symbolic automata is well-defined: for all state \(s \in S\), there exists at most one Boolean expression \(b \in \varepsilon(s)\) or transition \((b, s', p) \in \delta(s)\) s.t. \(\alpha\) satisfies \(b\).

\begin{remark}[Determinacy]
  The concretization of any symbolic \gkat automaton, like every other \gkat automaton, will deterministically accept, reject, or transition, when given a state and an atom.
  Thus, despite the use of power set in its signature, symbolic GKAT automata are still \emph{deterministic}, making our derivative constructions in~\Cref{sec:aut-construction-cfgkat} more akin to the deterministic derivative by~\citet{brzozowski_DerivativesRegularExpressions_1964} than the nondeterministic version by~\citet{antimirov_PartialDerivativesRegular_1996}.
\end{remark}

\begin{remark}[Blocked Transition]\label{rem:blocked-transition}
  Transitions \(s \transvia{b \mid p} s'\) where \(b \equiv 0\) are called \emph{blocked}. 
  For example, \({\color{red}(t_1 \land t_2) \seq \mathrm{loop}} \transvia{0 \mid p} {\color{blue}\mathrm{loop}}\) in~\Cref{fig:symb-gkat-aut-equiv-cfgkat-prog} is blocked.
  Removing such transitions will not change the concrete automata, as there doesn't exist an atom \(\alpha \leq 0\).
  Thus, we do not include blocked transition in future drawings of automata and in our implementation (see \Cref{sec:imple-detail}).
\end{remark}

By definition, the equivalence of symbolic GKAT automata coincides with the equivalence of the concrete automata.
Therefore, the correctness of the symbolic algorithm (\Cref{alg:symb-bisim}) can be reduced into the correctness of the non-symbolic algorithm (\Cref{alg:nonsymb-bisim}).
Indeed, the strategy for soundness of completeness proof is by defining a \emph{symbolic progression}, which corresponds to the conditions checked on line 9-13 in the symbolic algorithm.
Then we show the correspondence between symbolic progression on the symbolic automata and progression on the concrete automata; this allows us to derive the same post-condition as the non-symbolic correctness on the concrete automata, which concludes the correctness.

\begin{theorem}[Correctness]\label{thm:symb-corr}
  The algorithm in~\Cref{alg:symb-bisim} always terminates if the inputs are from finite \gkat automata, and it is \emph{sound and complete}:
  (1) \(\algEquiv\) returns true if and only if the inputs are trace-equivalent in the concretization;
  (2) if \(\algIsDead\) and \(\algKnownDead\) in~\Cref{alg:symb-bisim} is replaced with false, then \(\algEquiv\) returns true if and only if the inputs are bisimilar in the concretization.
\end{theorem}

\myparagraph{Complexity analysis of symbolic algorithm}
Since the algorithm require solving both SAT and UNSAT problems, our algorithm in~\Cref{alg:symb-bisim} is technically in PSPACE with respect to the number of primitive tests, even though both SAT and UNSAT are reasonably efficient in real-world applications.
Our algorithm achieves a complexity improvement over SymKAT~\cite{pous_SymbolicAlgorithmsLanguage_2015}, which despite its efficiency is technically in EXPSPACE.
This improvement is due to us opt to not completely unfold Boolean expressions into BDDs, which can be exponential in the number of primitive tests in the worst case.
This theoretical improvement is also reflected in real-world performance in our evaluations, where we observed orders-of-magnitude speed up compare to SymKAT, as detailed in~\Cref{sec:imple-evaluation}.

\section{Application: Decision Procedure for \cfgkat}\label{sec:aut-construction-cfgkat}

In~\Cref{sec:overview,sec:correctness-and-complexity}, we presented decision procedures for \gkat automata and symbolic \gkat automata. We now want to apply these algorithms to decide equivalence between programs directly and therefore need algorithms to compile expressions to automata. 
Although the procedure to generate \gkat automata from \gkat and \cfgkat expressions has been presented in previous works~\cite{smolka_GuardedKleeneAlgebra_2020,schmid_GuardedKleeneAlgebra_2021,zhang_CFGKATEfficientValidation_2025}, we need an algorithm to generate \emph{symbolic \gkat automata} on-the-fly in order to take advantage of our efficient algorithm in~\Cref{alg:symb-bisim}.

In this section, we introduce an algorithm to generate symbolic \gkat automata from \cfgkat programs, in the following steps:
\begin{center}
  \begin{tikzpicture}
    \node[flowchartnode, inner sep=5pt] (prog) [align=center] {\cfgkat\\Programs};
    \node[flowchartnode, inner sep=6pt] (cf-aut) [right=1.5cm of prog, align=center] {Symbolic\\\cfgkat\\Automata};
    \node[flowchartnode, inner sep=5pt] (gkat-aut) [right=2cm of cf-aut, align=center] {Symbolic\\\gkat\\Automata};
    \node[flowchartnode, inner sep=3pt] (equiv) [right=1.25cm of gkat-aut] {Equivalence};

    \draw[->, thick] (prog) edge 
      node[above=3pt] {Derivative} 
      node[below, align=center] {\Cref{sec:cf-gkat-symb-deriv}\\\Cref{sec:cfgkat-computing-loops}} 
      (cf-aut);
    \draw[->, thick] (cf-aut) edge 
      node[above=3pt] {Jump Resolution} 
      node[below] {\Cref{sec:cfgkat-label-extraction}} 
      (gkat-aut);
    \draw[->, thick] (gkat-aut) edge 
      node[below] {\Cref{alg:symb-bisim}} 
      (equiv);
  \end{tikzpicture}
\end{center}
Similar to the original algorithm~\cite{zhang_CFGKATEfficientValidation_2025}, we rely on an intermediate representation called symbolic \cfgkat automata, which holds additional information to indicate how the computation can be resumed.
Our algorithm first computes the symbolic \cfgkat automata on-the-fly using derivatives, as detailed in~\Cref{sec:cf-gkat-symb-deriv,sec:cfgkat-computing-loops}; then for each transition, we connect all the jump continuations, generated by \command{goto}s, to the appropriate program location, through \emph{label extraction} and \emph{jump resolution}, which we describe in~\Cref{sec:cfgkat-label-extraction}.

\begin{remark}[Thompson's Construction v.s.\ Derivatives]
  \citet{zhang_CFGKATEfficientValidation_2025} defines the operational semantics of \cfgkat through Thompson's construction, which inductively computes the \emph{entire automaton} from a program.
  While our algorithms can decide the equivalences between fully computed automata, the derivative leverages the on-the-fly nature of our algorithm: it computes \emph{one-step transitions} from a given program, allowing immediate termination of automata generation upon failures. 
  However, transitioning from Thompson's construction to the derivative is not straightforward, as the derivative lacks access to the automata of sub-expressions.
  One challenge is to ensure that \(\comBrk\) and \(\comCont\) resume correctly at their parent loop in the presence of nested loops, which is elegantly resolved using the unfolding operator (see~\Cref{fig:sem-CF-GKAT} and~\Cref{rem:unfolding-operator}).
\end{remark}

% \begin{figure}
%   \center
%   \begin{tikzpicture}
%     \node[flowchartnode, inner sep=5pt] (prog) [align=center] {\cfgkat\\Programs};
%     \node[flowchartnode, inner sep=6pt] (cf-aut) [right=1.5cm of prog, align=center] {Symbolic\\\cfgkat\\Automata};
%     \node[flowchartnode, inner sep=5pt] (gkat-aut) [right=2cm of cf-aut, align=center] {Symbolic\\\gkat\\Automata};
%     \node[flowchartnode, inner sep=3pt] (equiv) [right=1.25cm of gkat-aut] {Equivalence};

%     \draw[->, thick] (prog) edge 
%       node[above=3pt] {Derivative} 
%       node[below, align=center] {\Cref{sec:cf-gkat-symb-deriv}\\\Cref{sec:cfgkat-computing-loops}} 
%       (cf-aut);
%     \draw[->, thick] (cf-aut) edge 
%       node[above=3pt] {Jump Resolution} 
%       node[below] {\Cref{sec:cfgkat-label-extraction}} 
%       (gkat-aut);
%     \draw[->, thick] (gkat-aut) edge 
%       node[below] {\Cref{alg:symb-bisim}} 
%       (equiv);
%   \end{tikzpicture}
%   \caption{Symbolic On-the-fly Decision Procedure For \cfgkat.}\label{fig:symb-equiv-cf-gkat}
%   \vspace{-10pt}
% \end{figure}

One significant application of \cfgkat is validating the control-flow structuring phase of decompilers~\cite{yakdan_NoMoreGotos_2015,erosa_TamingControlFlow_1994}, which we used in some of our experiments. 
Interestingly, in the process of benchmarking our implementations, we identified a bug~\cite{nationalsecurityagency_possbile_2025} in the industry standard decompiler: Ghidra.

\begin{remark}[Equivalence for \gkat Expressions]
  Because \cfgkat is a super system of \gkat, the compilation algorithm for \gkat (
  \ifFull  
  in~\Cref{ap:aut-construction-gkat}
  \else
  in full version~\cite{zhangOutrunningBigKATs2026}
  \fi) can be derived from that of \cfgkat, by ignoring indicator variables and continuations.
\end{remark}

\subsection{Syntax and Symbolic \cfgkat Automata}

In addition to the common control-flow structues like if-states and while-loops, \cfgkat also includes indicator variable and non-local control-flow structures like break, continue, return, and goto's. 
Formally, \cfgkat expressions are defined over a finite set of indicator variables \(x \in X\), a finite set of indicator values \(i \in I\), and a finite set of labels \(l \in L\):
\begin{align*}
  \BExpI_T \ni b, c  \triangleq {}&
  0
  \mid 1
  \mid t \in T
  \mid {x = i}
  \mid b \lor c
  \mid b \land c
  \mid \overline{b} \\
  \cfgkat \ni e, f  \triangleq {}&
    b \in \BExpI
    \mid p \in \Sigma
    \mid {x := i}
    \mid {e \unfold f}
    \mid e \seq f
    \mid e +_b f
    \mid \\
  &
    e^{(b)}
    \mid { \comBrk}
    \mid { \comCont}
    \mid { \comRet}
    \mid { \comGoto{l}}
    \mid { \comLabel{l}}
\end{align*}
\(\BExpI_T\) are Boolean expressions generated by indicator tests $x=i$ and the primitive tests from  \(T\); and we write \(\BExpI\) when \(T\) is irrelevant or clear from context.
For a \cfgkat expression \(e\) to be a \cfgkat \emph{program} (or \emph{well-formed expressions}), we enforce the following constraint: (1) Each \(\comLabel{l}\) appears at most once in the program \(e\). (2) If \(\comGoto{l}\) appears in \(e\), then so does \(\comLabel{l}\). (3) \(\comBrk\) or \(\comCont\) either appears in loop, or the left side of an unfolding \(\unfold\). (4) All programs must end in return, i.e. of the form \(f \seq \comRet\) or \(\comRet\).

In addition to the syntax introduced in the original paper~\cite{zhang_CFGKATEfficientValidation_2025}, we added two additional constructions: \(\comCont\), a common control-flow manipulating structure used in while-loops; and the \emph{unfolding operator} \(e \unfold f\) is an expression replaces sequencing in the loop unfolding.

\begin{remark}[The Unfolding Operator]\label{rem:unfolding-operator}
  It is common to assume that the loop \(e^{(b)}\) can be unfolded into an if-statement: \(e^{(b)} \equiv e \seq e^{(b)} +_b 1\)~\cite{smolka_GuardedKleeneAlgebra_2020}.
  However, this identity no longer holds \emph{in \cfgkat} with the addition of break and continue. 
  Consider \(\comBrk^{(1)}\) (i.e. \(\comWhile{\true}{\comBrk}\)): the behavior of this program is equivalent to \(1\) (i.e. \(\comSkip\)), as the execution always enters the loop, then immediately breaks and halts, without executing any action.
  However, \(\comBrk \seq \comBrk^{(1)} +_b 1\) is equivalent to \(\comBrk +_b 1\), because \(\comBrk\) annihilates the right side of sequencing~\cite{zhang_CFGKATEfficientValidation_2025}.

  Therefore, we introduce the unfolding operator \(\unfold\), which satisfies the equality \(e^{(b)} \equiv e \unfold e^{(b)} +_b 1\), by properly handling the \(\comBrk\) and \(\comCont\) within \(e\).
  Additionally, we can soundly encode do-while loop: \(\comDoWhil{e}{b}\), as \(e \unfold e^{(b)}\), even when \(e\) contains \(\comCont\) or \(\comBrk\).
\end{remark}

As we have noted before, our goal for this section is to construct a symbolic \gkat automaton for each \cfgkat program, enabling the use of symbolic equivalence algorithm to efficiently decide equivalences of \cfgkat terms.
However, to get there, we will first construct symbolic \cfgkat automata, which is the symbolic version of the \cfgkat automata~\cite{zhang_CFGKATEfficientValidation_2025}.
The difference between \gkat and \cfgkat automata is that \cfgkat automata also keep track of \emph{ending state}, by outputting \emph{continuations}, allowing the computation to \emph{resume} under appropriate conditions.
Formally, we define the set of continuations \(C\) with the following elements: let \(\pi: X \to I\) ranges over all the indicator states,
\begin{itemize}[leftmargin=*, nosep]
  \item \(\acc{\pi}\): the program terminates normally with the ending assignment \(\pi\);  
  \item \(\retc\): the program halts by a \(\comRet\) statement, annihilating further actions;  
  \item \(\brkc{\pi}\): the program will break out of the current loop with assignment \(\pi\);  
  \item \(\contc{\pi}\): the program will restart the current loop with assignment \(\pi\);  
  \item \(\jmpc{l, \pi}\): the program jumps to label \(l\) with assignment \(\pi\).
\end{itemize}

\begin{definition}
  A \emph{symbolic \cfgkat automaton} is a tuple $(S, s_0, \varepsilon, \delta)$, with a set of states \(S\), a start state \(s_0\), an accepting function \( \varepsilon \), and a transition function \( \delta \) with the following signature:
  \begin{align*}
   \varepsilon & \colon S \to \mathcal{P}(\BExp \times C), &
   \delta & \colon S \to \mathcal{P}(\BExp \times S \times \Sigma).
  \end{align*}
\end{definition}
We use \(s \transOut{b}{c}\) to denote \((b, c) \in \varepsilon(s)\); \(s \transvia{b \mid p} s'\) to denote \((b, s', p) \in \delta(s)\); and \(s \transRes{b} r\) to denote \((b, c) \in \varepsilon(s)\) if \(r = c \in C\) or \((b, s', p) \in \delta(s)\) if \(r = (s', p) \in S \times \Sigma\).

\subsection{Generating Symbolic \cfgkat Automata With Derivatives}\label{sec:cf-gkat-symb-deriv}

We use derivatives~\cite{schmid_GuardedKleeneAlgebra_2021,kozen_CoalgebraicTheoryKleene_2017,antimirov_PartialDerivativesRegular_1996,brzozowski_DerivativesRegularExpressions_1964} to generate \cfgkat automata, which will enable early termination upon counter-examples in the symbolic algorithm.
We first introduce several operations on indicator states \(\pi: X \to I\).
The first operation reassigns a variable \(x \in X\) of a indicator state \(\pi\) to \(i \in I\):
\[\pi[x \mapsto i](y) \triangleq \begin{cases}
  i & \text{if } x = y \\  
  \pi(y) & \text{otherwise}
\end{cases}\]
The second operation will resolve all the indicator tests in \(b \in \BExpI\) when given an indicator state \(\pi\), giving us a condition \(b[\pi]: \BExp\) to be used for transitions in a \cfgkat automaton:
\begin{align*}
  \false[\pi] & \triangleq \false & 
  \true[\pi] & \triangleq \true \\
  (x = i)[\pi] & \triangleq \begin{cases}
    \true & \pi(x) = i \\  
    \false & \pi(x) \neq i 
  \end{cases} & 
  (b \land c)[\pi] & \triangleq b[\pi] \land c[\pi] \\
  (b \lor c)[\pi] & \triangleq b[\pi] \lor c[\pi] &
  ( \neg  b)[\pi] & \triangleq \neg (b[\pi])
\end{align*}

\begin{figure}[t]
  \centering
  \begin{tabular}{c}
    \toprule
    \multicolumn{1}{l}{\scriptsize\textbf{Primitives Rules}}\\
    \(
    \begin{aligned}
      & (\pi, \comAssert{b}) \transOut{b[\pi]}{\acc{\pi}} &
      & (\pi, \comBrk)  \transOut{\true}{\brkc{\pi}} &  
      & (\pi, \comCont)  \transOut{\true}{\contc{\pi}} \\
      & (\pi, \comGoto{l})  \transOut{\true}{\jmpc{l, \pi}} &
      & (\pi, \comRet)  \transOut{\true}{\retc} &
      & (\pi, x := i)  \transOut{\true}{\acc{\pi[x \mapsto i]}} \\
      && &(\pi, p)  \transvia{\true \mid p} (\pi,\comSkip) 
    \end{aligned} 
    \) 
    \\ \midrule
    \begin{tabular}{c | c}
      \multicolumn{1}{l}{\scriptsize\textbf{Sequencing Rules}} &
      \multicolumn{1}{|l}{\scriptsize\textbf{Unfolding Rules}} \\
      \begin{minipage}{0.4\textwidth}
        \vspace{-15pt}
        \begin{mathpar}\mprset{sep=3mm}
          \inferrule
          {(\pi, e) \transOut{b}{\acc{\pi'}} \\ (\pi', f) \transRes{a} r}
          {(\pi, e \seq f) \transRes{b \land a} r}
          \and 
          \inferrule
          {(\pi, e) \transvia{b \mid p} ( \pi', e')}
          {(\pi, e \seq f) \transvia{b \mid p} ( \pi', e' \seq f)}
          \and 
          {\inferrule
          {(\pi, e) \transOut{b}{c} \\ c \neq \acc{\pi'}}
          {(\pi, e \seq f) \transOut{b}{c}}}
        \end{mathpar}
      \end{minipage}
      &
      \begin{minipage}{0.6\textwidth}
        \vspace{-10pt}
        \begin{gather*}
          {\inferrule
          {(\pi, e) \transOut{b}{\brkc{\pi'}}}
          {(\pi, e \unfold f) \transOut{b}{\acc{\pi'}}}}
          \\
          {\mprset{sep=3mm}
          \inferrule
          {(\pi, e) \transOut{b}{c} \\ c \in \{\contc{\pi'}, \acc{\pi'}\} \\ (\pi', f) \transRes{a} r}
          {(\pi, e \unfold f) \transRes{b \land a}{r}}}
          \\ 
          {\mprset{sep=3mm}
          \inferrule
          {(\pi, e) \transOut{b}{c} \\ c \in \{\retc, \jmpc{(l, \pi')}\}}
          {(\pi, e \unfold f) \transOut{b}{c}}}
          \\
          \inferrule
          {(\pi, e) \transvia{b \mid p} ( \pi', e')}
          {(\pi, e \unfold f) \transvia{b \mid p} ( \pi', e' \unfold f)}
        \end{gather*}
      \end{minipage}
    \end{tabular} 
    \\ \midrule
    \multicolumn{1}{l}{\scriptsize\textbf{Conditional Rules and Loop Rules}} \\
    \begin{minipage}{\textwidth}
    \vspace{-10pt}
    \[
        \inferrule
        {(\pi, e) \transRes{a} r}
        {(\pi, e +_b f) \transRes{b[\pi] \land a} r}
        \quad 
        \inferrule
        {(\pi, f) \transRes{a} r}
        {(\pi, e +_b f) \transRes{ \neg b[\pi] \land a} r}
        \quad 
        \inferrule
        {\\}
        {(\pi, e^{(b)}) \transOut{ \neg b[\pi]}{\acc{\pi}}}
        \quad 
        {\inferrule
        {(\pi, e \unfold e^{(b)}) \transRes{a} r}
        {(\pi, e^{(b)}) \transRes{b[\pi] \land a} r}}
    \]
    \end{minipage}
    \\ \bottomrule
  \end{tabular}
  \caption{Operational rules to derive the continuation semantics of \cfgkat.}\label{fig:sem-CF-GKAT}\vspace{-.4cm}
\end{figure}

In~\Cref{fig:sem-CF-GKAT}, we present the operational semantics of a \cfgkat expression as a symbolic \cfgkat automaton.
The states in the automaton consist of two components: \(\pi: X \to I\) is the current indicator state, and \(e: \cfgkat\) is the remaining expression to execute.
%%% NOTE: already introduced in the defintion of symbolic CF-GKAT automata
% The notation used in~\Cref{fig:sem-CF-GKAT} is standard notation for \cfgkat automata, namely \((\pi, e) \transOut{b} c\) identifies the element \((b, c) \in \epsilon(\pi, e)\); and \((\pi, e) \transvia{b \mid p} (\pi', e')\) identifies the element \((b, (\pi', e), p) \in \delta(\pi, e)\).
% Similar to the operational semantics of \gkat,  we also use the shorthand \((\pi, e) \transRes{b} r\) to denote that the state \((\pi, e)\) will reach some result \(r\) given condition \(b\), where \(r\) can be either a continuation or transition.
%%% NOTE: as the GKAT automata is no longer in the main text, there is no reason to compare with it.
% Besides the rule generated by the new primitives and connectives, we also have a new rule for the while loops, which we explain more in-depth in~\Cref{sec:cfgkat-computing-loops}.
With the rules in~\Cref{fig:sem-CF-GKAT}, we can now construct an automaton for every expression \(e\) with a starting indicator state \(\pi\), which we denote by \(\langle \pi, e \rangle\). The states of \(\langle \pi, e \rangle\) are all the pairs of expressions and indicator-states that are reachable from \((\pi, e)\) by the transition function \(\delta\) defined in~\Cref{fig:sem-CF-GKAT}, the start state is \((\pi, e)\), and the accepting and transition functions are \(\varepsilon\) and \(\delta\) defined in~\Cref{fig:sem-CF-GKAT} restricted to the reachable states.

\begin{example}\label{exp:cf-gkat-loop-body-example}
  We compute automata \(\langle x \mapsto 3, e \rangle\) and \(\langle x \mapsto 4, e \rangle\) for the loop body \(e\) in~\Cref{fig:loop-cfgkat-program-cumu}.
  For the starting indicator state \(x \mapsto 3\), the branches \(x = 1\) and \(x = 0\) are not executed, because \((x=1)[x \mapsto 3] \triangleq \false\) and \((x=0)[x \mapsto 3] = \false\).
  Thus, the automaton can only go into third branch when \(x = 3\).
  On the other hand, when \(x\) is set to \(4\), \((x = 3 \land b)[x \mapsto 4]\) also evaluates to false; therefore, only the last branch will execute. We draw both automata as follows (with blocked transitions removed):
  \begin{center}
    \vspace{-0.2cm}
    \begin{tikzpicture}
      \node (init) {};
      \node[textstate] (e) [right=5mm of init] {\((x \mapsto 3, e)\)};
      \node[textstate] (skip) [right=of e] {\((x \mapsto 3, \comSkip)\)};
      \node (x3Acc) [above=3mm of skip] {\(\acc{(x \mapsto 3)}\)};
      \node (x4Acc) [above=3mm of e] {\(\acc{(x \mapsto 4)}\)};
      \draw[->] (init) edge (e);
      \draw[->] (e) edge[output-edge] node[auto] {\(b\)} (x4Acc);
      \draw[->] (e) edge node[auto] {\(\overline{b} \land a \mid p\)} (skip);
      \draw[->] (skip) edge[output-edge] node[auto,swap] {\(1\)} (x3Acc);

      % Boxes
      \node (x3aut) [
        fit=(init) (e) (skip) (x3Acc) (x4Acc), 
        draw=gray, dashed, rounded corners=5pt, thick, draw opacity=0.5,
        label={[anchor=south west]north west:\(\langle x \mapsto 3, e \rangle:\)}] {};
    \end{tikzpicture}
    \quad
    \begin{tikzpicture}
      \node (init) {};
      \node[textstate] (e) [right=5mm of init] {\((x \mapsto 4, e)\)};
      \node[textstate] (skip) [right=of e] {\((x \mapsto 4, \comSkip)\)};    
      \node (x4Acc) [above=3mm of skip] {\(\acc{(x \mapsto 4)}\)};
      \draw[->] (init) edge (e);
      \draw[->] (e) edge node[auto] {\(a \mid p\)} (skip);
      \draw[->] (skip) edge[output-edge] node[auto,swap] {\(1\)} (x4Acc);

      \node (x4aut) [
        fit=(init) (e) (skip) (x4Acc), 
        draw=gray, dashed, rounded corners=5pt, thick, draw opacity=0.5,
        label={[anchor=south west]north west:\(\langle x \mapsto 4, e \rangle:\)}] {};
    \end{tikzpicture}
  \end{center}
\end{example}

% \subsection{Challenges in Computing the Derivatives for Loops}

% % Most of the premise of the derivation of derivatives is a sub-expression of the conclusion. 
% % For example, because of the following two rules:
% % \begin{mathpar}
% %   \inferrule
% %   {(\pi , e) \transOut{a}{c}}
% %   {(\pi , e +_b f) \transOut{b[\pi] \land a}{c}}
% %  \and 
% %   \inferrule
% %   {(\pi , f) \transOut{a}{c}}
% %   {(\pi , e +_b f) \transOut{\neg b[\pi] \land a}{c}}
% % \end{mathpar}
% % we can recursively compute \( \varepsilon (\pi, e +_b f)\) from \( \varepsilon (\pi, e)\) and \( \varepsilon (\pi, f)\):
% % \[ \varepsilon (\pi, e +_b f) \triangleq  \langle b[\pi]|~ \varepsilon (\pi, e) \cup \langle \neg b[\pi]|~ \varepsilon (\pi, f).\]
% % Notice that this recursion is well-formed, because the input to the recursive call of \( \varepsilon \), i.e. \(e\) and \(f\), are sub-expressions of \(e +_b f\).

% The while loop rule is a bit more intricate to compute, because of the addition of indicator variable.
% Consider a \gkat expression \(e^{(b)}\), if starting condition \(a \in \BExp\) that implies \(b\), and the loop halts immediately upon \(a\) without executing any action, then \(e^{(b)}\) will necessarily reject under \(a\), as the loop will keep restarting with the same condition \(a\) and never terminate or execute a primitive action.
% However, in \cfgkat,   because the loop bodies can modify the indicator variables without executing any action, an immediately accepting loop bodies in \cfgkat no longer implies rejection.

\subsection{Computing Loops}\label{sec:cfgkat-computing-loops}

As shown in~\Cref{fig:sem-CF-GKAT}, the transitions of most expressions can be computed from transitions of its sub-expressions.
While loops \(e^{(b)}\) are the only exception, as the premise for the second loop rule relies on the transition of its unfolding \(e \unfold e^{(b)}\), which is a larger expression than \(e^{(b)}\).

\begin{figure}[t]
  \begin{subfigure}[b]{0.3\textwidth}
    \centering
    \begin{align*}
      & \colorWhile{x \neq 2}{\\[-.5ex]  
      & \quad\colorIfThen{x = 1}{\\[-.5ex]  
      & \qquad x := 0;} \\[-.5ex]  
      & \quad\colorElseIf{x = 0}{\\[-.5ex]    
      & \qquad x := 1;} \\[-.5ex]  
      & \quad\colorElseIf{x = 3 \land b}{ \\[-.5ex]  
      & \qquad x := 4;} \\[-.5ex]  
      & \quad \colorElse{ \\[-.5ex]  
      & \qquad (\colorAssert{a}); p;}
      }
    \end{align*}
    \caption{}\label{fig:loop-cfgkat-program-cumu}
  \end{subfigure}
  \hfil
  \begin{subfigure}[b]{0.7\textwidth}
    \centering
    \begin{tikzpicture}
      \node (x1Init) {};
      \node (x2Init) [below=13mm of x1Init] {};
      \node (x3Init) [below=13mm of x2Init] {};
      \node[textstate] (x1) [right=5mm of x1Init] {\((x \mapsto 1, e^{(x \neq 2)})\)};
      \node[textstate] (x2) [right=5mm of x2Init] {\((x \mapsto 2, e^{(x \neq 2)})\)};
      \node[textstate] (x3) [right=5mm of x3Init] {\((x \mapsto 3, e^{(x \neq 2)})\)};
      \node[textstate] (x4) [right=of x3] {\((x \mapsto 4, e^{(x \neq 2)})\)};
      \node (x2acc) [right=5mm of x2] {\(\acc{(x \mapsto 2)}\)};

      \draw[->] (x1Init) edge (x1);
      \draw[->] (x2Init) edge (x2);
      \draw[->] (x3Init) edge (x3);
      \draw[->] (x2) edge[output-edge] node[above] {\(\true\)} (x2acc);
      \draw[->] (x3) edge[out=-125, in=-55, looseness=2] node[below=-1mm] (x3loop) {\(\overline{b} \land a \mid p\)} (x3);
      \draw[->] (x3) edge node[above] {\(b \land a \mid p\)} (x4);
      \draw[->] (x4) edge[out=-125, in=-55, looseness=2] node[below] (x4loop) {\(a \mid p\)} (x4);

      % Boxes
      \node (x3aut) [
        fit=(x3Init) (x3) (x4) (x3loop) (x4loop), 
        draw=gray, dashed, rounded corners=5pt, thick, draw opacity=0.5,
        label={[anchor=south west]north west:\(\langle x \mapsto 3, e^{(x\neq2)} \rangle:\)}] {};

      \coordinate (x2boxEnd) at ([xshift=0cm]x4.east |- x2Init);
      \node (x2aut) [
        fit=(x2Init) (x2) (x2acc) (x2boxEnd), 
        draw=gray, dashed, rounded corners=5pt, thick, draw opacity=0.5,
        label={[anchor=south west]north west:\(\langle x \mapsto 2, e^{(x\neq2)} \rangle:\)}] {};

      \coordinate (x1boxEnd) at ([xshift=0cm]x4.east |- x1Init);
      \node (x1aut) [
        fit=(x1Init) (x1) (x1boxEnd), 
        draw=gray, dashed, rounded corners=5pt, thick, draw opacity=0.5,
        label={[anchor=south west]north west:\(\langle x \mapsto 1, e^{(x\neq2)} \rangle:\)}] {};
    \end{tikzpicture}
    \caption{}\label{fig:loop-aut-cfgkat-cumu}%The automata \(\langle x \mapsto 1, e^{(x \neq 2)} \rangle\), \(\langle x \mapsto 2, e^{(x \neq 2)} \rangle\), and \(\langle x \mapsto 3, e^{(x \neq 2)} \rangle\), with blocked transitions removed.}
  \end{subfigure}
  \caption{Program \(e^{(x \neq 2)}\), on the left, where \(e\) denotes the loop body. On the right, the automata for different starting indicator assignments,  with blocked transitions removed.}\label{fig:loop-cfgkat-example-cumu}\vspace{-.4cm}
\end{figure}

In this section, we will show that the transitions of \(e^{(b)}\) is computable from the transition of \(e\).
for \cfgkat, computing the transition of \(e^{(b)}\) is particularly difficult, because the first action might be encountered after arbitrarily long iteration of the loop body.
For example, consider the transition of state \((x \mapsto 3, e^{(x \neq 2)})\) in~\Cref{fig:loop-aut-cfgkat-cumu}, its self-loop under condition \(\overline{b} \land a\) is derivable in one iteration of the body; yet the transition to \((x \mapsto 4, e^{(x \neq 2)})\) under condition \(b \land a\) requires two iterations (transitions of the loop body \(e\) is in~\Cref{exp:cf-gkat-loop-body-example}):
\begin{center}

  \scalebox{0.75}{
    \(\inferrule*
    {
      \inferrule*
      {(x \mapsto 3, e) \transvia{\overline{b} \land a \mid p} (x \mapsto 3, e)}
      {\vdots}
    }
    {(x \mapsto 3, e^{(x \neq 2)}) \transvia{\overline{b} \land a \mid p} (x \mapsto 3, e^{(x \neq 2)})}
    \quad
    \inferrule*
    {
      \inferrule*
      {(x \mapsto 3, e) \transOut{b} \acc{(x \mapsto 4)} \\
      (x \mapsto 4, e) \transvia{a \mid p} (x \mapsto 4, \comSkip)}
      {\vdots}
    }
    {(x \mapsto 3, e^{(x \neq 2)}) \transvia{b \land a \mid p} (x \mapsto 4, e^{(x \neq 2)})}
\)
}
\end{center}
\ifFull
Formal derivation of these two results can be found in~\Cref{ap:omitted} (\Cref{fig:deriv-out-going-transition-cumu-loop-example}). 
\fi

To compute the transition of a loop, \textbf{we will need to iterate the transition of the loop body}.
\citet{zhang_CFGKATEfficientValidation_2025} named this process ``iteration lifting'', which iterates the result for a fixed atom, until the program outputs a continuation, executes an action, or goes into an infinite loop without executing any
actions. However, the situation in the symbolic case is more complicated, as we will need to combine the result of the first iteration with the rest. 
\begin{wrapfigure}{l}{5.5cm}
\vspace{-1.2cm}
  \begin{minipage}{6cm}
    \[
      \inferrule*
      {
        \inferrule*{
          (\pi, e) \transOut{a} \acc{\pi'} \\  
          (\pi', e^{(b)}) \transOut{d} c
        }
        {\vdots}
      }
      {(\pi, e^{(b)}) \transOut{b[\pi] \land a \land d} c}
    \]
  \end{minipage}
  \vspace{-.9cm}
\end{wrapfigure}
For example, take the derivation on the left, we need to combine the condition to enter the loop \(b[\pi]\), the condition for the first iteration \(a\), with the condition of the rest \(d\), to get the final condition \(b[\pi] \land a \land d\).
To account for this difference, our accumulation function, defined below, outputs a function \(h\) in the recursive case to transform the result obtained from recursion.

\begin{definition}[Accumulation]\label{def:accumulation}
  We use \(N\) to denote the type of input, and \(R\) the type of output.
  \(\accu\) takes two functions \(\res: N \to \mathcal{P}(R)\), \(\con: N \to \mathcal{P}(N \times (R \to R))\), and an initial input \(n \in N\):
  \begin{itemize}[nosep,leftmargin=*]
    \item \(\res\) takes an input in \(N\) and return a set of results computed in one step.
    \item \(\con\) takes an input \(n  \in N\) and return a set of next iterations \(n'  \in N\) paired with the function \(h: R \to R\) to be applied onto the results generated by \(n'\).
  \end{itemize}
  Formally, \(\accu_{\res, \con}(n)\) is the set of elements derivable by the inference rules:
  \vspace{-.8em}
  \begin{mathpar}
    \inferrule
    {r \in \res(n)}
    {r \in \accu_{\res, \con}(n)}
    \and 
    \inferrule
    {(n', h) \in \con(n) \\ r  \in \accu_{\res, \con}(n')}
    {h(r) \in \accu_{\res, \con}(n)}
    \vspace{-.8em}
  \end{mathpar}
\end{definition}

% We then present an algorithm to compute \(\accu\) for a finite set of possible input \(N\); this result will later give us a way to compute the transitions for while-loops and jump continuations.
% This algorithm maintains an auxiliary variable \(M \subseteq N\) to track all explored inputs. 
% When an input has already been explored, i.e. \(n \in M\), such branch will no longer produce any distinct results, hence we terminate such branch by returning \(\emptyset\).
% Since the set of possible inputs \(N\) is finite and the size of \(M\) monotonically increases with each recursive call, the algorithm is guaranteed to terminate.

We then show that \(\accu\) is computable under some side conditions; this result will later give as a way to compute the transitions for while-loops and jump continuations expressing their transitions using \(\accu\).

\begin{theorem}[Computablity]\label{thm:cumulation-computes}
  When the set of input \(N\) is finite and all the \(h\) in \((n', h) \in \con(n)\) is computable, then \(\accu_{\res,\con}(n)  \in \mathcal{P}(R)\) is computable. 
\end{theorem}

In the rest of the section, we use the \(\accu\) function to express the \(\delta\) and \(\varepsilon\) for a loop \(e^{(b)}\), then we prove the correctness of this encoding by \(\accu\). 
Combining the computability of \(\accu\) with the following two results gives a correct and computable procedure to compute \(\delta(\pi, e^{b})\) and \(\varepsilon(\pi, e^{(b)})\).

\myparagraph{Expressing \(\delta\) of while-loops}
We first express \(\delta(\pi, e^{(b)})\) as \(\accu(\res_{\delta}, \con_{\delta}, \pi)\), where the elements of \(\res_{\delta}\) and \(\con_{\delta}\) are defined by the following inference rules:
\vspace{-.8em}
\begin{mathpar}
  \inferrule
  {(\pi, e) \transvia{a \mid p} ( \pi', e')}
  {(b[\pi]  \land a, (\pi', e' \unfold e^{(b)}), p)  \in \res_{\delta}(\pi)}
  \and
  \inferrule
  {(\pi, e) \transOut{a}{c} \\ c \in \{\acc{\pi'}, \contc{\pi'}\}}
  {(\pi', \langle b[\pi] \land  a|)  \in \con_\delta (\pi)}
  \vspace{-.8em}
\end{mathpar}
where \(\langle b |\) guards the transition condition with \(b \in \BExp\):
\(
  \langle b |(a, s, p) \triangleq (b \land a, s, p).
\)

Intuitively, \(\res_{\delta}\) corresponds to the transitions that can computed by single iteration of the loop body.
Whereas \(\con_\delta\) describes the way to compute all the transitions that can be formed by more-than-one iterations of the loop body \(e\).
Indeed, the \(\res_{\delta}\) and \(\con_{\delta}\) function is obtained by unfolding the \emph{only two} ways to derive the transitions of \(e \unfold e^{(b)}\) and \(e^{(b)}\):
\begin{center}
  \scalebox{0.8}{\(
    \inferrule
    {
      \inferrule
      {(\pi, e) \transvia{a \mid p} (\pi', e')}
      {(\pi, e \unfold e^{(b)}) \transvia{a \mid p} (\pi', e' \unfold e^{(b)})}
    }
    {(\pi, e^{(b)}) \transvia{b[\pi] \land a \mid p} (\pi', e' \unfold e^{(b)})}
  \)}
  \qquad
  \scalebox{0.8}{\(
    \inferrule
    {
      \inferrule
      {(\pi, e) \transOut{a}{c} \\ c \in \{\acc{\pi'}, \contc{\pi'}\} \\ (\pi', e^{(b)}) \transvia{d \mid p} (\sigma, f)}
      {(\pi, e \unfold e^{(b)}) \transvia{a \land d \mid p} (\sigma, f)}
    }
    {(\pi, e^{(b)}) \transvia{b[\pi] \land a \land d \mid p} (\sigma, f)}
  \)}
\end{center}
The first case is represented by \(\res_{\delta}\): the transition of \(e^{(b)}\) is directly computed from transitions of \(e\); the second case is represented by \(\con_{\delta}\): the transition of \((\pi, e^{(b)})\) is inductively constructed by the transition of \((\pi', e^{(b)})\), then the resulting condition \(d\) is guarded by \(b[\pi] \land a\) to become \(b[\pi] \land a \land d\).

% For example, when \((\pi, e) \transOut{a}{\contc{\pi'}}\), a transiton can be derived only when the following 
% \[
%   \inferrule{
%     \inferrule{
%       (\pi, e) \transOut{a} \contc{\pi'} \\  
%       (\pi', e^{(b)}) \transvia{d \mid p} (\sigma, f)
%     }
%     {(\pi, e \unfold e^{(b)}) \transvia{a \land d \mid p} (\sigma, f)}
%   }{
%     (\pi, e^{(b)}) \transvia{b[\pi] \land a \land d \mid p} (\sigma, f)
%   }
% \]
% \[
%   \mprset{fraction={===}}
%   \inferrule*{
%     \mprset{fraction={---}}
%     \inferrule*{
%       \mprset{fraction={===}}
%       \inferrule*
%       {(\pi, e) \transOut{a} \contc{\pi'}}
%       {(\pi', \langle b[\pi] \land  a|)  \in \con(\pi')}
%       \\ 
%       \mprset{fraction={===}}
%       \inferrule*
%       {(\pi', e^{(b)}) \transvia{d \mid p} (\sigma, f)}
%       {(d, (\sigma, f), p) \in  \accu(\res_\delta , \con_\delta , \pi)}
%     }{
%       (b[\pi] \land a \land d, (\sigma, f), p) \in \accu(\res_\delta , \con_\delta , \pi)
%     }
%   }
%   {(\pi, e^{(b)}) \transvia{b[\pi] \land a \land d \mid p} (\sigma, f)}
% \]

The \(\varepsilon\) of while loop, i.e. elements in \(\varepsilon(\pi, e^{(b)})\) can also be expressed using \(\accu\) by defining \(\res_{\varepsilon}\) and \(\con_{\varepsilon}\), similarly to how we defined \(\res_\delta\) and \(\con_\delta\). 
\ifFull
We defer the precise formulation of this result in~\Cref{ap:omitted}.
\fi

\begin{theorem}[Correctness]\label{thm:loop-comp-cfgkat-correctness}
  For any \gkat expression \(e\), a Boolean expression \(b\), and an indicator assignment \(\pi\), the transition function \(\delta\) and continuation \(\varepsilon\) computed by \(\accu\) is the same as derived from the operational rules in~\Cref{fig:sem-CF-GKAT}:
  \begin{mathpar}
    \delta(\pi, e^{(b)}) = \accu(\res_{\delta}, \con_{\delta}, \pi) \and \text{and} \and
    \varepsilon(\pi, e^{(b)}) = \accu(\res_{\varepsilon}, \con_{\varepsilon}, \pi)
  \end{mathpar}
\end{theorem}

\subsection{Label Extraction And Jump Resolution}
\label{sec:cfgkat-label-extraction}

By an induction and closure argument, we can show that given \(e\), then for all \(\pi: X \to I\), the states in \(\langle \pi, e \rangle\) can only output either \(\jmpc{l, \pi}\) or \(\retc\) as a continuation.
Indeed, because all the \(\comCont\) and \(\comBrk\) can only either appear on the left side of the unfold operation \(\unfold\) or in a loop, then the continuation \(\contc{\pi}\) and \(\brkc{\pi}\) will be resolved when computing the transition for \(\unfold\), as detailed in~\Cref{fig:sem-CF-GKAT}.
Also, since every program ends with \(\comRet\), all the \(\acc{\pi}\) continuation will be replaced by \(\retc\) in the end.
Hence, to obtain a symbolic \gkat automaton from \(\langle \pi, e \rangle\), we will only need to resolve the jump continuation by connecting it to the appropriate location of the program. 

Indeed, the location of each label \(l\) in \(e\) can be extracted syntactically, by recursively diving into the expression until the command \(\comLabel{l}\) is reached.
For readers familiar with the work of~\citet{zhang_CFGKATEfficientValidation_2025}, our definition is a syntactic counterpart to their inductive semantics \(\sem{-}^{l}\).
Formally, given a label \(l \in L\), and an expression \(e \in \cfgkat\) where \(\comLabel{l}\) appears \emph{exactly once}, then the \emph{label extraction} \(e_l\) gives an expression that corresponds to the location of \(l\):
\begin{align*}
  % (-)_{l}&: \cfgkat \to \cfgkat \\  
  (\comLabel{l})_{l} & \triangleq 1 \quad  
  (e^{(b)})_{l}  \triangleq 
    (e)_l \unfold e^{(b)}&
  (e +_b f)_{l} & \triangleq \begin{cases}
    \mathrlap{(e)_{l}} \hphantom{(e)_{l} \unfold f} & \text{if \(\comLabel{l}\) in \(e\)}\\
    (f)_{l} & \text{if \(\comLabel{l}\) in \(f\)}
  \end{cases} \\  
  (e; f)_{l} & \triangleq \begin{cases}
    (e)_{l}; f & \text{if \(\comLabel{l}\) in \(e\)}\\
    (f)_{l} & \text{if \(\comLabel{l}\) in \(f\)}
  \end{cases} & 
  (e \unfold f)_{l} & \triangleq \begin{cases}
    (e)_{l} \unfold f & \text{if \(\comLabel{l}\) in \(e\)}\\
    (f)_{l} & \text{if \(\comLabel{l}\) in \(f\)}
  \end{cases}
\end{align*}
We ignore other expressions, as they do not contain \(\comLabel{l}\).
Unlike the approach taken by~\citet{zhang_CFGKATEfficientValidation_2025}, which resolves the semantics of \(l\) in an expression without \(\comLabel{l}\) to divergence, we only define label extraction on expression \(e\) that contains a unique \(\comLabel{l}\).
This constraint aligns with the requirement of most real-world languages, and is compatible with a \cfgkat program \(e\): all the \(\jmpc{l, \sigma'}\) continuations are produced by \(\comGoto{l}\); and if \(e\) contains \(\comGoto{l}\), then \(e\) contains a unique \(\comLabel{l}\), by definition of a \cfgkat program.

We formally define the jump resolution process: when given a program \(e\), jump resolution will turn a symbolic \cfgkat automaton \(\langle \pi, e \rangle\) into a \emph{symbolic \gkat automaton} \(\langle \pi, e \rangle\JmpRes{}\).
The transition of state \((\sigma, f)\) in \(\langle \pi, e \rangle\JmpRes{}\) can be computed from the transition of \((\sigma, f)\) in \cfgkat automaton \(\langle \pi, e \rangle\), by connection all the jump continuation \(\jmpc{l, \sigma'}\) with the dynamics of the state \((\sigma', e_l)\).
Formally, take any state \((\sigma, f) \in \langle \pi, e \rangle\), we derive its transition in \(\langle \pi, e \rangle\JmpRes{}\) with the following rules, when \((\sigma, f) \transOut{a}{\jmpc{l, \sigma'}}\) in \(\langle \pi, e \rangle\):
\begin{mathpar}
  \inferrule
  {
    (\sigma, f) \transOut{a}{\jmpc{l, \sigma'}} \text{ in } \langle \pi, e \rangle \\ 
    (\sigma', e_l) \transRes{b} r \text{ in } \langle \pi, e \rangle\JmpRes{}
  }
  {(\sigma, f) \transRes{a \land b} r \text{ in } \langle \pi, e \rangle\JmpRes{}}
\end{mathpar}
All transitions of \((\sigma, f)\) that do not result in a jump remain the same as in \(\langle \pi, e \rangle\).

The inference rule for jump resolution is also recursive: the transitions of \(\langle \pi, e \rangle\JmpRes{}\) appears both in the conclusion and the premise, thus we can compute the transitions in \(\langle \pi, e \rangle\JmpRes{}\) with the \(\accu\) function defined in~\Cref{def:accumulation}.

\begin{theorem}[Computability]\label{thm:cfgkat-jmp-compute}
  Fix a program \(e\), the transitions in the symbolic automata \(\langle \pi, e \rangle\) is computable, by expressing it as a \(\accu\).
\end{theorem}

Finally, given any program \(e\) and a starting indicator state \(\pi\), we obtain an on-the-fly algorithm to construct the symbolic GKAT automata \(\langle \pi, e \rangle\JmpRes{}\).
Thus, to decide the equivalence between the equivalence of two \cfgkat programs \(e\) and \(f\) with its respective starting indicator states \(\pi\) and \(\sigma\), we can construct the automata \(\langle \pi, e \rangle \JmpRes{}\) and \(\langle \sigma, f \rangle \JmpRes{}\), and apply the algorithm in~\Cref{alg:symb-bisim}.

Indeed, this development represents a significant increase in the efficiency of equivalence checking for \cfgkat and \gkat.
In the next section, we outline how we implemented the decision procedure for both \gkat and \cfgkat in Rust, and compared them with the current state-of-the-art on both generated examples and real-world control-flow to demonstrate our efficiency.
Indeed, our implementation have not only exhibited orders-of-magnitude performance improvement over the existing tool, but also able to help us locate a bug in Ghidra~\cite{NationalSecurityAgency_Ghidra_2025}, an industry standard decompiler.

\section{Implementation and Evaluation}\label{sec:implementation}

In this section, we detail our Rust implementation~\cite{fu_EfficientDecisionProcedures_2026} and evaluation of the on-the-fly symbolic decision procedures for \CFOrGKAT.
We implemented the symbolic on-the-fly algorithm in \Cref{alg:symb-bisim}, the symbolic derivatives from~\Cref{sec:aut-construction-cfgkat}, and the \gkat variants 
\ifFull 
in~\Cref{ap:aut-construction-gkat}
\fi. 
Then, we compare our implementation against SymKAT~\cite{pous_SymbolicAlgorithmsLanguage_2015} and artifact of \cfgkat~\cite{zhang_2024_13938565}.

\subsection{Implementation Details}\label{sec:imple-detail}

In our implementations, all Boolean expressions are directly converted into the internal format of the Boolean solvers, which currently include miniSAT~\cite{_BooleworksLogicngrs_2025} and BDD~\cite{somenzi_CUDDCUDecision_2018,lewis_PclewisCuddsys_2025}. 
But, we can easily extend our implementation with more solvers as our algorithm supports the use of arbitrary (UN)SAT solvers.
Both \gkat and \cfgkat expressions are hash-consed for faster syntactical comparison and reduced memory footprint.
Our optimization on the expressions are relatively light compared to previous works~\cite{pous_SymbolicAlgorithmsLanguage_2015,katch_artifact_2024}: we only remove left sequencing and unfoldings of skip: \((1 \seq e) \rightsquigarrow e\) and \((1 \unfold e) \rightsquigarrow e\); yet, we can already outperform existing works by a significant margin.
Besides standard optimizations, the following paragraphs present additional details that are specific to our algorithm.

\myparagraph{\textbf{Eager Pruning of Blocked Transitions}}
When generating a transition \(s \transvia{b \mid p} s'\), we will immediately check whether this transition is blocked, i.e. whether \(b \equiv 0\); if so, this transition will be immediately removed (see~\Cref{rem:blocked-transition} for soundness argument).
This optimization significantly reduces the size of the generated automata, and avoids many unnecessary Boolean inequivalence checks.

\myparagraph{\textbf{Encoding of Indicator Variables}}
When adapting our algorithm to \cfgkat, we need to handle indicator variables within \(\BExpI\). 
While it is possible to naively construct their abstract syntax tree, and instantiate them when given an indicator state, it is better to store these expressions using the native representations within the Boolean solver. 
Since off-the-shelf solvers operate on Boolean formulae without indicator variables (\BExp), we encode constraints on indicator variables using fresh (Boolean) variables in the solver,
\noindent
\begin{wrapfigure}{l}{3.2cm}
  \vspace{-0.8cm}
  \begin{minipage}{3.2cm}
    \centering
    \begin{tikzpicture}
      \node[textstate] (formula) {\((x = 1) \land (y = 2) \lor a\)};
      \node[textstate, draw=gray, dashed, rounded corners=5pt, thick, minimum width=3cm,
        label={[anchor=south west]north west:Encoded Expression}]
        (encFormula) [below=1.3cm of formula] {\(t_1 \land t_2 \lor a\)};
      \node[textstate, draw=gray, dashed, rounded corners=5pt, thick, minimum width=3cm,
        label={[anchor=south west]north west:Indicator Constraint Table}] 
        (map) [below=0.5cm of encFormula, align=left] {
      \((x = 1) \mapsto t_1\)\\\((y = 2) \mapsto t_2\)};
      % Add padding only above encFormula
      \coordinate (paddingTop) at ([yshift=1em]encFormula.north);
      % Fit node around encFormula and map
      \node[fit=(paddingTop)(encFormula)(map)] (encAll) {};
      % Draw squiggly arrow from formula to the fit node
      \draw[->, decorate, decoration={snake, amplitude=2px, segment length=3mm}]
      (formula) -- node[right=1mm]{\small encode} (encAll);
    \end{tikzpicture}
  \end{minipage}
  \vspace{-0.9cm}
\end{wrapfigure}
and use a table to track the correspondence of indicator constraints to these fresh variables. For example, consider the expression on the left, where $x$ and $y$ are indicator variables and $a \in \BExp$.
% \begin{align*}
%   (x = 1) \land (y = 2) \lor b
%   \quad\xRightarrow{\text{encode}}\quad
%   \big\{ t_1 \land t_2 \lor b, [(x = 1) \mapsto t_1, (y = 2) \mapsto t_2] \big\}
% \end{align*}
The encoding produces two fresh variables \(t_1\) and \(t_2\) that represent the constraints 
\(x = 1\) and \(y = 2\), respectively. 
The expression is transformed into \(t_1 \land t_2 \lor a\), which is compatible with solver backends, with an indicator constraint table.
When generating these fresh Boolean variables, we ensure that the same indicator constraint is always mapped to the same variable. 
When resolving the encoded expression $b \in \BExpI$ with some indicator state $\pi$, we iterate through all $(x = i) \mapsto t$ in the table for \(b\): if \(\pi(x)\) is indeed \(i\), then we set the variable \(t\) in \(b\) to be true, otherwise false.

\subsection{Empirical Evaluation}\label{sec:imple-evaluation}
Our implementation can utilize arbitrary (UN)SAT solvers; for evaluation we choose two as backend: BDD (CUDD)~\cite{lewis_PclewisCuddsys_2025,somenzi_CUDDCUDecision_2018} and SAT (MiniSAT)~\cite{_BooleworksLogicngrs_2025}. We then compare against existing implementations: SymKAT~\cite{pous_SymbolicAlgorithmsLanguage_2015} and \cfgkat~\cite{zhang_2024_13938565}.
Our experiments are performed on a laptop with an AMD Ryzen 7 8845HS CPU with a time limit of 1000s and memory limit of 8GB. 
We study the performance of these algorithms on both randomly generated benchmarks and real-world control-flows extracted from GNU Coreutils 9.5~\cite{gnuprojectCoreutilsGNUCore2024}.

\begin{figure}[t]
  \begin{tabular}{l l}
  \hspace{-0.7em}\includegraphics[width=0.5\textwidth]{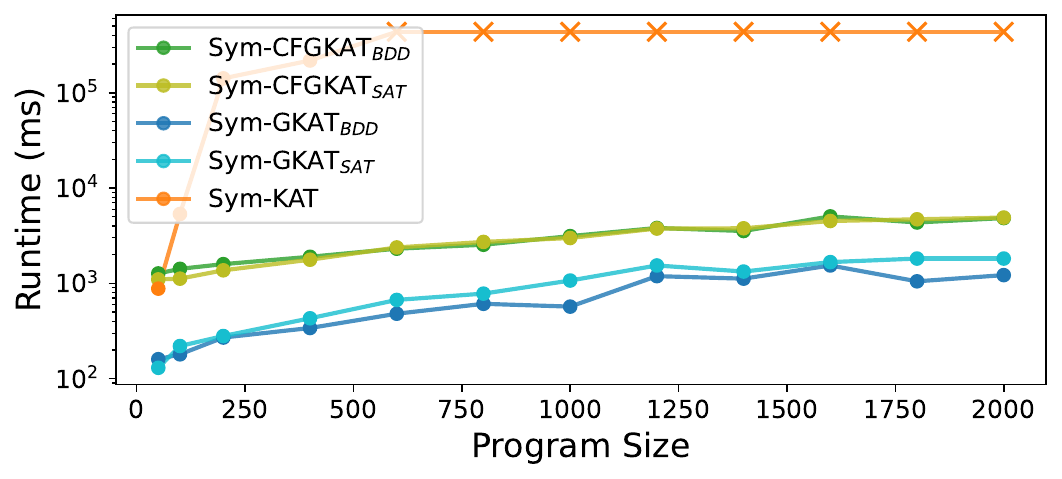} &
  \hspace{-0.7em}\includegraphics[width=0.5\textwidth]{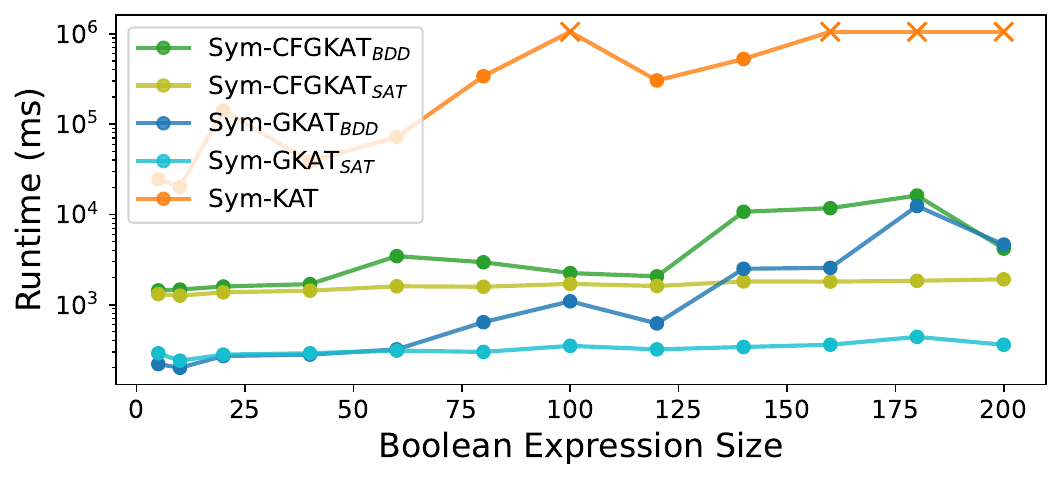} \\
  \hspace{-0.7em}\includegraphics[width=0.5\textwidth]{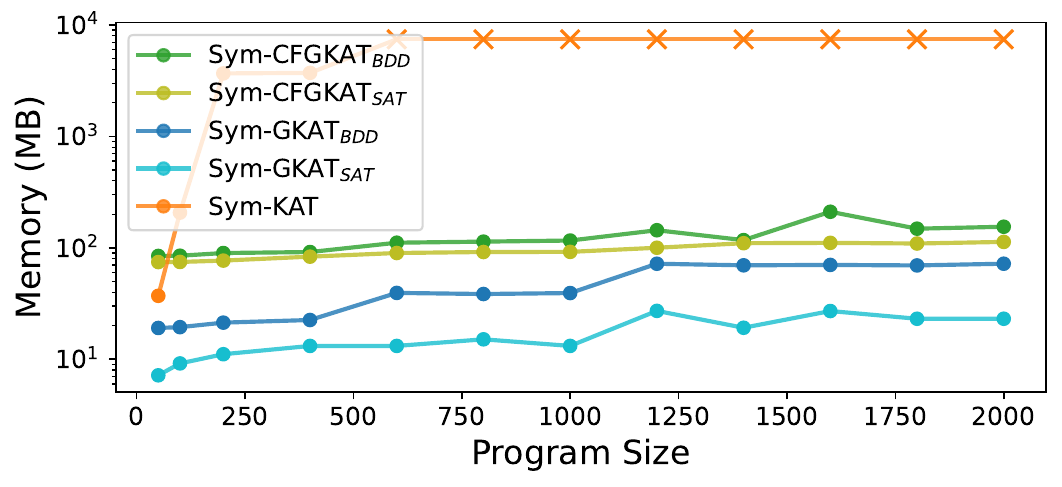} &
  \hspace{-0.7em}\includegraphics[width=0.5\textwidth]{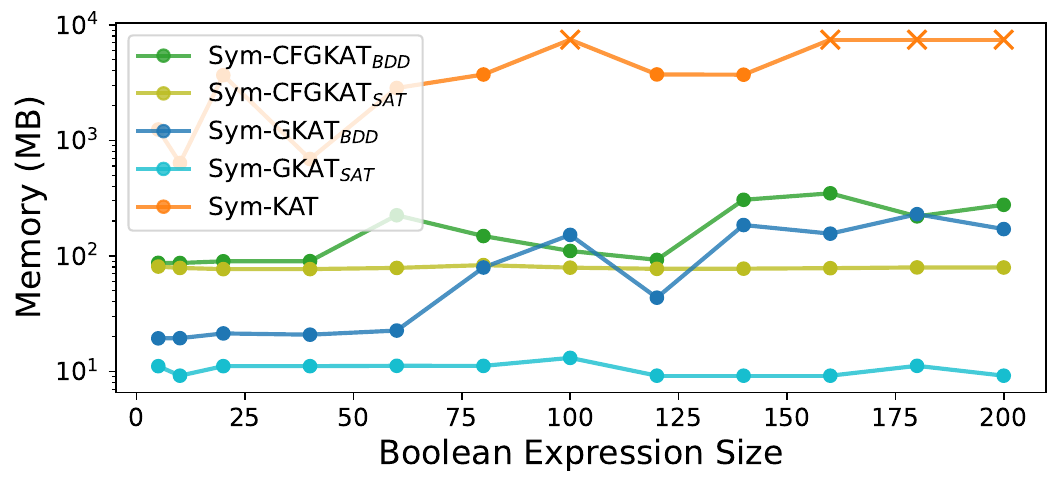}
  \end{tabular}
  \caption{Performance comparison of symbolic \gkat, symbolic \cfgkat, and SymKAT in \underline{lo}g\underline{arithm scale} on the vertical axis.}
  \vspace{-.4cm}
  \label{fig:performance-comparison}
\end{figure}

\myparagraph{\textbf{Synthetic Benchmarks}}
Using the equational rules of GKAT~\cite{smolka_GuardedKleeneAlgebra_2020}, we randomly generate equivalent \gkat-program pairs of varying sizes and Boolean expression sizes. 
We compare our symbolic algorithms against a patched version of SymKAT that handles a larger number of primitive tests on these equivalent pairs, and present the runtime and memory footprint of these experiments in \Cref{fig:performance-comparison}. 
Each point in these graphs represents an average performance on 100 randomly generated program pairs, and the {\color{orange}`X'} markers on the SymKAT line indicate that a dataset failed to be checked within the time limit or memory limit. 
We do not include the implementation of \cfgkat~\cite{zhang_2024_13938565} in these comparisons, as it failed to complete any benchmark within the allotted time and memory.%; indeed, the non-symbolic implementation of \cfgkat prevents it from scaling to larger programs.

Our symbolic algorithms exhibit orders-of-magnitude performance increase over SymKAT in both runtime and memory usage, especially for larger programs. 
Notably, even the symbolic \cfgkat implementations, with the overheads of handling indicator variables and continuations, are able to outperform SymKAT. 
We also note that the performance of BDD-based algorithms, including SymKAT, becomes more erratic as the size of the Boolean expressions grows; whereas MiniSAT-based algorithms exhibit more stable performance behaviors.

\myparagraph{\textbf{GNU Coreutils Decompilation Validation}}
We evaluate our symbolic \cfgkat algorithm by compiling and decompiling real-world control flows extracted from GNU Coreutils 9.5~\cite{gnuprojectCoreutilsGNUCore2024}; then test the trace equivalence of the original and decompiled control flow via \cfgkat. 
Similarly to the original \cfgkat evaluation~\cite{zhang_CFGKATEfficientValidation_2025,zhang_2024_13938565}, we extract the control-flow of C source code through \emph{blinding}, which replaces all the commands and conditions with undefined functions.
This process abstracts away the data flow, and ensures that the compiler and decompiler only transforms control flows.
% Although our focus on control-flow might seem narrow, the fuzzing of control-flow structuring have already shown great promise in bug-discovery for industrial decompilers~\cite{gorzynski_FuzzFleshRandomisedTesting_2025}. 
% Additionally, \textbf{our experiments have identified a potential bug in Ghidra}, thus we expect our efficient, sound, and complete symbolic algorithm for \cfgkat to play a role in the research and practice of decompiler testing.
Compared to the original blinder~\cite{zhang_2024_13938565}, our new blinder supports more control-flow constructs such as $\comCont$, simple switch statements, do-while loops, and for-loops with nested $\comCont$s.
We also make fewer assumptions about the conditions in C in this iteration of the blinder. 
In particular, the semantics of \cfgkat assumes the same primitive test always has the same truth value within a Boolean expression.
However, in the condition \texttt{(x = !x) \&\!\& (x = !x)}, the first test \texttt{(x = !x)} will have different truth value than the second one, even though they are syntactically the same.
Thus, our blinder always generates fresh primitive tests, even for syntactically equal conditions.
As a trade-off, our blinder produces a larger number primitive tests, which increases the complexity of equivalence checking. 
Through blinding, we obtain the abstract control flow of a function, which we compile using Clang 20.1.2~\cite{llvm_LLVMLlvmLlvmproject_2025}, and then decompile using Ghidra 11.4~\cite{NationalSecurityAgency_Ghidra_2025}. 
Additionally, some manual cleanup is performed to remove artifacts generated by Ghidra, adjusting it into a fragment of C that can be read by our implementation of \cfgkat.

\begin{figure}[t]
  \begin{tabular}{l l l}
  \hspace{-0.7em}\includegraphics[width=0.33\textwidth]{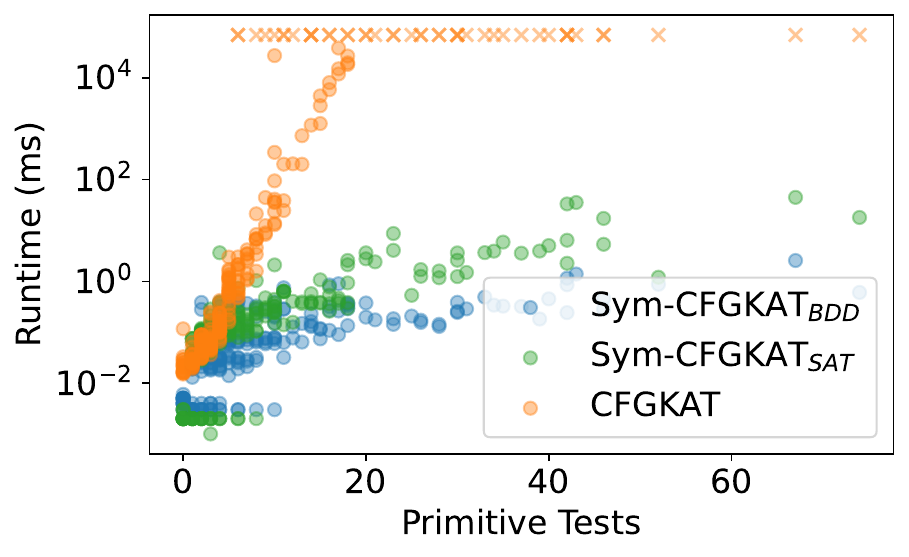} &
  \hspace{-0.7em}\includegraphics[width=0.33\textwidth]{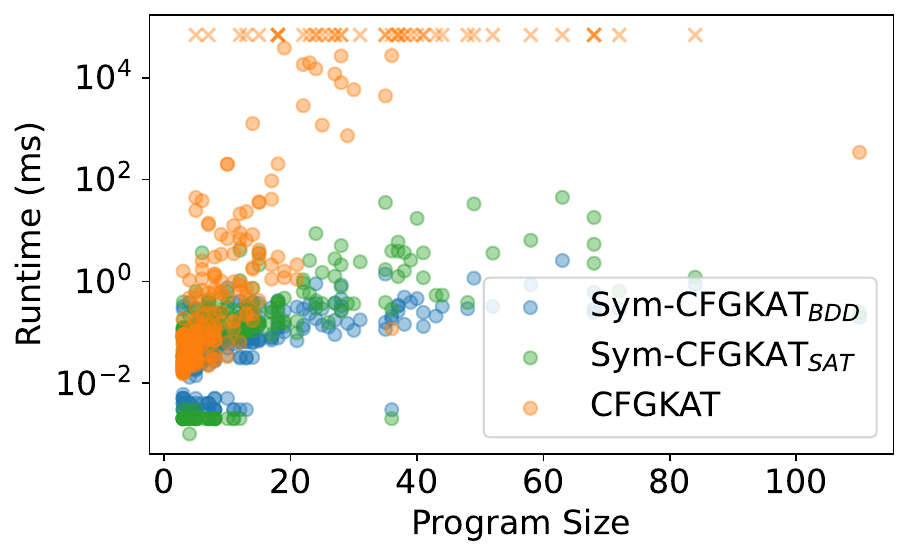} &
  \hspace{-0.7em}\includegraphics[width=0.33\textwidth]{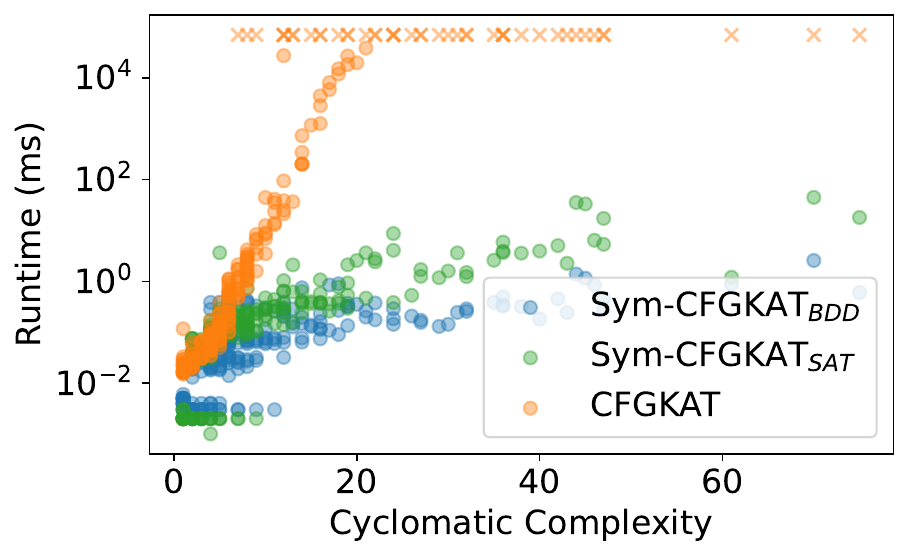}
  \end{tabular}
  \vspace{-1em}
  \begin{tabular}{l l l}
  \hspace{-0.7em}\includegraphics[width=0.33\textwidth]{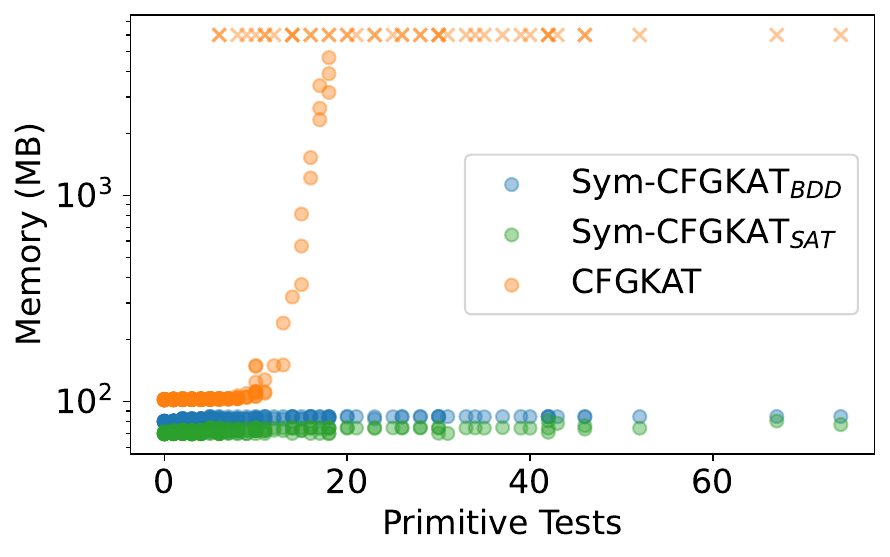} &
  \hspace{-0.7em}\includegraphics[width=0.33\textwidth]{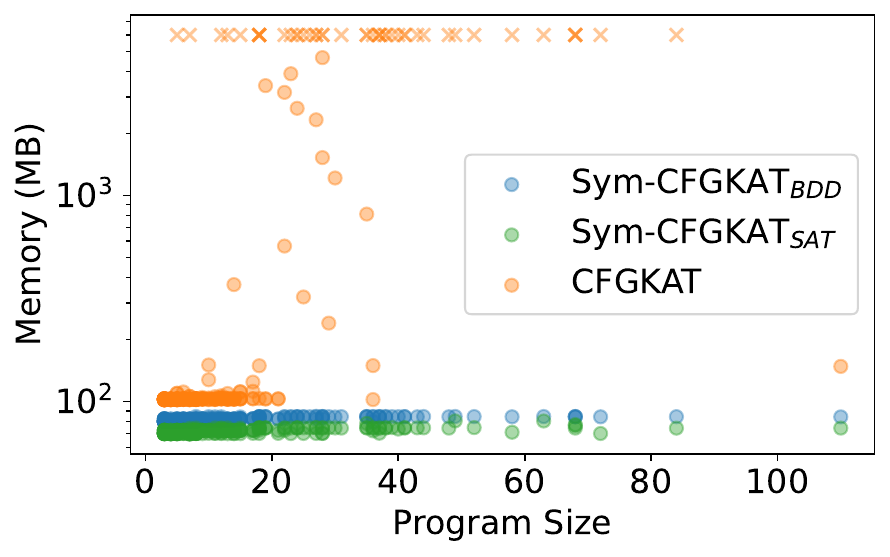} &
  \hspace{-0.7em}\includegraphics[width=0.33\textwidth]{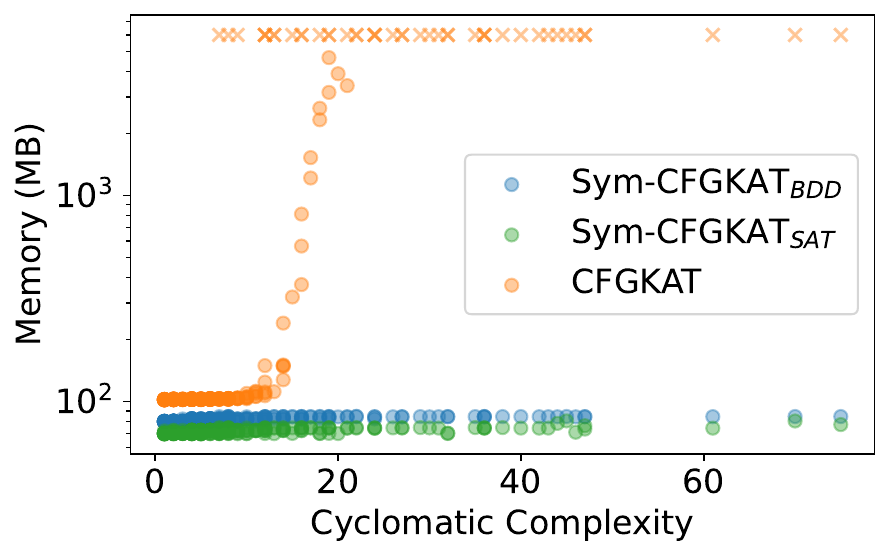}\\[2.8ex]
  \end{tabular}
  \vspace{-5pt}
  \caption{Performance comparison of Symbolic \cfgkat and \cfgkat on Coreutils.}% in \underline{lo}g\underline{arithm scale} on the vertical axis.}
  \label{fig:coreutils-performance-comparison}
  \vspace{-.4cm}
\end{figure}

\Cref{fig:coreutils-performance-comparison} shows the performance of our symbolic \cfgkat algorithm against the original \cfgkat implementation on 237 Coreutils functions. The {\color{orange}`X'} markers in orange indicate that either the original \cfgkat implementation failed to check these test cases within the time or memory limit or encountered unsupported control-flow constructs, like \(\comCont\). 
SymKAT is not included in this comparison as it does not support the control-flow structures in \cfgkat. 

Our symbolic \cfgkat implementation outperforms the original one on all metrics by significant margins, regardless of the Boolean solver. 
In fact, both variants of symbolic \cfgkat are able to check the almost all of the 237 functions in approximately 7 seconds and within a memory footprint of less than 100MB. 
In this particular application, the BDD-based symbolic \cfgkat algorithm performs slightly better than miniSAT in runtime, as the size of the conditions in Coreutils are usually below where SAT overtakes BDD in efficiency. 

\myparagraph{Discovered Bug in Ghidra}
During benchmarking of the implementation, we identified a bug in Ghidra~\cite{nationalsecurityagency_possbile_2025}, which is confirmed and fixed a month later.

This bug is identified by a mismatch of control‑flow between the blinded code for the \texttt{copy\_internal} function in Coreutils 9.5, and its compiled and decompiled output. 
The original \texttt{copy\_internal} function spans 1195 lines, and the blinded one spans 559 lines.
Without the efficiency of our algorithm, testing control‑flow equivalence for functions of this size would be difficult.
Moreover, the efficiency of our tool provides rapid feedback when shrinking the code, allowing us to quickly isolate a minimal counter‑example for submission to the developer.
The minimal input code that triggers the bug contains multiple \command{goto} statements and indicator variables. 
Thus, supports for these complex control-flow structures are required to identify bugs like these. 
The fix for this bug takes roughly 100 lines of code changes, including modifications to the underlying logic.

This experience confirms that our approach can uncover complex bugs in widely used decompilers, even when we didn't target bug-finding in our experiments. 
It demonstrates the practicality of our method and highlights the challenges in implementing correct decompilers, which our tools are well‑equipped to facilitate.
In the future, we plan to design more comprehensive fuzzing frameworks to evaluate additional decompilers across a broader set of input programs.

\section{Discussion}
We presented new symbolic decision procedures for \gkat, leveraging SAT solvers for equivalence checking. We evaluated our procedures  on a series of benchmarks that show significant performance improvement over existing implementations for \CFOrGKAT. 
Notably, our experiments also revealed a bug in Ghidra, an industry-standard decompiler, highlighting the practicality of our algorithms.

Symbolic techniques have been successfully applied to various problems surrounding (extended) regular expressions, with a wide range of real-world applications, including but not limited to low level program analysis~\cite{dallapreda_AbstractSymbolicAutomata_2015a}, list comprehension~\cite{saarikivi_FusingEffectfulComprehensions_2017}, constraint solving~\cite{stanford_SymbolicBooleanDerivatives_2021}, HTML decoding, malware fingerprinting, image blurring, location privacy~\cite{veanes_SymbolicFiniteState_2012}, regex processing, and string sanitizer~\cite{veanes_ApplicationsSymbolicFinite_2013}. Our paper deviates from this line of research in the following ways: first, the problem of normalization is unique to \gkat and not present in other works; second, \gkat automata have a unique shape, where accepting states are replaced by functions on atoms, similar to automata on guarded strings~\cite{kozen_CoalgebraicTheoryKleene_2017}.
The closest symbolic work to ours is by~\citet{pous_SymbolicAlgorithmsLanguage_2015} and~\citet{moeller_KATchFastSymbolic_2024}, which provide symbolic algorithm for \kat and NetKAT, respectively, using BDDs. 

Our work demonstrates the efficiency and flexibility of Boolean formulae, and it would be interesting to explore whether our approach can be applied to extensions of \kat and not just \gkat. Further research is also needed to apply our method in other extensions of \gkat other than \cfgkat.
%
%
%%A more generic symbolic framework is proposed by~\citet{bonchi_CoalgebraicSymbolicSemantics_2009}, who examined the use of symbolic coalgebra to give compact semantics for systems. 
%%Whereas, our work focus on improving the performance of \gkat and \cfgkat, thus we place a special focus on efficient algorithms for automata generation and equivalence checking.
%
%% \myparagraph{\gkat and Applications}
%% The system of \gkat also stem from the research of software-defined networks~\cite{smolka_ScalableVerificationProbabilistic_2019,smolka_GuardedKleeneAlgebra_2020}, and is quickly applied to probabilistic verification~\cite{ro.zowski_ProbabilisticGuardedKAT_2023}, probabilistic program logic~\cite{gomes_KleeneAlgebraTests_2024}, networks~\cite{wasserstein_GUARDEDNETKATSOUNDNESS_2023}, weighted computation~\cite{koevering_WeightedGKATCompleteness_2025}, and control-flow validation~\cite{zhang_CFGKATEfficientValidation_2025}.
%
% \alex{and so?}
%\section{Conclusion and Future Work}
For example, the decision procedures for weighted and probabilistic \gkat~\cite{ro.zowski_ProbabilisticGuardedKAT_2023,koevering_WeightedGKATCompleteness_2025} both rely on coalgebraic partition refinement~\cite{deifel_GenericPartitionRefinement_2019}. 
It would be interesting to see if we can generalize our symbolic techniques to provide support for (coalgebraic) partition refinement.
Another challenging extension is (Guarded) \textsf{NetKAT}~\cite{anderson_NetKATSemanticFoundations_2014,wasserstein_GUARDEDNETKATSOUNDNESS_2023}, which allows packet field assignments and tests, similar to indicator assignments and tests in \cfgkat. 
However, the indicator state in \cfgkat is not symbolic, and hence our current algorithm cannot be easily adapted to improve the performance of (Guarded) \textsf{NetKAT}. 
This inefficiency is mostly inconsequential when testing Ghidra, as indicator variables are seldom produced by its decompilation.
Nevertheless, it would be interesting to see if we can develop a symbolic version of \textsf{NetKAT} automata~\cite{foster_CoalgebraicDecisionProcedure_2015} with Boolean formulae, and evaluate it against KATch~\cite{moeller_KATchFastSymbolic_2024}, an equivalence checker based on extended BDDs.

\ifFull
\section*{Acknowledgement}

We are grateful to our POPL and ESOP reviewers who provided valuable comments, and helped us improve our paper significantly.
The symbolic derivative for \cfgkat is based on a not-yet-published non-symbolic version, developed in collaboration with Wojciech Różowski from Lean FRO (work performed at University College London) and Spencer Van Koevering from Cornell University.
We gratefully acknowledge their significant contribution to this paper.

C. Zhang's work was partially supported by the ERC grant Autoprobe (no. 101002697) and National Science Foundation Award No. 1845803 and No. 2040249.
M. Gaboardi’s work was partially supported by the National Science Foundation under Grant No. 2314324.
A. Silva's work was supported by the Naval Surface Warfare
Center-Corona Division (NSWC-Corona) under Contract No. N6426725C7016.
Any opinions, findings and conclusions or recommendations expressed in this material
are those of the author(s) and do not necessarily reflect the views of the
NSWC-Corona.
\fi

\section*{Data Availability Statement}

The code for our implementations and test suites are included in our artifact~\cite{fu_EfficientDecisionProcedures_2026}.
The software and source code used in evaluating our project, like Ghidra~\cite{NationalSecurityAgency_Ghidra_2025}, Clang~\cite{llvm_LLVMLlvmLlvmproject_2025}, SymKAT~\cite{pous_SymbolicAlgorithmsLanguage_2015}, previous \cfgkat implementation~\cite{zhang_CFGKATEfficientValidation_2025,zhang_2024_13938565}, and CoreUtils~\cite{gnuprojectCoreutilsGNUCore2024} are all publically available and open-source.

%% Bibliography
\renewcommand{\refname}{References}
\makeatletter
\renewcommand{\bibsection}{%
   \section*{\refname%
            \@mkboth{\MakeUppercase{\refname}}{\MakeUppercase{\refname}}%
   }
}
\makeatother
\bibliographystyle{splncs04nat}
\bibliography{refs}

\ifFull

\newpage
\appendix

\section{Correctness of Non-symbolic Algorithm And Up-to Technique}\label[apx]{sec:non-symb-corr-proof}

\subsection{Progression and Compatiblity}

Before we prove the correctness of the non-symbolic version, we will introduce a more generic notation for the checks performed on in algorithm~\Cref{alg:nonsymb-bisim}.
In this definition we range over a generic predicate \(P\) instead of just \(\algIsDead\) as in~\Cref{alg:nonsymb-bisim}; in this way, we have a uniform notation for the trace equivalence algorithm and bisimulation equivalence algorithm, where \(\algIsDead\) is replaced by the constant false function.

\begin{definition}[Progression]\label{def:alg-progression}
  Given two automata over state sets \(S\) and \(U\), and two relations \(R, R' \subseteq S \times U\). 
  We write \(R \transCorr{P}_\alpha R'\) when all the following conditions hold:
  \begin{enumerate}
    \item \(s \transOut{\alpha} \retc\) if and only if \(u \transOut{\alpha} \retc\);  
    \item if \(s \transvia{\alpha \mid p} s'\) and \(u \transRej{\alpha}\), then \(P(s')\);  
    \item if \(u \transvia{\alpha \mid p} u'\) and \(s \transRej{\alpha}\), then \(P(u')\); 
    \item if \(s \transvia{\alpha \mid p} s'\) and \(u \transvia{\alpha \mid q} u'\) and \(p \neq q\), then both \(P(s')\) and \(P(u')\);
    \item if \(s \transvia{\alpha \mid p} s'\) and \(u \transvia{\alpha \mid p} u'\) then \((s', u') \in R'\).
  \end{enumerate}
  We write \(R \transCorr{P} R'\) when \(R \transCorr{P}_\alpha R'\) for all \(\alpha \in \At\), and we call the relation \(R\) an \emph{invariant} if \(R \transCorr{P} R\).
\end{definition}

The definition of \(R \transCorr{P} R'\) parameterize over a predicate \(P\).
In the case of trace equivalence, the predicate \(P\) is instantiated to \(\algIsDead\), then \(R \transCorr{\algIsDead} R'\) exactly represents the condition checked by the non-symbolic algorithm~\Cref{alg:nonsymb-bisim};  
when \(P\) is instantiated to the constant false function, then \(R \transCorr{\false} R'\) will represent the conditions checked by the version of the algorithm that checks for bisimulation equivalence, where all the \(\algIsDead\) and \(\algKnownDead\) is replaced by the constant false function.

The name \emph{progression} is not chosen at random, \citet{pous_CompleteLatticesUpTo_2007} give a general criteria for progressions that supports the up-to techniques in a complete lattice.
Indeed, as \(R\) and \(R'\) are relations, they live in a complete lattice ordered by set inclusion.

\begin{theorem}[Progression Soundness]\label{thm:prog-sound}
  \emph{Progression} as defined in~\Cref{def:alg-progression} indeed satisfies the two requirements set by~\citet{pous_CompleteLatticesUpTo_2007}, which are listed below:
  \begin{mathpar}
    \inferrule
    {R_1 \subseteq R_0 \\ R_0 \transCorr{P} R_0' \\ R_0' \subseteq R_1'}
    {R_1 \transCorr{P} R_1'}
    \and  
    \inferrule
    {\forall i \in I,~ R_i \transCorr{P} R'}
    {(\textstyle{\bigcup_{i \in I} R_i}) \transCorr{P} R'}
  \end{mathpar}
\end{theorem}

\begin{proof}
  The proof is simply by unfolding the five conditions in~\Cref{def:alg-progression}.

  To prove \(R_1 \transCorr{P} R_1'\) in the first requirement, we will need to show all five conditions in~\Cref{def:alg-progression} holds, take any pair of states \((s, u) \in R_1 \subseteq R_0\), we will need to show the transition of \((s, u)\) follows all the five conditions required for \(R_1 \transCorr{P} R_1'\):
  \begin{enumerate}
    \item \(s \transOut{\alpha} \retc\) if and only if \(u \transOut{\alpha} \retc\), because \((s, u) \in R_0\) and \(R_0 \transCorr{P} R_0'\).
    \item If \(s \transvia{\alpha \mid p} s'\) and \(u \transRej{\alpha}\), then \(P(s')\), because \((s, u) \in R_0\) and \(R_0 \transCorr{P} R_0'\).
  \end{enumerate}
  And the rest of the cases follows exactly the same pattern.

  Then we show \(\bigcup_{i \in I} R_i \transCorr{P} R'\), take any pair of states \((s, u) \in \bigcup_{i \in I} R_i\), then there exists \(j \in I\), s.t. \((s, u) \in R_j\).
  Thus, all the five condition indeed hold for \((s, u)\), because \((s, u) \in R_j \transCorr{P} R'\).
\end{proof}

We first prove that the equivalence closure \(e\) is a compatible function with the progression we have developed, this allows us to use the up-to technique in the correctness theorem of the algorithms (\Cref{thm:inf-trace-corr,thm:fin-trace-corr}).

\begin{theorem}
  Reflexive closure \(\refClo\) and symmetric closure \(\symClo\) are compatible with the progression \(R \transCorr{P} R'\) for \(R, R' \subseteq S \times S\).
\end{theorem}

\begin{proof}
  Recall that a function \(f\) is compatible with the progression \(\transCorr{P}\) if 
  \[R \transCorr{P} R' \implies f(R) \transCorr{P} f(R')\]
  which because our progression is sound (\Cref{thm:prog-sound}), we only need to show for all \((s, u) \in f(R)\), then \(\{(s, u)\} \transCorr{P} f(R')\),
  then the conclusion \(f(R) \transCorr{P} f(R')\) can be reached by unioning all of these singleton set \(\{(s,u)\}\).

  \noindent
  \textbf{\(\refClo\) is compatible:}
  since \(\refClo(R) = R \cup \{(s, s) \mid s \in S\}\),
  \begin{itemize}[nosep]
    \item If \((s, u) \in R\), then \(\{(s, u)\} \transCorr{P} R'\), and because \(R' \subseteq \refClo(R')\), by definition of progression, we get \(\{(s, u)\} \transCorr{P} \refClo(R')\).
    \item If \((s, u)\) for \(s = u\), then there transitions needs will satisfy all the rules of~\Cref{def:alg-progression}: rule 1 and rule 5 is trivial, the premise of rule 2, 3, 4 will always be false.
  \end{itemize}

  \noindent
  \textbf{\(\symClo\) is compatible:}
  since \(\symClo(R) = R \cup \{(u, s) \mid (s, u) \in R\}\),
  \begin{itemize}[nosep]
    \item If \((s, u) \in R\), then \(\{(s, u)\} \transCorr{P} R'\), and because \(R' \subseteq \symClo(R')\), by definition of progression, we get \(\{(s, u)\} \transCorr{P} \symClo(R')\).
    \item If \((s, u) \in R\), then \(\{(u, s)\} \transCorr{P} \symClo(R')\), because all the conditions in our progression are symmetric: condition 1 is true for \((u, s)\) because condition 1 is true for \((s, u)\); condition 2 is true for \((u, s)\) because condition 3 is true for \((s, u)\); similarly for the rest of the conditions.
  \end{itemize}
\end{proof}

However, the situation for transitive closure is a bit more convoluted, as the transitive closure \(\transClo\) is \emph{not} compatible with our instance of progression \(\transCorr{P}\), even in the special case where \(P = \algIsDead\).

\begin{example}[Transitive Closure is not Compatible]\label{exp:trans-clos-incompatible}
  Consider a \gkat automaton with only one atom \(1\) (i.e. have no primitive tests), and the following states:
  \[s_0 \transRej{1} \qquad 
  s_1 \transvia{1 \mid p} s_1 \qquad
  s_2 \transvia{1 \mid p} s_3 \qquad
  s_3 \transOut{1} \retc\]
  then let \(R \triangleq \{(s_0, s_1), (s_1, s_2)\}\) and \(R' \triangleq \{(s_1, s_3)\}\): even though \(R \transCorr{\algIsDead} R'\), the transitive closure \(\transClo(R) = \{(s_0, s_1), (s_1, s_2), (s_0, s_2)\}\) does not progress to \(\transClo(R') = \{(s_1, s_3)\}\), as \(s_0 \transRej{1}\), yet \(s_2\) transitions to a live state \(s_3\).
  Thus,
  \[R \transCorr{\algIsDead} R' 
  \quad\text{does \emph{not} imply}\quad
  \transClo(R) \transCorr{\algIsDead} \transClo(R')\ .\]
  For people familiar with the work of~\citet{pous_CompleteLatticesUpTo_2007}, this example also shows that the progression \(\transCorr{\algIsDead}\) does not preserve the internal monoid of relations, where the monoidal operation is relation composition.
\end{example}

\subsection{Correctness of the Algorithm for Infinite-trace Equivalence}\label{ap:inf-trace-equiv-alg-correct}

In fact, the core difficulty behind the incompatibility of transitive closure is that the progression \(\transCorr{\algIsDead}\) does not preserve relation composition \(\fatsemi\), as demonstrated by~\Cref{exp:trans-clos-incompatible}:
\[R_1 \transCorr{\algIsDead} R_1' 
\text{ and }
R_2 \transCorr{\algIsDead} R_2'
\quad\text{does \emph{not} imply}\quad
R_1 \fatsemi R_2 \transCorr{\algIsDead} R_1' \fatsemi R_2'\ .\]
However, the progression \(\transCorr{\false}\), used for infinite-trace equivalence, does preserve composition, which we will demonstrate below, and show how the preservation of composition leads to compatiblity of transitive closure.

\begin{lemma}[Preservation of Composition]
  The progression \(\transCorr{\false}\) preserves relation composition \(\fatsemi\):
  \[R_1 \transCorr{\false} R_1'
  \quad\text{and}\quad
  R_2 \transCorr{\false} R_2'
  \quad\text{implies}\quad
  R_1 \fatsemi R_2 \transCorr{\false} R_1' \fatsemi R_2'\ .
  \]
\end{lemma}
\begin{proof}
  We show that for all \((s, u) \in R_1 \fatsemi R_2\) and \(\alpha \in \At\), \(\{(s, u)\} \transCorr{\false} R_1' \fatsemi R_2'\).
  Recall that \((s, u) \in R_1 \fatsemi R_2\) if and only if there exists an \(m\) s.t. \((s, m) \in R_1\) and \((m, u) \in R_2\): 
  \begin{itemize}[nosep]
    \item Condition 1 holds because \(s \transOut{\alpha} \retc\) iff \(m \transOut{\alpha} \retc\) iff \(u \transOut{\alpha} \retc\).
    \item Condition 2: assume \(s \transRej{\alpha}\) and \(u \transvia{\alpha \mid p} u'\), we will derive a contradiction by analyzing the transition of \(m\):
    \begin{itemize}[nosep]
      \item if \(m \transOut{\alpha} \retc\), then there is a contradiction because \(s \transRej{\alpha}\) and \((s, m) \in R_1\).
      \item if \(m \transRej{\alpha}\), then there is a contradiction because \(u \transvia{\alpha \mid p} u'\) and \((m, u) \in R_2\);
      \item if \(m \transvia{\alpha \mid q} m'\), then there is a contradiction because \(s \transRej{\alpha}\) and \((s, m) \in R_1\);
    \end{itemize}
    \item The proof of condition 3 is symmetric to condition 2.
    \item Condition 4: assume \(s \transvia{\alpha \mid p} s'\) and \(u \transvia{\alpha \mid q} u'\) and \(p \neq q\), we can derive a contradiction by case analysis on the transition of \(m\), similar to condition 2.
    \item Condition 5: when \(s \transvia{\alpha \mid p} s'\) and \(u \transvia{\alpha \mid p} u'\), we will show \(m \transvia{\alpha \mid p} m'\), by case analysis on the transition of \(m\):
    \begin{itemize}[nosep]
       \item if \(m \transOut{\alpha} \retc\), then because \(s \transvia{\alpha \mid p} s'\) and \((s, m) \in R_1\), contradicting condition 1 for \(R_1\);
      \item if \(m \transRej{\alpha}\), then because \(s \transvia{\alpha \mid p} s'\) and \((s, m) \in R_1\), a contradiction can be derived via condition 2;
      \item if \(m \transvia{\alpha \mid q} m'\), and \(p \neq q\), then because \(s \transvia{\alpha \mid p} s'\) and \((s, m) \in R_1\), a contradiction cna be derived via condition 4;
    \end{itemize}
    Therefore the only possible transition is \(m \transvia{\alpha \mid p} m'\), and by condition 5 for both \(R_1\) and \(R_2\), we conclude \((s', m') \in R_1'\) and \((m', u') \in R_2'\).
    Then, by definition of relation composition, \((s', u') \in R_1' \fatsemi R_2'\).
  \end{itemize}
\end{proof}

\begin{theorem}
  The transitive closure \(\transClo\) is compatible with \(\transCorr{\false}\):
  \[
    R \transCorr{\false} R' 
    \qquad\text{implies}\qquad
    \transClo(R) \transCorr{\false} \transClo(R')
  \]
\end{theorem}

\begin{proof}
  Recall that \(\transClo(R)\) is defined as \(\bigcup \{R^n \mid n \in \mathbb{N}, n \geq 1\}\).
  where \(R^1 \triangleq R\) and \(R^{k+1} \triangleq R \fatsemi R^{k}\).
 
  Because \(\transCorr{\false}\) preserves relation composition, we can show that \(R^{n} \transCorr{\false} R^{n}\) for all \(n \geq 1\) by a simple induction.
  And we can derive the desired result by the following chain of implications:
  \begin{align*}
    & \forall n \geq 1, R^{n} \transCorr{\false} R^{n} 
      & \text{simple induction}\\
    \text{implies}\quad &
      \forall n \geq 1, R^{n} \transCorr{\false} \transClo(R)
      & \text{\Cref{thm:prog-sound}, \(R^{n} \subseteq \transClo(R)\)} \\
    \text{implies}\quad &
      \bigcup_{n \geq 1} R^{n} \transCorr{\false} \transClo(R)
      & \text{\Cref{thm:prog-sound}} \\
    \text{implies}\quad &
      \transClo(R) \transCorr{\false} \transClo(R) 
      & \text{definition of \(\transClo\)}
  \end{align*}
\end{proof}

\begin{corollary}[Compatiblity of Equivalence Closure]\label{thm:equiv-clo-compatible-bisim}
  Because the refexive closure \(\refClo\), symmetric closure \(\symClo\), and transitive closure \(\transClo\) are compatible with the progression \(\transCorr{\false}\), then the equivalence closure \(e = \refClo \circ \symClo \circ \symClo\) is also compatible.
  More concretely \(R \transCorr{\false} e(R)\) implies \(R\) is contained in an invariant of \(\transCorr{\false}\).
\end{corollary}

The first proof we give is bisimulation version of the algorithm, where the checked condition can be represented as \(R \transCorr{\false} R'\).
This version is easier to prove because, by definition, \(R\) is a bisimulation relation if and only if \(R \transCorr{\false} R\), that is when \(R\) is an invariant of \(\transCorr{\false}\).
Indeed, if \(P\) is false, then condition (2)-(4) will be false when their premises are true; thus the pair \(s, u\) either both return (condition (1)) or both transition and execute the same action (condition (5)).
Additionally, the optimization from lines 4-7 will never be invoked.

\begin{theorem}[Correctness of Infinite-trace Equivalence]\label{thm:inf-trace-corr}
  When the \(\algKnownDead\) and \(\algIsDead\) is replaced by constant false function, then \(\algEquiv(s, u)\) will always terminate if the input automata are finite, and \(\algEquiv(s, u)\) will return true if and only if \(s\) and \(u\) are bisimilar.
\end{theorem}

\begin{proof}
  We denote all pairs of the inputs for function \algUnion{}  during the program execution as \(R \subseteq S \times U\).
  Recall that union-find generates the least equivalence relation from all the inputs of \algUnion{} , which we denote as \(e(R)\).
  We note that \(e(R)\) is an equivalence relation on the disjoint union \(S + U\), and this relation coincide with the equivalence of representatives:
  \[e(R) = \{(s, u) \mid \algRep(s) = \algRep(u)\}.\]

  We first obtain termination: 
  because the state sets \(S\) and \(U\) are finite, thus so is \(S \times U\).
  Then notice if \((s, u) \notin e(R)\), then \(s, u\) is sent into \(\algUnion{} \) on line 3, making \((s, u) \in R \subseteq e(R)\) before the recursive call on line 13 in~\Cref{alg:nonsymb-bisim}.
  Therefore, at each recursive call of \algEquiv, the size of the finite set \((S \times U) \setminus e(R)\) will decrease.

  As for the correctness, we will construct an inductive invariant for the program, i.e. a triple that holds for the program if it holds for all recursive calls.
  For readability, we will split the invariant into two cases (formally, the two invariants translates to the conjunction of two implications): the first case is when the return is true; the second is when it is false.

  First we focus on the case when \(\algEquiv(s, u) = \true\).
  We claim the following post condition holds under the precondition \(R_0 = R\), where \(R_0\) is the auxiliary variable to record the starting \(R\):
  \[(s, u) \in R  \text{ and }  (R \setminus R_{0}) \transCorr{\false} e(R).\]
  Thus, if the start \(R_{0}\) is empty, i.e. starting from an empty union-find object, then we obtain \(R \transCorr{\false} e(R)\) at the end of the execution.
  Because \(R\) is an invariant of \(\transCorr{\false}\) if and only if \(R\) is a bisimulation; and equivalence \(e\) is compatible with \(\transCorr{\false}\) (\Cref{thm:equiv-clo-compatible-bisim}), then the ending condition \(R \transCorr{\false} e(R)\) implies \((s, u) \in R\) is in a bisimulation between \(S\) and \(U\).

  We proceed to sketch the proof of the above inductive invariant when \(\algEquiv(s, u) = \true\).
  We let \(R_{1} = R_{0} \cup \{(s, u)\}\), then the following condition is a loop invariant for the for-all loop on lines 8-13:
  \[(s, u) \in R \text{ and } (R \setminus R_{1}) \transCorr{\false} e(R) \text{ and } \forall  \alpha \in \Lambda , (s, u)  \rightarrowtail_\alpha  e(R),\]
  where \(\Lambda\) is the explored atoms in the loop.
  After the loop finishes, \(\Lambda\) will just be \(\At\), thus we can obtain the following conclusion:
  \[(s, u) \in R \text{ and } (R \setminus R_{1}) \transCorr{\false} e(R) \text{ and } (s, u) \transCorr{\false} e(R).\]
  Because \(R_{1} = R_{0} \cup \{(s, u)\}\), we obtain our conclusion: \((s, u) \in R \text{ and } (R \setminus R_{0}) \transCorr{\false} e(R).\)

  Second, we consider the case when \(\algEquiv(s, u) = \false\).
  We show the invariant that \(s \nsim u\) in the post-condition, which implies \(s\) and \(u\) are not bisimilar.
  We note that since \algKnownDead{} and \algIsDead{} are replaced with constant false function, therefore in order for \(\algEquiv(s, u) = \false\), there exists an atom s.t. at least one of the following is true:
  \begin{itemize}[nosep]
    \item one of \(s\) or \(u\) accepts, but not the other;
    \item one of \(s\) or \(u\) transitions, but the other transitions; 
    \item \(s \transvia{\alpha \mid p} s'\) and \(u \transvia{\alpha \mid q} u'\), but either \(p \neq q\) or \(\algEquiv(s', u') = \false\).
  \end{itemize}
  By the inductive invariant, \(\algEquiv(s', u') = \false\) implies \(s' \nsim u'\). 
  It is not hard to see if any of the condition above is violated, then \(s, u\) violates the condition of the bisimulation. 
  In other word, if \(s \sim u\) then the relation \(\sim\) cannot be a bisimulation, concluding the desired post-condition \(s \nsim u\).
\end{proof}

\subsection{Correctness of the Algorithm for Finite-trace Equivalence}\label{ap:trace-equiv-alg-correct}

Unfortunately, the correctness of the finite-trace semantics case is not as easy as the bisimulation (or infinite-trace) case, mainly for the following two reasons.
First, since we no longer explicitly construct a bisimulation with our algorithm, it is hard to demonstrate every invariant \(R\) identifies a bisimulation on normalized coalgebra, which is at the core of the correctness proof of~\citet{smolka_GuardedKleeneAlgebra_2020}. 
In particular, to show every invariant of \(\transCorr{\algIsDead}\) is a bisimulation on the normalized coalgebra, we will need to show \(s'\) and \(u'\) in the last condition have the same liveness; yet, we only know \((s', u') \in R\).
Second, we will need to use the up-to technique to show that if \(R \transCorr{\algIsDead} e(R)\) implies \(R\) is in an invariant.
Yet, as we have shown in~\Cref{exp:trans-clos-incompatible}, transitive closure, and thus equivalence closure, is \emph{not} compatible with \(\transCorr{\algIsDead}\).

Perhaps quite surprisingly, the easiest way to mitigate both of these problems is to work directly with trace equivalence instead of bisimulation on normalized coalgebra.
Indeed, by reasoning about trace, it gives us a strong enough induction hypothesis to prove both that trace equivalence is the maximal invaraint, and the transitive closure \(\transClo\) is sound (called ``correct'' by~\citet{pous_CompleteLatticesUpTo_2007}) with respect to \(\transCorr{\algIsDead}\).

Before we prove our core result, we will first recall some common lemma around traces in GKAT automata.

\begin{lemma}\label{thm:live-iff-have-trace}
  A state \(s\) is live if and only if it has at least one trace.
\end{lemma}

\begin{proof}
  If \(s\) is live, then there exists a chain to a accepting state \(s_n\):
  \[s \transvia{\alpha_1 \mid p_1} s_1 \transvia{\alpha_2 \mid p_2} \dots \transvia{\alpha_n \mid p_n} s_n \transOut{\alpha} \retc,\]
  Then by an induction argument, we can show that the trace \(\alpha_1 p_1 \alpha_2 p_2 \dots \alpha_n p_n \alpha\) is a trace in \(\sem{s}\).

  In the reverse direction, if we have a trace \(\alpha_1 p_1 \alpha_2 p_2 \dots \alpha_n p_n \alpha \in \sem{s}\), then a path to accepting state can be recovered by computing the transitions on \(\alpha_1, \alpha_2, \alpha_3\) and so on:
  \[s \transvia{\alpha_1 \mid p_1} s_1 \transvia{\alpha_2 \mid p_2} \dots \transvia{\alpha_n \mid p_n} s_n \transOut{\alpha} \retc,\]
  which implies that \(s\) is live.
\end{proof}

% \begin{lemma}\label{thm:dead-imply-trans-dead}
%   If a state \(s'\) is live in an automaton and \(s \transvia{\alpha \mid p} s'\) for some \(\alpha\) and \(p\), then \(s\) is also live.
% \end{lemma}

% \begin{proof}
%   Because \(s'\) is live, then there exists a sequence of transition 
%   \[s' \transvia{\alpha_1 \mid p_1} s_1 \transvia{\alpha_2 \mid p_2} \dots \transvia{\alpha_n \mid p_n} s_n,\]
%   where \(s_n\) is an accepting state.
%   Then \(s\) can also reach the accepting state \(s_n\) via \(s'\):
%   \[s \transvia{\alpha \mid p} s' \transvia{\alpha_1 \mid p_1} s_1 \transvia{\alpha_2 \mid p_2} \dots \transvia{\alpha_n \mid p_n} s_n,\]
%   showing that \(s\) is also live.
% \end{proof}

\begin{lemma}[Trace Determinacy~\cite{smolka_CantorMeetsScott_2017}]\label{thm:trace-deter}
  Given a state \(s\) and a starting atom \(\alpha\), there exists at most one trace that starts with \(\alpha\) in \(\sem{s}\).
\end{lemma}

\begin{remark}
  The general construction outlined by~\citet{bonchi_CoinductionUptoFibrational_2014a} does not easily apply in this situation, because (1) we work with GKAT automata, instead of a bialgebra; (2) as we have mentioned in the beginning of this sub-section, it is hard to show that every invariant is a bisimulation.
  Finally, equivalence closure is compatible with progressions defined by~\citet{bonchi_CoinductionUptoFibrational_2014a}, yet we have shown in~\Cref{exp:trans-clos-incompatible}, transitive closure, and thus equivalence closure, is not compatible with our progression \(\transCorr{\algIsDead}\).
\end{remark}

We will then prove several core theorems that powers our correctness proof. 
Specifically, we show that \textbf{trace equivalence is the largest invariant of \(\transCorr{\algIsDead}\)} and \textbf{transitive closure \(\transClo\)is sound\footnote{named ``correct'' by~\citet{pous_CompleteLatticesUpTo_2007}, we call it ``sound'' to avoid confusion with the ``correctness'' of our algorithm} with respect to \(\transCorr{\algIsDead}\)}.
Interestingly, both of these results can be derived by the following lemma.

\begin{lemma}\label{thm:trace-equiv-trans-clo-invariant-core-lemma}
  Consider two relation \(R, R' \subseteq S \times U\) on two \gkat automata \(S\) and \(U\), if for all traces \(w\):
  \[
    \forall (s, u) \in R, w \in \sem{s} \text{ iff } w \in \sem{u}
    \qquad\text{implies}\qquad
    \forall (s, u) \in R', w \in \sem{s} \text{ iff } w \in \sem{u}
  \]
  then \(R \transCorr{\algIsDead} R'\) and \((s, u) \in R\) implies \(\sem{s} = \sem{u}\).
\end{lemma}

\begin{proof}
  We will show that for any pair of states \((s, u) \in R\) and any trace \(w\), \(w \in \sem{s}\) if and only if \(w \in \sem{u}\), by induction on the length of \(w\).
  \begin{itemize}
    \item 
    If \(w\) is a single atom \(\alpha\), then by the first condition, we got the following chain of implications:
    \[\alpha \in \sem{s} 
    \quad\text{iff}\quad s \transOut{\alpha} \retc 
    \quad\text{iff}\quad u \transOut{\alpha} \retc
    \quad\text{iff}\quad \alpha \in \sem{u}.\]
    \item
    If \(w\) is \(\alpha p \cdot w' \in \sem{s}\), then we know that \(s \transvia{\alpha \mid p} s'\) for some \(p \in \Sigma\), \(s' \in S\) and \(w' \in \sem{s'}\).
    By~\Cref{thm:live-iff-have-trace}, \(s'\) is a live state.
    By induction hypothesis, for all \((s', u') \in R\), \(w' \in \sem{s'}\) iff \(w' \in \sem{u'}\).
    Finally, by the implication in the premise, for all \((s', u') \in R'\), \(w' \in \sem{s'}\) iff \(w' \in \sem{u'}\).

    We show \(\alpha p \cdot w \in \sem{u}\) by case analysis on the transition of \(u\) with atom \(\alpha\): 
    \begin{itemize}
      \item if \(u \transOut{\alpha} \retc\), then by the condition 1: \(s \transOut{\alpha} \retc\), violating \(s \transvia{\alpha \mid p} s'\);   
      \item if \(u \transRej{\alpha}\), then by the condition 2, \(s'\) is dead, violating our conclusion that \(s'\) is live;
      \item if \(u \transvia{\alpha \mid q} u'\) and \(p \neq q\), then by condition 4, \(s'\) is dead; again violating our conclusion that \(s'\) is live.
    \end{itemize}
    Thus, the only possible transition \(u\) can take under \(\alpha\) is \(u \transvia{\alpha \mid p} u'\).
    By condition 5, \((s', u') \in R'\); then because \(w' \in \sem{s'}\) iff \(w' \in \sem{u'}\), concluding that \(\alpha p \cdot w \in \sem{u}\).
    The fact that \(\alpha p \cdot w \in \sem{u}\) implies \(\alpha p \cdot w \in \sem{s}\) can be shown by symmetry.
  \end{itemize}
\end{proof}

\begin{lemma}[Invariant Implies Traces Equivalence]\label{thm:invariant-perserve-trace}
  If \(R \transCorr{\algIsDead} R\) then every pair of states in \(R\) are trace-equivalent: for all \((s, u) \in R, \sem{s} = \sem{u}\).
\end{lemma}

\begin{proof}
  We first note the implication in~\Cref{thm:trace-equiv-trans-clo-invariant-core-lemma} holds.
  Indeed,
   \[
    \forall (s, u) \in R, w \in \sem{s} \text{ iff } w \in \sem{u}
    \qquad\text{implies}\qquad
    \forall (s, u) \in R, w \in \sem{s} \text{ iff } w \in \sem{u}
  \]
  Then by~\Cref{thm:trace-equiv-trans-clo-invariant-core-lemma}, we conclude that \(\sem{s} = \sem{u}\)
\end{proof}

% \begin{lemma}\label{thm:alg-preserve-liveness}
%   For a relation \(R\) on the state sets \(S, U\) of two automata, if \(R \transCorr{\algIsDead} R\), then \(R\) preserves liveness: for all \((s, u) \in R\), \(s\) is live if and only if \(u\) is live.
% \end{lemma}

% \begin{proof}
%   Because a state is live if and only if it contains at least one trace (\Cref{thm:live-iff-have-trace}), and progression perserves liveness 
% \end{proof}

% \begin{lemma}\label{thm:trace-equiv-corr-forward-lemma-norm-bisim}
%   Given a relation \(R\) on the state sets of two automata, if \(R \transCorr{\algIsDead} R\), then \(R\) is a bisimulation on the normalized automata.
% \end{lemma}

% \begin{proof}

% \end{proof}

\begin{lemma}[Trace Equivalence is an Invariant]\label{thm:trace-equiv-is-invariant}
  Let \(R_\equiv\) be the trace equivalence relation:
  \[R_{\equiv} \triangleq \{(s, u) \mid (s, u) \in S \times U, \sem{s} = \sem{u}\}\ , 
  \qquad\text{then}\qquad
  R_{\equiv} \transCorr{\algIsDead} R_{\equiv}.\]
\end{lemma}

\begin{proof}
  To show all the conditions for \(R_{\equiv} \transCorr{\algIsDead} R_{\equiv}\) are satisfied, we take any pair of trace equivalence states \((s, u) \in R\) and atom \(\alpha \in \At\), we case analysis whether there exists a trace starting with \(\alpha\) in \(\sem{s}\) and \(\sem{u}\):
  \begin{itemize}
    \item If there is no trace that starts with \(\alpha\) in both \(\sem{s}\) and \(\sem{u}\), then \(s\) and \(u\) either reject \(\alpha\) or transition to a dead state via \(\alpha\):
    \begin{itemize}
      \item if both rejects, then it satisfies all the condition;
      \item if one reject and the other transitions to a dead state, then it also satisfies all the conditions; in particular, this case is mentioned in conditions 2 and 3.
      \item if both transition to dead states \(s'\) and \(u'\), then because \(s'\) and \(u'\) are dead, hence have no trace (\Cref{thm:live-iff-have-trace}), thus trace equivalent \((s', u') \in R_{\equiv}\).
    \end{itemize}
    \item If there is a trace \(\alpha\) in \(\sem{s}\) and \(\sem{u}\), then by definition \(s \transOut{\alpha} \retc\) and \(u \transOut{\alpha} \retc\), satisfying all the five conditions.
    \item If there is a trace \(\alpha p \cdot w\) in \(\sem{s}\) and \(\sem{u}\), then there exists states \(s'\) and \(u'\):
    \[s \transvia{\alpha \mid p} s' \qquad\text{and}\qquad
    u \transvia{\alpha \mid p} u'\ .\]
    By definition of the trace semantics, we can obtain the following equivalences:
    \[w \in \sem{s'} \iff \alpha p \cdot w \in \sem{s} 
    \quad\text{and}\quad
    w \in \sem{u'} \iff \alpha p \cdot w \in \sem{u}\ .\]
    Finally, we show that \(\sem{s} = \sem{u}\), which conclude that this case also satsifies all the conditions in~\Cref{def:alg-progression}:
    \begin{align*}
      w \in \sem{s'} 
      \iff \alpha p \cdot w \in \sem{s}
      \iff \alpha p \cdot w \in \sem{u}
      \iff \alpha w \in \sem{u}
    \end{align*}
  \end{itemize}
  Notice that these are the only cases we need to worry about since \((s, u) \in R_{\equiv}\), implying \(\sem{s} = \sem{u}\).
  Thus, we reach our desired conclusion: \(R_{\equiv} \transCorr{\algIsDead} R_{\equiv}\)

\end{proof}

When combining~\Cref{thm:trace-equiv-is-invariant,thm:invariant-perserve-trace}, we obtain that trace equivalence is indeed the maximal invariant that we are looking for.

\begin{theorem}[Trace Equivalence is The Maximal Invariant]\label{thm:trace-equiv-is-max-invari}
  Given two automata over state sets \(S\) and \(U\), then the trace equivalence relation \(R_{\equiv}\):
  \[R_{\equiv} \triangleq \{(s, u) \mid (s, u) \in S \times U, \sem{s} = \sem{u}\},\]
  is the maximal invariant, i.e. \(R_\equiv \transCorr{\algIsDead} R_\equiv\) and for all \(R \subseteq S \times U\), \(R \transCorr{\algIsDead} R \implies R \subseteq R_{\equiv}\).
\end{theorem}
 
\begin{proof}
  \(R_{\equiv}\) is an invariant by~\Cref{thm:trace-equiv-is-invariant}
  Then we prove that \(R_{\equiv}\) is maximal. Take any invariant \(R \transCorr{\algIsDead} R\), by~\Cref{thm:invariant-perserve-trace}, every pair \((s, u) \in R\) is trace equivalent: \(\sem{s} = \sem{u}\). 
  Therefore, \(R\) is contained with in \(R_{\equiv}\): \(R \subseteq R_{\equiv}\)
\end{proof}

It is important that our property of interest, trace equivalence is the maximal invariant, as it allows us to apply the up-to technique~\cite{pous_CompleteLatticesUpTo_2007}.
In particular, we will require the up-to technique with equivalence closure \(e\) to prove our algorithm being sound, since we utilize a union-find object in the algorithm to keep trace of the explored states.
This optimization enlarges the resulting states of a progression by the equivalence closure.

Unfortunately, transitive closure \(\transClo\) and thus equivalence closure \(e\) is not compatible with \(\transCorr{\algIsDead}\).
We will show a weaker notion next, that is \(\transClo\) is ``sound'' (called ``correct'' by~\citet{pous_CompleteLatticesUpTo_2007}): although a sound function is safe to use in an up-to proof, but unlike compatible function, composition of sound function is not necessary still sound. 
However, post-composing a sound function with a compatible function, still yields a sound Function~\cite[Theorem 1.7]{pous_CompleteLatticesUpTo_2007}. 
This result is enough for us to derive the soundness of the equivalence closure \(e\) to be used in the up-to technique, which is crucial for the correctness of our algorithm.

\begin{theorem}
  The transitive closure \(\transClo\) is sound with respect to \(\transCorr{\algIsDead}\).
\end{theorem}

\begin{proof}
  We claim that it suffices to show that \(R \transCorr{\algIsDead} \transClo(R)\) implies that for all \((s, u) \in R\), \(\sem{s} = \sem{u}\).

  First, we show the above claim indeed implies soundness as defined by~\citet{pous_CompleteLatticesUpTo_2007}, which is \(\nu(m_{\transCorr{}} \circ \transClo) \subseteq \nu(m_{\transCorr{}})\), where 
  \[
    m_{\transCorr{}}(R') \triangleq \bigcup \{R \mid R \transCorr{\algIsDead} R'\}\ ; \qquad
    \nu(f) = \bigcup \{R \mid R \subseteq f(R) \}\ .
  \]
  We prove this fact mostly by unfolding the definitions.
  Recalling a theorem by~\citet{pous_CompleteLatticesUpTo_2007}, the following is true for all \(f\)
  \[R \subseteq m_{\transCorr{}} \circ f(R)  \qquad\iff\qquad R \transCorr{\algIsDead} f(R)\]
  which means \(\nu(m_{\transCorr{}} \circ \transClo)\) is the largest \(R\) s.t. \(R \transCorr{\algIsDead} \transClo(R)\); and \(\nu(m_{\transCorr{}})\) is the largest invariant \(R \transCorr{\algIsDead} R\), which is just trace equivalence \(R_{\equiv}\) by~\Cref{thm:trace-equiv-is-max-invari}.
  Hence, \(\nu(m_{\transCorr{}} \circ \transClo) \subseteq \nu(m_{\transCorr{}})\) simply means that the every \(R\) s.t. \(R \transCorr{\algIsDead} \transClo(R)\) is contained within \(R_{\equiv}\), i.e. for all \((s, u) \in R\), \(\sem{s} = \sem{u}\).

  Then to prove that \(R \transCorr{\algIsDead} \transClo(R)\) implies for all \((s, u) \in R\), \(\sem{s} = \sem{u}\), we only need to use~\Cref{thm:trace-equiv-trans-clo-invariant-core-lemma}, and show that fix any trace \(w\), \(\forall (s, u) \in R, w \in \sem{s} \iff w \in \sem{u}\) implies \(\forall (s, u) \in \transClo(R), w \in \sem{s} \iff \sem{u}\). 
  Recall that \((s, u) \in \transClo(R)\) if and only if there exists a sequence of \(m_i\) s.t.
  \[s \mathrel{R} m_1 \mathrel{R} m_2 \mathrel{R} \cdots \mathrel{R} m_n \mathrel{R} u\ .\]
  Thus, by induction on the length of above chain, we can show that \(w \in \sem{s}\) if and only if \(w \in \sem{u}\).
  Finally, by~\Cref{thm:trace-equiv-trans-clo-invariant-core-lemma}, \(R \transCorr{\algIsDead} \transClo(R)\) implies \((s, u) \in R\) implies \(\sem{s} = \sem{u}\), which then implies that \(\transClo\) is sound with respect to \(\transCorr{\algIsDead}\).
\end{proof}

\begin{corollary}[Soundness of Equivalence Closure]\label{thm:equiv-closure-sound-trace-prog}
  We recall that for any sound function \(f\), \(R \transCorr{\algIsDead} f(R)\) implies \(R\) is in the maximal invariant of \(\transCorr{\algIsDead}\), which is trace equivalence.
  Because \(\transClo\) is sound and post-composition of sound function with compatible function gives sound functions, thus \(e = \refClo \circ \symClo \circ \transClo\) is sound with respect to \(\transCorr{\algIsDead}\).
  In other word, 
  \[R \transCorr{\algIsDead} e(R) \qquad\implies\qquad
  \forall (s, u) \in R, \sem{s} = \sem{u}.\]
\end{corollary}

\begin{remark}\label{rem:coprod-aut}
  To be completely formal, we will need to reason in a different automaton.
  For example, let \(R_{\equiv}'\) be the equivalence relation over \(A_S\) and \(A_U\) with state sets \(S\) and \(U\) respectively, whereas \(e(R)\) is a relation over the disjoint union \(S + U\) and itself.
  In this case, \(R_{\equiv}'\) is not necessarily contained in \(e(R)\) even when \(R\) is an invariant. 
  However, we can parallel compose (a.k.a. coproduct) two automata over state sets \(S\) and \(U\), this yields an automaton over state set \(S + U\).
  The state in the coproduct automata will transition according to the original automata:
  \[\zeta(s, \alpha) \triangleq \begin{cases}
    \zeta_S(s, \alpha) & s \in S \\  
    \zeta_U(u, \alpha) & s \in U
  \end{cases},\]
  where \(\zeta_S\) and \(\zeta_U\) are the transition function of the original automata; therefore the trace of each state also coincides with the original trace.
  Then the trace equivalence relation \(R_{\equiv}\) over the coproduct automata and itself is indeed the greatest invariant of \(\transCorr{\algIsDead}\), hence contains the equivalence closure of any invariant \(R\) over the coproduct automata and \(e(R)\).
\end{remark}

Finally, we have all the tools to prove that the algorithm is sound and complete for trace equivalence, which is the second part of~\Cref{thm:non-symb-corr}

\begin{theorem}[Correctness of Finite-Trace Equivalence]\label{thm:fin-trace-corr}
  \(\algEquiv(s, u)\) will always terminate if the input automata are finite, and \(\algEquiv(s, u)\) will return true if and only if \(s\) and \(u\) are bisimilar.
\end{theorem}

\begin{proof}
  We split the proof into termination and partial correctness. 
 
  \textbf{First, we show the termination of the algorithm.}
  Recall that \(\algIsDead\) is a DFS on the finite state sets; and \(\algKnownDead\) are checking inclusion of the cached dead states, therefore both functions always terminate.
  We can construct the same termination metric as the proof of~\Cref{thm:inf-trace-corr}, \((S \times U) \setminus e(R)\).
  Indeed, by observing the algorithm in~\Cref{alg:nonsymb-bisim}, we realize that if \(s, u \notin e(R)\) then on line 3, we will send them into the function \(\algUnion{} \).
  After the union \((s, u) \in R \subseteq e(R)\), which means that \((S \times U) \setminus e(R)\) will decrease before the recursive call on line 13.

  \textbf{We then proceed with the partial correctness proof.}
  First, we recall that the notation \(R \transCorr{\algIsDead} R'\) for two relation \(R, R' \in S \times U\) when \(R\) transition to \(R'\) following all the conditions in~\Cref{alg:nonsymb-bisim} (see the formal definition of \emph{progression} in~\Cref{def:alg-progression}).

  Like the proof of~\Cref{thm:inf-trace-corr}, we split the inductive invariants into two cases for readability.
  And we also use the relation \(R\) to denote all the inputs that have been sent into the \(\algUnion{}\) function. 
  Let \(e(R)\) be the least equivalence generated by \(R\) on the disjoint union of \(S\) and \(U\). 
  And we recall that this least equivalence coincide with the equivalence of representatives in the union-find object,
  \[e(R) = \{(s, u) \mid \algRep(s) = \algRep(u)\}.\]

  First, we focus on the case where the output \(\algEquiv(s, u) = \true\). 
  Unlike the previous proofs, we will need to take into account of the optimization performed on lines 4-7.
  We assert the following invariant for the recursion: given the precondition \(R_0 = R\), where \(R_0\) is an auxiliary variable to keep track of the starting \(R\), then the following post-condition can be derived:
  \[\algIsDead(s) \text{ and } \algIsDead(u)
  \quad\text{or}\quad
  (s, u) \in R \text{ and } (R \setminus R_0) \transCorr{\algIsDead} e(R).\]
  Indeed, if either \(\algKnownDead(s)\) or \(\algKnownDead(u)\) is true, then it means both \(s\) and \(u\) are dead, reaching the condition \(\algIsDead(s) \text{ and } \algIsDead(u)\).
  If neither of these are true, we let \(R_1 \triangleq R_0 \cup \{(s, u)\}\); we construct the following invariant for the for-all loop on line 8-13:
  \[(s, u) \in R \text{ and }
  (R \setminus R_1) \transCorr{\algIsDead} e(R) \text{ and } 
  \forall \alpha \in \Lambda, \{(s, u)\} \transCorr{\algIsDead}_\alpha e(R),\]
  where \(\Lambda\) is all the explored atoms.
  Then at the end of the loop \(\Lambda\) will just be \(\At\), therefore, we derive 
  \[(s, u) \in R \text{ and }
  (R \setminus R_1) \transCorr{\algIsDead} e(R) \text{ and }
  \{(s, u)\} \transCorr{\algIsDead} e(R),\]
  which gives us the desired post-condition \(R \transCorr{\algIsDead} e(R)\) and \((s, u) \in R\), because \(R \setminus R_0 = R \setminus R_1 \cup \{(s, u)\}\).

  Then for the case where \(\algEquiv(s, u) = \false\), we prove the invariant that under the precondition true, we can derive the post-condition \(\sem{s} \neq \sem{u}\).
  This result is a direct consequence of~\Cref{thm:trace-equiv-is-invariant}: if \(\algEquiv(s, u) = \false\), then one of the following will be true:
  \begin{itemize}
    \item One of \(s\) and \(u\) is dead, but the other is not. Then by~\Cref{thm:live-iff-have-trace}, one of them will have trace, whereas the other don't, hence they are clearly not trace equivalent.
    \item there exists \(\alpha\) s.t. \((s, u) \not\transCorr{\algIsDead}_\alpha R_{\equiv}\), then by~\Cref{thm:trace-equiv-is-invariant}, \(s\) and \(u\) are not trace equivalent.
  \end{itemize}
\end{proof}

Finally, combining the proof of~\Cref{thm:inf-trace-corr} and~\Cref{thm:fin-trace-corr}, we obtain the proof of the correctness for the non-symbolic algorithm.

\begin{proof}[Proof of~\Cref{thm:non-symb-corr}]
  Each item is proven by~\Cref{thm:inf-trace-corr} and~\Cref{thm:fin-trace-corr} respectively.
\end{proof}

\section{Correctness of the Symbolic Algorithm}\label[apx]{sec:symb-corr-proof}

To prove the correctness of the symbolic notation, we first set some notation for \emph{symbolic progression}, which is characterizes the checks performed on the symbolic algorithm in~\Cref{alg:symb-bisim}.

\begin{definition}[Notation for Symbolic Progression]
  Given two symbolic GKAT automata over state space \(S\) and \(U\) respectively, we write \(R \transCorrSymb{P} R'\) for a predicate \(P\) on the disjoint union \(S + U\), when \(R\) transitions according to the conditions from line 9-13 in~\Cref{alg:symb-bisim}: for all \((s, u) \in R\),
  \begin{itemize}
    \item \(\bigvee \epsilon(s) \equiv \bigvee \epsilon(u)\) under axioms of Boolean algebra; 
    \item if \(s \transvia{b \mid p} s'\) and \(b \land \rho(u) \not\equiv 0\), then \(P(s')\);
    \item if \(u \transvia{b \mid p} u'\) and \(b \land \rho(u) \not\equiv 0\), then \(P(u')\);
    \item if \(s \transvia{b \mid p} s'\), \(u \transvia{a \mid q} u'\), \(b \land a \not\equiv 0\), and \(p \neq q\), then \(P(s')\) and \(P(u')\);
    \item if \(s \transvia{b \mid p} s'\), \(u \transvia{a \mid p} u'\), and \(b \land a \not\equiv 0\), then \((s', u') \in R'\).
  \end{itemize}
  In the case of computing trace equivalence, the \(P\) will just be dead state detection \(\algIsDead\); and in the case of computing bisimulation equivalence, \(P\) will be the constant false predicate.
\end{definition}

We first prove a correspondence theorem, to show that the symbolic progression exactly correspond to the non-symbolic progression (\Cref{def:alg-progression}) on its concretization.

\begin{theorem}[Correspondence]
  Take any two relation over the state sets of two symbolic \gkat automata \(R, R'\), the following correspondence holds: \(R \transCorrSymb{P} R'\) if and only if \(R \transCorr{P} R'\) in their concretizations.
\end{theorem}

\begin{proof}
  By unfolding the definition of concretization, we obtain the following implications:
  \begin{align*}
    s \transOut{b} \retc & \text{ implies } s \transOut{\alpha} \retc \text{ for all \(\alpha \leq b\), in the concrete automata}; \\
    s \transvia{b \mid p} s' & \text{ implies } s \transvia{\alpha \mid p} s' \text{ for all \(\alpha \leq b\), in the concrete automata}.
  \end{align*}
  First, assuming \(R \transCorrSymb{P} R'\), we prove some conditions in \(R \transCorr{P} R'\) as examples: let \((s, u) \in R\),
  \begin{enumerate}
    \item if \(s \transOut{\alpha} \retc\) in the concrete automata, then there exists \(b\) s.t. \(s \transOut{b} \retc\) and \(\alpha \leq b\). 
    Because \(\bigvee \varepsilon(s) \equiv \bigvee \varepsilon(u)\), therefore there exists \(\alpha \leq a\) and \(u \transOut{a} \retc\), which then implies \(u \transOut{\alpha} \retc\) in the concrete automaton with states \(U\).
    Thus, we have proven \(s \transOut{\alpha} \retc\) implies \(u \transOut{\alpha} \retc\), and the reverse direction can be obtained by symmetry.
    \item Assume \(s \transvia{\alpha \mid p} s'\) and \(u \transRej{\alpha}\) in the automata. 
    By definition of concretization, there exists \(b\) s.t. \(\alpha \leq b\) and \(s \transvia{b \mid p} s'\). 
    Additionally, \(\alpha \in \rho(u)\). 
    This means \(\transvia{b \mid p} s'\) and \(\rho(u) \land b \not\equiv 0\), because \(R \transCorrSymb{P} R'\), we obtained \(P(b)\).
  \end{enumerate}
  Then assume \(R \transCorr{P} R'\), we will prove some condition in \(R \transCorrSymb{P} R'\) as examples: let \((s, u) \in R\)
  \begin{enumerate}
    \item take any \(\alpha \leq \bigvee \varepsilon(s)\), then there exists a \(b \in \varepsilon(s)\), s.t. \(\alpha \leq b\), which implies \(s \transOut{\alpha} \retc\) in the concretization.
    Because \(s \transOut{\alpha} \retc\) iff \(u \transOut{\alpha} \retc\), there exists \(u \transOut{a} \retc\) where \(\alpha \leq a\), implying \(\alpha \leq \bigvee \varepsilon(u)\).
    Because atoms are truth assignments, and we just proved \(\alpha \leq \bigvee \varepsilon(s)\) implies \(\alpha \leq \bigvee \varepsilon(u)\), this means \(\bigvee \varepsilon(s) \leq \bigvee \varepsilon(u)\) in Boolean algebra; and the reverse direction \(\bigvee \varepsilon(s) \geq \bigvee \varepsilon(u)\) is shown by symmetric argument.
    \item Assume \(s \transvia{b \mid p} s'\) and \(b \land \rho(u) \not\equiv 0\), then there exists an atom \(\alpha\) s.t. \(\alpha \leq b\) and \(\alpha \leq \rho(u)\).
    By definition of concretization, \(s \transvia{\alpha \mid p} s'\) and \(u \transRej{\alpha}\).
    Because \(R \transCorr{P} R'\), therefore we reached \(P(s')\).
  \end{enumerate}
  All the other cases follows exactly the same pattern as the listed case, therefore we opt to omit them in this paper.
\end{proof}

\begin{proof}[Proof of~\Cref{thm:symb-corr}]
  The proof essentially by reducing to~\Cref{thm:fin-trace-corr,thm:inf-trace-corr}, where we show that we can construct similar loop invariants with \(\transCorrSymb{}\). 
  Then by the correspondence theorem, these loop invariants can be unfolded into the loop invariants in~\Cref{thm:fin-trace-corr,thm:inf-trace-corr} in the concrete automata, essentially reduced the correctness of the symbolic algorithm to the correctness of the non-symbolic ones.

  We take the hardest case where we show that the symbolic \(\algEquiv(s, u) = \true\) then \(\sem{s} = \sem{u}\), where recall that the trace semantics of a symbolic automata coincide with its concretization.  
  We assert the following induction invariant when \(\algEquiv(s, u) = \true\):
  \[\algIsDead(s) \text{ and }\algIsDead(u)
  \quad\text{or}\quad
  (s, u) \in e(R) \text{ and } (R \setminus R_0) \transCorrSymb{\algIsDead} e(R),
  \]
  where \(R\) is all the previous inputs of \(\algUnion\), \(R_0\) is the starting value of \(R\) for \(\algEquiv\), and \(e(R)\) is the smallest equivalence relation that contains \(R\).
  Verifying the soundness of the induction invariant of the symbolic algorithm is more straightforward as there is no loop, but the argument is eactly the same as \Cref{thm:fin-trace-corr}.
  Thus, when starting with an empty union-find object, \(R_0\) will be empty, thus we case analysis on the post condition:
  \begin{itemize}
    \item if both \(s\) and \(u\) are dead, then by~\Cref{thm:live-iff-have-trace}, both \(\sem{s}\) and \(\sem{u}\) are empty, hence equal; 
    \item if \((R) \transCorrSymb{\algIsDead} e(R)\), then by correspondence of the symbolic and concrete automata, we have \((R) \transCorr{\algIsDead} e(R)\), which implies \(e(R)\) is an invariant. 
    Because \((s, u) \in e(R)\), therefore \(s\) and \(u\) are trace equivalent.
  \end{itemize}
\end{proof}

\section{Application: Decision Procedure for \gkat}\label[apx]{ap:aut-construction-gkat}

% \begin{itemize}
%   \item Syntax of \gkat with background (model programs)
%   \item Symbolic derivatives 
%   \item Computing the symbolic derivatives 
%   \item (Maybe?) Thompson's construction and homomorphism
% \end{itemize}

In~\Cref{sec:overview,sec:correctness-and-complexity}, we presented decision procedures for \gkat automata and symbolic \gkat automata. We now want to apply these algorithms to expressions directly and therefore need algorithms to compile expressions to automata. 
Although the procedure to generate \gkat automata from \gkat and \cfgkat expressions has been presented in previous works~\cite{smolka_GuardedKleeneAlgebra_2020,schmid_GuardedKleeneAlgebra_2021,zhang_CFGKATEfficientValidation_2025}, we need an algorithm to generate \emph{symbolic \gkat automata} in order to take advantage of the on-the-fly symbolic algorithm in~\Cref{alg:symb-bisim}.

In this section, we will focus on the foundational system of \gkat,  and explain how to generate symbolic \gkat automata from \gkat expressions. 
In the next section, we will then define a separate compilation procedure for \cfgkat,   whose operational semantics is also presented as \gkat automata~\cite{zhang_CFGKATEfficientValidation_2025} but with additional structures, like non-local control flows, that make the compilation procedure more complex and requiring extra steps that are non-trival. 
These two compilation procedures together will bridge the gap to reusing of the symbolic equivalence checking algorithm in~\Cref{alg:symb-bisim} for both systems.

The syntax of \gkat largely mimics that of the imperative programs, except it abstracts away the meaning of the primitive actions and the tests of the program to enable a rich variety of semantics~\cite{smolka_GuardedKleeneAlgebra_2020,schmid_GuardedKleeneAlgebra_2021}.
Formally, \gkat expressions over a finite set of \emph{primitive tests} \(T\) and a finite set of \emph{primitive actions} \(\Sigma\) is defined as follows:
\begin{align*}
  \BExp \ni b, c & \triangleq 
  0 \mid 1 \mid b \land c \mid b \lor c \mid \overline{b} &
  \gkat \ni e, f & \triangleq 
  b \mid p \in \Sigma \mid e \seq f \mid e +_b f \mid e^{(b)}
\end{align*}
\(e +_b f\) and \(e^{(b)}\) are the compact notations for \(\comIfElse{b}{e}{f}\) and \(\comWhile{b}{e}\).

The operational semantics of a \gkat expression is defined as a \gkat automaton, and the trace equivalence of the generated \gkat automata corresponds to the trace equivalence of the original \gkat expression~\cite{smolka_GuardedKleeneAlgebra_2020,schmid_GuardedKleeneAlgebra_2021}.
This result enables us to determine trace equivalence of GKAT expressions by combining the efficient algorithm presented in~\Cref{alg:nonsymb-bisim} with existing automata generation processes like Thompson's construction and derivatives.

\begin{figure}
  \begin{mathpar}
    \inferrule[]{\\}
    {p \transvia{1 \mid p} 1} \and  
    \inferrule[]{\\}
    {b \transOut{b} \retc} \and  
    \inferrule[]
    {e \transRes{a} r}
    {e +_b f \transRes{b \land a} r}
    \and
    \inferrule[]
    {f \transRes{a} r}
    {e +_b f \transRes{b \land a} r}
    \and 
    \inferrule[]
    {e \transvia{b \mid p} e'}
    {e \seq f \transvia{b \mid p} e' \seq f}
    \and
    \inferrule[]
    {e \transOut{b} \retc \\ f \transRes{a} r}
    {e \seq f \transRes{b \land a} r}
    \and  
    \inferrule[]
    {\\}
    {e^{(b)} \transOut{\overline{b}} \retc}  
    \and  
    \inferrule[]
    {e \transvia{a \mid p} e'}
    {e^{(b)} \transvia{b \land a \mid p} e'\seq e^{(b)}}
  \end{mathpar}
  \caption{Symbolic Derivatives for GKAT}\label{fig:derivatives-gkat}
\end{figure}

However, to effectively utilize the on-the-fly symbolic decision procedure in~\Cref{alg:symb-bisim}, we will need to devise an algorithm to generate symbolic \gkat automata \emph{directly} from \gkat expressions.
This process is given by \emph{symbolic derivatives}, which generates a transition system on-the-fly, with states being \gkat expressions (denoted as \(\gkat\)):
\begin{align*}
  \varepsilon &: \gkat \to \powSet{\BExp}, & 
  \delta &: \gkat \to \powSet{\BExp \times \gkat \times \Sigma}.
\end{align*}
Indeed, this transition system defines the operational semantics for \gkat expressions.

Concretely, symbolic derivatives for \gkat are given in~\Cref{fig:derivatives-gkat}, where we use the notation \(e \transRes{b} r\) to say that \(e\) under condition \(b\) will reach result \(r\), where \(r\) ranges over \(\{\retc\} + \Sigma \times \gkat\).
Intuitively, a rule with \(e \transRes{b} r\) can be split into a return rule \(e \transOut{b} \retc\) and a transition rule \(e \transvia{b \mid p} e'\).
As an example, the following rule for sequencing can be soundly split as follows:
\[
    \inferrule
    {e \transOut{b} \retc \\ f \transRes{a} r}
    {e \seq f \transRes{b \land a} r}
    \quad 
    \longrightsquigarrow 
    \quad
    \inferrule
    {e \transOut{b} \retc \\ f \transOut{a} \retc}
    {e \seq f \transOut{b \land a} \retc} 
    \quad 
    \inferrule
    {e \transOut{b} \retc \\ f \transvia{a \mid p} f'}
    {e \seq f \transvia{b \land a \mid p} f'}
\]

With the operational semantics defined by derivatives, we can construct a symbolic \gkat automaton \(\langle e \rangle\) for any expression \(e\), where the states are the \emph{reachable} expressions by taking the derivatives of \(e\), and the start state is the 
\begin{wrapfigure}{r}{3cm}
  \vspace{-0.9cm}
  \begin{minipage}{3cm}
  \begin{flushright}
    \begin{tikzpicture}
      \node (init-if) {};
      \node[textstate] (if) [right=3mm of init-if] {\(1 +_b p\)};
      \node[textstate] (skip) [below=5mm of if] {\(1\)};
      \node (if-ret) [right=3mm of if] {\(\retc\)};
      \node (skip-ret) [right=3mm of skip] {\(\retc\)};
      \draw[->] (init-if) edge (if);
      \draw[->] (if) edge[output-edge] node[below] (ifOut) {\(b\)} (if-ret);
      \draw[->] (skip) edge[output-edge] node[above] (skipOut)  {\(1\)} (skip-ret);
      \draw[->] (if) edge node[left] {\(\overline{b} \mid p\)} (skip);

      \node (init-loop) [below=3cm of init-if] {};
      \node[textstate] (loop) [right=3mm of init-loop] {\((1 +_b p)^{(c)}\)};
      \node (loop-ret) [below=3mm of loop] {\(\retc\)};
      \draw[->] (init-loop) edge (loop);
      \draw[->] (loop) edge[output-edge] node[right] {\(\overline{c}\)} (loop-ret);
      \draw[->] (loop) edge[out=125,in=55,loop,looseness=2] node[above] (loop-trans) {\(c \land \overline{b} \mid p\)} (loop);

      % Boxes
      \node (if-aut) [
        fit=(init-if) (if) (skip) (if-ret) (skip-ret) (ifOut) (skipOut), 
        draw=gray, dashed, rounded corners=5pt, thick, draw opacity=0.5,
        label={[anchor=south west]north west:\(\langle 1 +_b p \rangle:\)}] {};

      \node (loop-aut) [
        fit=(init-loop) (loop) (loop-ret) (loop-trans), 
        draw=gray, dashed, rounded corners=5pt, thick, draw opacity=0.5,
        label={[anchor=south west]north west:\(\langle (1 +_b p)^{(c)} \rangle:\)}] {};
  \end{tikzpicture}
  \end{flushright}
\end{minipage}
\end{wrapfigure}
expression \(e\).
Then by running the algorithm in~\Cref{alg:symb-bisim} on symbolic automata \(\langle e \rangle\) and \(\langle f \rangle\), we can decide the trace equivalence of the expressions \(e\) and \(f\).

As an example, We compute the derivative for expression \((1 +_b p)^{(c)}\). 
There is only one rule that leads to the (immediate) termination of while loops: when the loop condition \(c\) is unsatisfied: \(e^{(c)} \transOut{\overline{c}} \retc\).
For the transition case, the loop \(e^{(c)}\) only transition to \(e' \seq e^{(c)}\) when the body \(e\) transition to \(e'\). 
Since only \(p\) branch of \(1 +_b p\) will transition: \(p \transvia{1 \mid p} 1\); we obtain \(1 +_b p \transvia{\overline{b} \mid p} 1\) and \((1 +_b p)^{(c)} \transvia{c \land \overline{b} \land 1 \mid p} 1 \seq e^{(c)}\).
Finally, because \(c \land b \land 1 \equiv c \land b\) and \(1 \seq e^{(c)} \equiv e^{(c)}\), we optimize the derivative of \((1 +_b p)^{(c)}\):
\begin{align*}
  \varepsilon((1 +_b p)^{(c)}) & = \{\overline{c}\}\ ,&  
  \delta((1 +_b p)^{c}) & = \{(c \land \overline{b}, (1 +_b p)^{(c)}, p)\}\ .
\end{align*}
Notice the atoms satisfying \(c \land b\) will neither terminate nor transition, which means they silently diverge, or formally, are rejected by the state \((1 +_b p)^{(c)}\).

\section{Omitted proofs in~\Cref{sec:aut-construction-cfgkat}}\label[apx]{ap:omitted}

\begin{figure}[h]
\begin{mathpar}
  \inferrule*
  {
    \inferrule*[Left=\(e^{(x \neq 2)} \equiv \comSkip \unfold e^{(x \neq 2)}\)]
    {
      \inferrule*
      {(x \mapsto 3, e) \transvia{\overline{b} \land a \mid p} (x \mapsto 3, \comSkip)}
      {(x \mapsto 3, e \unfold e^{(x \neq 2)}) \transvia{\overline{b} \land a \mid p} (x \mapsto 3, \comSkip \unfold e^{(x \neq 2)})}
    }
    {(x \mapsto 3, e \unfold e^{(x \neq 2)}) \transvia{\overline{b} \land a \mid p} (x \mapsto 3, e^{(x \neq 2)})}
  }
  {(x \mapsto 3, e^{(x \neq 2)}) \transvia{\overline{b} \land a \mid p} (x \mapsto 3, e^{(x \neq 2)})}
  \and
  \inferrule*
  {
    \inferrule*
    {
      (x \mapsto 3, e) \transOut{b} \acc{(x \mapsto 4)} \\
      \inferrule*
      {
        \inferrule*[Left=\(e^{(x \neq 2)} \equiv \comSkip \unfold e^{(x \neq 2)}\)]
        {
          \inferrule*
          {(x \mapsto 4, e) \transvia{a \mid p} (x \mapsto 4, \comSkip)}
          {(x \mapsto 4, e \unfold e^{(x \neq 2)}) \transvia{a \mid p} (x \mapsto 4, \comSkip \unfold e^{(x \neq 2)})}
        }
        {(x \mapsto 4, e \unfold e^{(x \neq 2)}) \transvia{a \mid p} (x \mapsto 4, e^{(x \neq 2)})}
      }
      {(x \mapsto 4, e^{(x \neq 2)}) \transvia{a \mid p} (x \mapsto 4, e^{(x \neq 2)})}
    }
    {(x \mapsto 3, e \unfold e^{(x \neq 2)}) \transvia{b \land a \mid p} (x \mapsto 4, e^{(x \neq 2)})}}
  {(x \mapsto 3, e^{(x \neq 2)}) \transvia{b \land a \mid p} (x \mapsto 4, e^{(x \neq 2)})}
\end{mathpar}
\caption{The derivation of the outgoing transitions for \((x \mapsto 3, e^{x \neq 2})\), where \(e^{x \neq 2}\) is from~\Cref{fig:loop-cfgkat-program-cumu}. 
We also include the standard optimization \(e^{x \neq 2} \equiv \comSkip \unfold e^{x \neq 2}\), which removes \(\comSkip\), in the derivation.}\label{fig:deriv-out-going-transition-cumu-loop-example}
\end{figure}

\begin{proof}[Proof for~\Cref{thm:cumulation-computes}]
  Before we start the formal proof, we would like to note that \(\accu_{\res, \con}\) can also be characterized as the least fixpoint of the following equation:
  \[\accu_{\res, \con}(n) \triangleq \res(n) \cup \bigcup \{h(r) \mid (h, n') \in \con(n), r \in \accu_{\res, \con}(n')\}.\]
  Hence, the Scott-continuity of the function in the above fixpoint equation will imply \(\accu\) can be approximated from below by the Kleene's fixpoint theorem.
  However, we opt to present our current proof, since it is closer to our implementation.

  We design an algorithm to explicitly compute the result of \(\accu_{\res, \con}(n)\) by keeping track of an axillary variable \(M\), which keeps track of the explored \(n \in N\):
  \begin{align*}
    \accu'_{\res,\con}(n, M) & \triangleq \begin{cases}
      \emptyset & \text{if } n  \in M \\
      \begin{aligned}
        \res(n) \cup 
        \{h(r) \mid 
          {}& (n', h)  \in \con(n) \text{ and }\\
          {}& r  \in \accu'_{\res,\con}(n', M \cup \{n\})\}
      \end{aligned}
      & \text{otherwise} 
    \end{cases} \\
    \accu_{\res,\con}(n) & = \accu'_{\res,\con}(n,  \emptyset ).
  \end{align*}
  The algorithm will terminate because the size of \(N \setminus M\) decrease for each recursive call.

  If \(r  \in \res(n)\), then we say that \(r  \in \accu'_{\res,\con}(n, M)\) is computable in 1 call of \(\accu'\); and if \(r  \in \accu'_{\res,\con}(n', M \cup \{n\})\) is computable in \(k\) calls of \(\accu'\), we say \(h(r)  \in \accu'_{\res,\con}(n', M)\) is computable in \(k + 1\) calls of \(\accu'\).
  Since every call decreases the size of \(N \setminus M\), then \(r  \in \accu'_{\res,\con}(n',  \emptyset )\) if and only if it can be computed within \(|N \setminus  \emptyset | = |N|\) calls.

  Then we can prove that if \(r  \in \accu_{\res,\con}(n)\) is derivable with proof of size \(k\) if and only if \(r\) can be computed with exactly \(k\) call of \(\accu'_{\res,\con}\).
  \begin{itemize}
    \item If \(k = 1\), then 
    \begin{align*}
      & r  \in \accu_{\res,\con}(n) \text{ is derivable in 1 step} \\
      & \iff r  \in \res(n) \\
      & \iff r  \in \accu'_{\res,\con}(n,  \emptyset ) \text{ is computable in 1 call}.
    \end{align*}
    \item If \(r  \in \accu_{\res,\con}(n)\) can be derived using a derivation of size \(k + 1\).
    Then let \((n', h) = \con(n)\) and \(r' \in \accu(n')\) is derivable in \(k\) step, s.t. \(r = h(r')\).
    \begin{align*}
      & r \in \accu_{\res,\con}(n) \text{ derivable in \(k + 1\) step}\\
      & \iff r = h(r') \text{ where } (n', h) = \con(n) 
        \text{ and } 
        r' \in \accu(n') \text{is derivable in \(k\) steps}\\
      & \iff r = h(r') \text{ where } (n', h) = \con(n) 
        \text{ and } 
        r' \in \accu(n') \text{is computable in \(k\) calls}\\
      & \iff r \in \accu'_{\res,\con}(n,  \emptyset ) \text{ in \(k + 1\) call}.%\qedhere
    \end{align*}
  \end{itemize} 
\end{proof}

\begin{proof}[Proof for~\Cref{thm:loop-comp-cfgkat-correctness}] 
  The 2 equivalences are proven separately by induction.

  First to prove $\delta(\pi, e^{(b)}) = \accu(\res_{\delta}, \con_{\delta}, \pi)$, 
  we begin in the forward direction, i.e., we show that for all \(b, p\),
  \(\sigma\) and \(f\) such that there is 
  \((\pi, e^{(b)}) \transvia{b \mid p} (\sigma, f)\), there is 
  \((b, (\sigma, f), p) \in \accu(\res_{\delta}, \con_{\delta}, \pi)\). 
  The proof proceeds by induction on the derivation of 
  \((\pi, e^{(b)}) \transvia{b \mid p} (\sigma, f)\). 
  It is clear by the rules of \Cref{fig:sem-CF-GKAT} that there are only two cases to consider:
  \vspace{-0.5em}
  \begin{center}
    \scalebox{0.85}{\(
      \inferrule
      {
        \inferrule
        {(\pi, e) \transvia{a \mid p} (\pi', e')}
        {(\pi, e \unfold e^{(b)}) \transvia{a \mid p} (\pi', e' \unfold e^{(b)})}
      }
      {(\pi, e^{(b)}) \transvia{b[\pi] \land a \mid p} (\pi', e' \unfold e^{(b)})}
      \quad (1)
    \)}
    \qquad
    \scalebox{0.85}{\(
      \inferrule
      {
        \inferrule
        {(\pi, e) \transOut{a}{c} \\ c \in \{\acc{\pi'}, \contc{\pi'}\} \\ (\pi', e^{(b)}) \transvia{d \mid p} (\sigma, f)}
        {(\pi, e \unfold e^{(b)}) \transvia{a \land d \mid p} (\sigma, f)}
      }
      {(\pi, e^{(b)}) \transvia{b[\pi] \land a \land d \mid p} (\sigma, f)}
      \quad (2)
    \)}
  \end{center}
  In case (1), from the definition of \(\res_{\delta}\), we immediately have:
  \begin{small}
  \begin{mathpar}
    \inferrule
    {(\pi, e) \transvia{a \mid p} (\pi', e')}
    { (b[\pi] \land a, (\pi', e' \unfold e^{(b)}), p) \in \res_{\delta}(\pi) }
  \end{mathpar}
  \end{small}
  which shows that \((b[\pi] \land a, (\pi', e' \unfold e^{(b)}), p) \in \accu(\res_{\delta}, \con_{\delta}, \pi)\).

  \noindent In case (2), from the premises \((\pi, e) \transOut{a}{c}\) and \(c \in \{\acc{\pi'}, \contc{\pi'}\}\) obtained from inversion, we have \((\pi', \langle b[\pi] \land a |) \in \con_{\delta}(\pi)\). Furthermore, from the induction hypothesis, we know that \((d, (\sigma, f), p) \in \accu(\res_{\delta}, \con_{\delta}, \pi')\). This means that the following holds:
  \begin{small}
  \begin{mathpar}
    \inferrule 
    { (\pi', \langle b[\pi] \land a |) \in \con_{\delta}(\pi) \\ 
      (d, (\sigma, f), p) \in \accu(\res_{\delta}, \con_{\delta}, \pi') }
    { (b[\pi] \land a \land d, (\sigma, f), p) \in \accu(\res_{\delta}, \con_{\delta}, \pi) }
  \end{mathpar}
  \end{small}
  which concludes the proof of the forward direction.

  To prove the reverse direction, we need to show that for all 
  \((b_0, (\sigma, f), p) \in \accu(\res_{\delta}, \con_{\delta}, \pi)\), there is
  \((\pi, e^{(b)}) \transvia{b_0 \mid p} (\sigma, f)\). By induction on 
  the derivation of \((b_0, (\sigma, f), p) \in \accu(\res_{\delta}, \con_{\delta}, \pi)\), 
  we have the following cases to consider:
  \begin{mathpar}
    \inferrule
    {(b, (\sigma, f), p) \in \res_{\delta}(\pi)}
    {(b, (\sigma, f), p) \in \accu(\res_{\delta}, \con_{\delta}, \pi)}
    (1)
    \and 
    \inferrule
    { (\pi', h) \in \con_{\delta}(\pi) \\ 
      r \in \accu(\res_{\delta}, \con_{\delta}, \pi') }
    {h(r) \in \accu(\res_{\delta}, \con_{\delta}, \pi)}
    (2)
  \end{mathpar}
  In case (1), by inversion on \((b_0, (\sigma, f), p) \in \res_{\delta}(\pi)\) we have:
  \begin{mathpar}
    \inferrule 
    { (\pi, e) \transvia{a \mid p} (\pi', e') }
    { (b[\pi] \land a, (\pi', e' \unfold e^{(b)}), p) \in \res_{\delta}(\pi) }
  \end{mathpar}
  and \(b_0 = b[\pi] \land a\), \(\sigma = \pi'\), \(f = e' \unfold e^{(b)}\). 
  Now from the rules of \Cref{fig:sem-CF-GKAT}, we can derive the transition
  \((\pi, e^{(b)}) \transvia{b[\pi] \land a \mid p} (\pi', e' \unfold e^{(b)})\) which is just
  \((\pi, e^{(b)}) \transvia{b_0 \mid p} (\sigma, f)\).

  \noindent In case (2), we have \(h(r) = (b_0, (\sigma, f), p)\) by assumption. 
  Inverting \((\pi', h) \in \con_{\delta}(\pi)\) we have:
  \begin{mathpar}
    \inferrule    
    { (\pi, e) \transOut{a}{c} \\ 
      c \in \{ \acc{\pi'}, \contc{\pi'} \} }
    { (\pi', \langle b[\pi] \land a| \in \con_{\delta}(\pi) ) }
  \end{mathpar}
  and \(h = \langle b[\pi] \land a|\). 
  Let \(r = (d, (\sigma', f'), p') \in \accu(\res_{\delta}, \con_{\delta}, \pi')\).
  By \(h(r) = (b_0, (\sigma, f), p)\) we have 
  \(b_0 = b[\pi] \land a \land d\), \(\sigma = \sigma'\), \(f = f'\), and \(p = p'\).
  From the induction hypothesis, we know that there is transition 
  \((\pi', e^{(b)}) \transvia{d \mid p} (\sigma, f)\). The rules of \Cref{fig:sem-CF-GKAT} then allow us to derive the transition 
  \((\pi, e^{(b)}) \transvia{b[\pi] \land a \land b \mid p} (\sigma, f)\) which concludes the
  proof of the reverse direction.

  We now begin proving the second equivalence 
  \(\varepsilon(\pi, e^{(b)}) = \accu(\res_{\varepsilon},\con_{\varepsilon},\pi)\). 
  In the forward direction, we show that for all \((\pi, e^{(b)}) \transOut{b_0} c_0\),
  we have \((b_0, c_0) \in \accu(\res_{\varepsilon},\con_{\varepsilon},\pi)\). 
  By induction on the derivation of \((\pi, e^{(b)}) \transOut{b_0} c_0\), the only cases to 
  consider are:
  \begin{small}
  \begin{mathpar}
  \inferrule
  {\\}
  {(\pi, e^{(b)}) \transOut{ \neg b[\pi]}{\acc{\pi}}}
 \and 
  \inferrule
  {(\pi, e) \transOut{a}{\brkc{ \pi'}}}
  {(\pi, e^{(b)}) \transOut{b[\pi] \land a}{\acc{ \pi'}}}
  \and
  \inferrule
  {(\pi, e) \transOut{a}{c} \\ c = \retc \text{ or } c = \jmpc{(l,  \pi')}}
  {(\pi, e^{(b)}) \transOut{b[\pi] \land a}{c}}
  \and
  \inferrule
  {(\pi, e) \transOut{a}{c} \\ c = \acc{ \pi'} \text{ or } \contc{ \pi'} \\ ( \pi', e^{(b)}) \transOut{d}{c_0}}
  {(\pi, e^{(b)}) \transOut{b[\pi] \land a \land d}{c_0}}
  \end{mathpar}
  \end{small}
  The first 3 cases are trivial as they all belong in \(\res_{\varepsilon}(\pi)\), and by 
  extension \(\accu(\res_{\varepsilon}, \con_{\varepsilon}, \pi)\), via the following rules:
  \vspace{-1em}
  \begin{small}
  \begin{mathpar}
  \inferrule
  {\\}
  {(\neg b[\pi], \acc{\pi}) \in \res_{\varepsilon}(\pi)} 
  \and 
  \inferrule
  {(\pi, e)\transOut{a}{\brkc{\pi'}}}
  {(b[\pi] \land a, \acc{\pi'}) \in \res_{\varepsilon}(\pi)}
  \and 
  \inferrule
  {(\pi, e)\transOut{a}{c} \\ c = \retc \text{ or } c = \jmpc{(l, \pi')}}
  {(b[\pi] \land a, c) \in \res_{\varepsilon}(\pi)}
  \end{mathpar}
  \end{small}
  From the first 2 premises of the 4th case, we can derive
  \((\pi', \langle b[\pi] \land a|) \in \con_{\varepsilon}(\pi)\) by definition of
  \(\con_{\varepsilon}\). 
  Applying the induction hypothesis on \((\pi', e^{(b)}) \transOut{d} c_0\) gives us
  \((d, c_0) \in \accu(\res_{\varepsilon}, \con_{\varepsilon}, \pi')\). The \(\accu\) rule
  then allows us to derive:
  \begin{small}
  \begin{mathpar}
  \inferrule
  {(\pi', \langle b[\pi] \land a|) \in \con_{\varepsilon}(\pi) \\ 
   (d, c_0) \in \accu(\res_{\varepsilon}, \con_{\varepsilon}, \pi')}
  {(b[\pi] \land a \land d, c_0) \in \accu(\res_{\varepsilon}, \con_{\varepsilon}, \pi)}
  \end{mathpar}
  \end{small}
  which concludes the proof of the forward direction.

  To prove the reverse direction, we need to show that for all 
  \((b_0, c_0) \in \accu(\res_{\varepsilon}, \con_{\varepsilon}, \pi)\), there is 
  \((\pi, e^{(b)}) \transOut{b_0} c_0\). By induction on the derivation of 
  \((b_0, c_0) \in \accu(\res_{\varepsilon}, \con_{\varepsilon}, \pi)\), 
  we have the following cases to consider:
  \begin{small}
  \begin{mathpar}
    \inferrule
    {(b_0, c_0) \in \res_{\varepsilon}(\pi)}  
    {(b_0, c_0) \in \accu(\res_{\varepsilon}, \con_{\varepsilon}, \pi)}
    (1)
    \and
    \inferrule
    {(\pi', h) \in \con_{\varepsilon}(\pi) \\ 
     r \in \accu(\res_{\varepsilon}, \con_{\varepsilon}, \pi')}
    {h(r) \in \accu(\res_{\varepsilon}, \ con_{\varepsilon}, \pi)}
    (2)
  \end{mathpar}
  \end{small}
  Similarly to the forward direction, case (1) is trivial as inversion on the
  proof of \((b_0, c_0) \in \res_{\varepsilon}(\pi)\) yields the following cases:
  \begin{small}
  \begin{mathpar}
  \inferrule
  {\\}
  {(\neg b[\pi], \acc{\pi}) \in \res_{\varepsilon}(\pi)} 
  \and 
  \inferrule
  {(\pi, e)\transOut{a}{\brkc{\pi'}}}
  {(b[\pi] \land a, \acc{\pi'}) \in \res_{\varepsilon}(\pi)}
  \and 
  \inferrule
  {(\pi, e)\transOut{a}{c_0} \\ c_0 = \retc \text{ or } c_0 = \jmpc{(l, \pi')}}
  {(b[\pi] \land a, c_0) \in \res_{\varepsilon}(\pi)}
  \end{mathpar}
  \end{small}
  which allow us to derive \((\pi, e^{(b)}) \transOut{\neg b[\pi]} \acc{\pi}\),
  \((\pi, e^{(b)}) \transOut{b[\pi] \land a} \acc{\pi'}\) and 
  \((\pi, e^{(b)}) \transOut{b[\pi] \land a} c_0\) respectively.

  \noindent In case (2), we have \(h(r) = (b_0, c_0)\) by assumption.
  Inverting \((\pi', h) \in \con_{\varepsilon}(\pi)\) we have:
  \begin{small}
  \begin{mathpar}
  \inferrule
  {(\pi, e) \transOut{a}{c} \\ c = \acc{\pi'} \text{ or } \contc{\pi'} }
  {(\pi', \langle b[\pi] \land a|) \in \con_{\varepsilon}(\pi)}
  \end{mathpar}
  \end{small}
  and \(h = \langle b[\pi] \land a|\). 
  Let \(r = (d, c') \in \accu(\res_{\varepsilon}, \con_{\varepsilon}, \pi')\).
  From the induction hypothesis, we have \((\pi', e^{(b)}) \transOut{d} c'\) and:
  \begin{small}
  \begin{mathpar}
    \inferrule
    { (\pi, e) \transOut{a}{c} \\ 
      c = \acc{\pi'} \text{ or } \contc{\pi'} \\ 
      (\pi', e^{(b)}) \transOut{d}{c'} }
    { (\pi, e^{(b)}) \transOut{b[\pi] \land a \land d}{c'} }
  \end{mathpar}
  \end{small}
  Because \((b_0, c_0) = h(r) = (b[\pi] \land a \land d, c')\), there is
  \((\pi, e^{(b)}) \transOut{b_0} c_0\) which concludes the proof.
\end{proof}

\myparagraph{Expressing the \(\varepsilon\) of while-loops}
Similar to the \(\delta\) case, we unfold the proof rules for \(e \unfold e^{(b)}\) and extract \emph{all} the premise that can lead to \(e^{(b)} \transOut{a} c\) for some condition \(a \in \BExp\) and continuation \(c \in C\); we omit the middle step \((\pi, e \unfold e^{(b)}) \transOut{a} c\):
\begin{mathpar}
  \inferrule
  {\\}
  {(\pi, e^{(b)}) \transOut{ \neg b[\pi]}{\acc{\pi}}}
 \and 
  \inferrule
  {(\pi, e) \transOut{a}{\brkc{ \pi'}}}
  {(\pi, e^{(b)}) \transOut{b[\pi] \land a}{\acc{ \pi'}}}
  \and
  \inferrule
  {(\pi, e) \transOut{a}{c} \\ c = \retc \text{ or } c = \jmpc{(l,  \pi')}}
  {(\pi, e^{(b)}) \transOut{b[\pi] \land a}{c}}
  \and
  \inferrule
  {(\pi, e) \transOut{a}{c} \\ c = \acc{ \pi'} \text{ or } \contc{ \pi'} \\ ( \pi', e^{(b)}) \transOut{d}{c'}}
  {(\pi, e^{(b)}) \transOut{b[\pi] \land a \land d}{c'}}
\end{mathpar}
We notice that the first three rules compute elements in \(\varepsilon(\pi, e^{(b)})\) directly from \(\varepsilon(\pi, e)\); the final rule will compute the elements in \(\varepsilon(\pi, e^{(b)})\) inductively from \(\varepsilon(\pi', e^{(b)})\), and the resulting condition \(d\) is guarded by \(b[\pi] \land a\).
Thus, the first three rules are converted into rules for \(\res_{\varepsilon}\) and the final rule is converted into a rule for \(\con_{\varepsilon}\):
\begin{mathpar}
  \inferrule
  {\\}
  {(\neg b[\pi], \acc{\pi}) \in \res_{\varepsilon}(\pi)} 
 \and 
  \inferrule
  {(\pi, e)\transOut{a}{\brkc{\pi'}}}
  {(b[\pi] \land a, \acc{\pi'}) \in \res_{\varepsilon}(\pi)}
 \and 
  \inferrule
  {(\pi, e)\transOut{a}{c} \\ c = \retc \text{ or } c = \jmpc{(l, \pi')}}
  {(b[\pi] \land a, c) \in \res_{\varepsilon}(\pi)}
 \and 
  \inferrule
  {(\pi, e) \transOut{a} c \\ c = \acc{\pi'} \text{ or } c = \contc{\pi'}}
  {(\pi', \langle b[\pi] \land a|) \in \con_{\varepsilon}(\pi)}
\end{mathpar}

\begin{proof}[Proof for~\Cref{thm:cfgkat-jmp-compute}]
  We define \(\res_{\varepsilon}\), \(\con_{\varepsilon}\), \(\res_{ \delta }\), \(\con_{ \delta }\), via the following inference rules: let \(l\) ranges over all the labels in \(e\), and \(\sigma, \sigma'\) ranges over all indicator states,
  \begin{mathpar}
    \inferrule
    {(\sigma, f) \transOut{a}{c} \\ c \neq \jmpc{l, \sigma'}}
    {(a, c) \in \res_{\varepsilon}(\sigma, f)}
    \and  
    \inferrule
    {(\sigma, f) \transOut{a}{\jmpc{l, \sigma'}}}
    {((\sigma', e_l),  \langle a|) \in \con_{\varepsilon}(\sigma, f)}
    \and  
    \inferrule
    {(\sigma, f) \transvia{a \mid p} (\sigma', f')}
    {(a, (\sigma', f'), p) \in \res_{ \delta }(\sigma, f)}
    \and  
    \inferrule
    {(\sigma, f) \transOut{a}{\jmpc{l, \sigma'}}}
    {((\sigma', e_l),  \langle a|) \in \con_{ \delta }(\sigma, f)}
  \end{mathpar}
  We then show respective \(\accu\) computes the accepting \( \varepsilon \) and transition function \( \delta \) in jump resolution:
  \begin{align*}
    \varepsilon(\pi, f) &= \accu(\res_\varepsilon, \con_\varepsilon, (\pi, f)) &
    \delta(\pi, f) &= \accu(\res_\delta, \con_\delta, (\pi, f))
  \end{align*}
  via induction on the size of the derivation.

  First to prove 
  \(\varepsilon(\sigma, f) = \accu(\res_{\varepsilon}, \con_{\varepsilon}, (\sigma, f))\), 
  in the forward direction, we must show that for all
  \((\sigma, f) \transOut{b_0}{c}\) in \(\langle \pi, e \rangle\JmpRes{}\), there is
  \((b_0, c) \in \accu(\res_{\varepsilon}, \con_{\varepsilon}, (\sigma, f))\). 
  By induction on the derivation of \((\sigma, f) \transOut{b_0}{c}\)
  in \(\langle \pi, e \rangle\JmpRes{}\), there are 2 cases to consider:
  \begin{small}
  \begin{mathpar}
    \inferrule
    { (\sigma, f) \transOut{a} c  \text{ in } \langle \pi, e \rangle\\ 
      c \neq \jmpc{l, \sigma'} }
    { (\sigma, f) \transOut{a} c \text{ in } \langle \pi, e \rangle\JmpRes{} }
    (1)

    \inferrule
    {
      (\sigma, f) \transOut{a}{\jmpc{l, \sigma'}} \text{ in } \langle \pi, e \rangle \\ 
      (\sigma', e_l) \transOut{b} r \text{ in } \langle \pi, e \rangle\JmpRes{}
    }
    {(\sigma, f) \transOut{a \land b} r \text{ in } \langle \pi, e \rangle\JmpRes{}}
    (2)
  \end{mathpar}
  \end{small}

  \noindent In case (1), by definition of \(\res_{\varepsilon}\), we have 
  \((a, c) \in \res_{\varepsilon}(\sigma, f)\) which immediately gives us
  \((a, c) \in \accu(\res_{\varepsilon}, \con_{\varepsilon}, (\sigma, f))\).

  \noindent In case (2), the first premise gives us 
  \(((\sigma', e_l), \langle a |) \in \con_{\varepsilon}(\sigma, f)\)
  by definition of \(\con_{\varepsilon}\). The induction hypothesis applied to the
  second premise gives us 
  \((b, r) \in \accu(\res_{\varepsilon}, \con_{\varepsilon}, (\sigma', e_l))\).
  The \(\accu\) rule then allows us to derive:
  \begin{small}
  \begin{mathpar}
    \inferrule
    { ((\sigma', e_l), \langle a |) \in \con_{\varepsilon}(\sigma, f) \\ 
      (b, r) \in \accu(\res_{\varepsilon}, \con_{\varepsilon}, (\sigma', e_l)) }
    { (a \land b, r) \in \accu(\res_{\varepsilon}, \con_{\varepsilon}, (\sigma, f)) }
  \end{mathpar}
  \end{small}
  which concludes the proof of the forward direction.

  To prove the reverse direction, we need to show that for all
  \((b_0, c) \in \accu(\res_{\varepsilon}, \con_{\varepsilon}, (\sigma, f))\), there is
  \((\sigma, f) \transOut{b_0}{c}\) in \(\langle \pi, e \rangle\JmpRes{}\). By
  induction on the derivation of 
  \((b_0, c) \in \accu(\res_{\varepsilon}, \con_{\varepsilon}, (\sigma, f))\),
  the cases to consider are:
  \begin{small}
  \begin{mathpar}
    \inferrule
    { (b_0, c) \in \res_{\varepsilon}(\sigma, f) }
    { (b_0, c) \in \accu(\res_{\varepsilon}, \con_{\varepsilon}, (\sigma, f)) }
    (1)
    \and
    \inferrule
    { (n, h) \in \con_{\varepsilon}(\sigma, f) \\ 
      r \in \accu(\res_{\varepsilon}, \con_{\varepsilon}, n) }
    { h(r) \in \accu(\res_{\varepsilon}, \con_{\varepsilon}, (\sigma, f)) }
    (2)
  \end{mathpar}
  \end{small}

  \noindent In case (1), inversion on \((b_0, c) \in \res_{\varepsilon}(\sigma, f)\) gives us:
  \begin{small}
  \begin{mathpar}
    \inferrule
    { (\sigma, f) \transOut{a}{c} \text{ in } \langle \pi, e \rangle \\ 
      c \neq \jmpc{l, \sigma'} }
    { (a, c) \in \res_{\varepsilon}(\sigma, f) }
  \end{mathpar}
  \end{small}
  which allows us to derive
  \((\sigma, f) \transOut{a}{c}\) in \(\langle \pi, e \rangle\JmpRes{}\).

  \noindent In case (2), inversion of the first premise gives us:
  \begin{small}
  \begin{mathpar}
    \inferrule  
    { (\sigma, f) \transOut{a}{\jmpc{l, \sigma'}} \text{ in } \langle \pi, e \rangle }
    { ((\sigma', e_l), \langle a |) \in \con_{\varepsilon}(\sigma, f) }
  \end{mathpar}
  \end{small}
  where \(n = (\sigma', e_l)\) and \(h = \langle a |\).
  Let \(r = (b, c') \in \accu(\res_{\varepsilon}, \con_{\varepsilon}, (\sigma', e_l))\).
  The induction hypothesis applied to the second premise gives us
  \((\sigma', e_l) \transOut{b}{c'}\) in \(\langle \pi, e \rangle\JmpRes{}\).
  We can now derive:
  \begin{small}
  \begin{mathpar}
    \inferrule 
    { (\sigma, f) \transOut{a}{\jmpc{l, \sigma'}} \text{ in } \langle \pi, e \rangle \\ 
      (\sigma', e_l) \transOut{b}{c'} \text{ in } \langle \pi, e \rangle\JmpRes{} }
    { (\sigma, f) \transOut{a \land b}{c'} \text{ in } \langle \pi, e \rangle\JmpRes{} }
  \end{mathpar}
  \end{small}
  which concludes the proof of the reverse direction.

  We now begin proving the second equivalence
  \(\delta(\sigma, f) = \accu(\res_{\delta}, \con_{\delta}, (\sigma, f))\).
  In the forward direction, we must show that for all
  \((\sigma, f) \transvia{b_0 \mid p} (\sigma', f')\) in 
  \(\langle \pi, e \rangle\JmpRes{}\), there is
  \((b_0, (\sigma', f'), p) \in \accu(\res_{\delta}, \con_{\delta}, (\sigma, f))\).
  We proceed by induction on the derivation of
  \((\sigma, f) \transvia{b_0 \mid p} (\sigma', f')\) in 
  \(\langle \pi, e \rangle\JmpRes{}\). 
  There are 2 cases to consider:
  \begin{small}
  \begin{mathpar}
    \inferrule
    { (\sigma, f) \transvia{b \mid p} (\sigma', f') \text{ in } \langle \pi, e \rangle }
    { (\sigma, f) \transvia{b \mid p} (\sigma', f') \text{ in } \langle \pi, e \rangle\JmpRes{} }
    (1)
    \and
    \inferrule
    {
      (\sigma, f) \transOut{a}{\jmpc{l, \sigma'}} \text{ in } \langle \pi, e \rangle \\
      (\sigma', e_l) \transvia{b \mid p} (\sigma'', f') \text{ in } \langle \pi, e \rangle\JmpRes{}
    }
    { (\sigma, f) \transvia{a \land b \mid p} (\sigma'', f') \text{ in } \langle \pi, e \rangle\JmpRes{} }
    (2)
  \end{mathpar}
  \end{small}

  \noindent In case (1), by definition of \(\res_{\delta}\), we have
  \((b, (\sigma', f'), p) \in \res_{\delta}(\sigma, f)\) which immediately gives us
  \((b, (\sigma', f'), p) \in \accu(\res_{\delta}, \con_{\delta}, (\sigma, f))\).

  \noindent In case (2), the first premise gives us
  \(((\sigma', e_l), \langle a |) \in \con_{\delta}(\sigma, f)\)
  by definition of \(\con_{\delta}\).
  The induction hypothesis applied to the second premise gives us
  \((b, (\sigma'', f'), p) \in \accu(\res_{\delta}, \con_{\delta}, (\sigma', e_l))\).
  The \(\accu\) rule then allows us to derive:
  \begin{small}
  \begin{mathpar}
    \inferrule
    { ((\sigma', e_l), \langle a |) \in \con_{\delta}(\sigma, f) \\ 
      (b, (\sigma'', f'), p) \in \accu(\res_{\delta}, \con_{\delta}, (\sigma', e_l)) }
    { (a \land b, (\sigma'', f'), p) \in \accu(\res_{\delta}, \con_{\delta}, (\sigma, f)) }
  \end{mathpar}
  \end{small}
  which concludes the proof of the forward direction.

  To prove the reverse, we need to show that for all
  \((b_0, (\sigma', f'), p) \in \accu(\res_{\delta}, \con_{\delta}, (\sigma, f))\), there is
  \((\sigma, f) \transvia{b_0 \mid p} (\sigma', f')\) in \(\langle \pi, e \rangle\JmpRes{}\).
  By induction on the derivation of
  \((b_0, (\sigma', f'), p) \in \accu(\res_{\delta}, \con_{\delta}, (\sigma, f))\),
  the cases to consider are:
  \begin{small}
  \begin{mathpar}
    \inferrule
    { (b_0, (\sigma', f'), p) \in \res_{\delta}(\sigma, f) }
    { (b_0, (\sigma', f'), p) \in \accu(\res_{\delta}, \con_{\delta}, (\sigma, f)) }
    (1)
    \and
    \inferrule
    { (n, h) \in \con_{\delta}(\sigma, f) \\ 
      r \in \accu(\res_{\delta}, \con_{\delta}, n) }
    { h(r) \in \accu(\res_{\delta}, \con_{\delta}, (\sigma, f)) }
    (2)
  \end{mathpar}
  \end{small}

  \noindent In case (1), inversion on \((b_0, (\sigma', f'), p) \in \res_{\delta}(\sigma, f)\) gives us:
  \begin{small}
  \begin{mathpar}
    \inferrule
    { (\sigma, f) \transvia{b_0 \mid p} (\sigma', f') \text{ in } \langle \pi, e \rangle }
    { (b_0, (\sigma', f'), p) \in \res_{\delta}(\sigma, f) }
  \end{mathpar}
  \end{small}
  which immediately gives us
  \((\sigma, f) \transvia{b_0 \mid p} (\sigma', f')\) in \(\langle \pi, e \rangle\JmpRes{}\).

  \noindent In case (2), inversion of the first premise gives us:
  \begin{small}
  \begin{mathpar}
    \inferrule  
    { (\sigma, f) \transOut{a}{\jmpc{l, \sigma''}} \text{ in } \langle \pi, e \rangle } 
    { ((\sigma'', e_l), \langle a |) \in \con_{\delta}(\sigma, f) }
  \end{mathpar}
  \end{small}
  where \(n = (\sigma'', e_l)\) and \(h = \langle a |\).
  Let \(r = (b, (\sigma', f'), p) \in \accu(\res_{\delta}, \con_{\delta}, (\sigma', e_l))\).
  The induction hypothesis applied to the second premise gives us
  \((\sigma'', e_l) \transvia{b \mid p} (\sigma', f')\) in \(\langle \pi, e \rangle\JmpRes{}\).
  We can now derive:
  \begin{small}
  \begin{mathpar}
    \inferrule 
    { (\sigma, f) \transOut{a}{\jmpc{l, \sigma''}} \text{ in } \langle \pi, e \rangle \\ 
      (\sigma'', e_l) \transvia{b \mid p} (\sigma', f') \text{ in } \langle \pi, e \rangle\JmpRes{} }
    { (\sigma, f) \transvia{a \land b \mid p} (\sigma', f') \text{ in } \langle \pi, e \rangle\JmpRes{} }
  \end{mathpar}
  \end{small}
  Because \((b_0, (\sigma', f'), p) = h(r) = (a \land b, (\sigma', f'), p)\), there is
  \((\sigma, f) \transvia{b_0 \mid p} (\sigma', f')\) in \(\langle \pi, e \rangle\JmpRes{}\)
  which concludes the proof.
\end{proof}
\fi

\end{document}